\documentclass[11pt,letter, english]{article}

\usepackage[margin=1in]{geometry}
\usepackage[utf8]{inputenc}
\usepackage{babel}
\usepackage{microtype}
\usepackage{url}

\usepackage{tocloft}
\usepackage{paralist, enumitem}
\setlist[enumerate]{itemsep=0pt}

\usepackage[linesnumbered,ruled,vlined]{algorithm2e}
\DontPrintSemicolon
\SetKwInOut{Input}{Input}
\SetKwInOut{Output}{Output}

\usepackage[modulo]{lineno}
\usepackage[compact]{titlesec}
\renewcommand{\paragraph}[1]{\medskip\noindent\textbf{#1}}

\usepackage{amsthm,amssymb, amsfonts, verbatim, nicefrac, hyphenat}
\usepackage{mathtools, etoolbox, thmtools, thm-restate, xparse}
\usepackage{dsfont}
\usepackage{bm}
\usepackage{mathrsfs}
\usepackage[mathscr]{euscript}

\usepackage{graphicx}
\usepackage{epsfig}
\usepackage{subcaption}
\usepackage{array}
\usepackage{multirow}
\usepackage{colortbl,makecell}

\usepackage[dvipsnames]{xcolor}
\definecolor{DarkRed}{rgb}{0.5,0.1,0.1}
\definecolor{DarkBlue}{rgb}{0.1,0.1,0.5}
\definecolor{ForestGreen}{rgb}{0.1333,0.5451,0.1333}
\definecolor{Red}{rgb}{0.9,0,0}

\usepackage{pifont}
\usepackage{soul}
\usepackage[normalem]{ulem}

\usepackage[
    style=alphabetic,
    backref=true,
    maxnames=99,
    maxalphanames=99,
    url=false,
    isbn=false,
    doi=true
]{biblatex}

\AtEveryBibitem{\clearname{editor}}
\AtEveryBibitem{\clearlist{publisher}}
\AtEveryBibitem{\clearfield{pages}}
\usepackage{nameref}
\usepackage[
    linktocpage=true,
    colorlinks,
    linkcolor=DarkRed,
    citecolor=ForestGreen,
    bookmarks,
    bookmarksopen,
    bookmarksnumbered
]{hyperref}

\usepackage[noabbrev,nameinlink,capitalize]{cleveref}

\usepackage{tikz}
\usetikzlibrary{arrows}
\usetikzlibrary{arrows.meta}
\usetikzlibrary{shapes}
\usetikzlibrary{backgrounds}
\usetikzlibrary{positioning}
\usetikzlibrary{decorations.markings}
\usetikzlibrary{patterns}
\usetikzlibrary{calc}
\usetikzlibrary{fit}

\tikzset{
vertex/.style={
    circle, black, fill=Yellow, line width=1pt,
    draw, minimum width=8pt, minimum height=8pt, inner sep=0pt
}}

\usepackage{tcolorbox}
\tcbuselibrary{skins,breakable}
\tcbset{enhanced jigsaw}

\usepackage[framemethod=TikZ]{mdframed}

\DeclareFontFamily{U}{mathx}{\hyphenchar\font45}
\DeclareFontShape{U}{mathx}{m}{n}{
      <5> <6> <7> <8> <9> <10>
      <10.95> <12> <14.4> <17.28> <20.74> <24.88>
      mathx10
}{}
\DeclareSymbolFont{mathx}{U}{mathx}{m}{n}
\DeclareMathSymbol{\bigtimes}{1}{mathx}{"91}

\newtheorem{theorem}{Theorem}
\newtheorem{lemma}{Lemma}[section]
\newtheorem{proposition}[lemma]{Proposition}

\newtheorem{claim}[lemma]{Claim}

\newtheorem{definition}[lemma]{Definition}

\newtheorem{openproblem}{Open Problem}

\newtheorem*{claim*}{Claim}
\newtheorem*{proposition*}{Proposition}
\newtheorem*{lemma*}{Lemma}
\newtheorem*{problem*}{Problem}

\newtheorem{remark}[lemma]{Remark}
\newtheorem{observation}[lemma]{Observation}
\newtheorem{ourProp}{Property}

\theoremstyle{definition}

\crefname{lemma}{Lemma}{Lemmas}
\crefname{claim}{Claim}{Claims}
\crefname{property}{property}{Property}
\creflabelformat{property}{(#1)#2#3}
\crefname{equation}{Eq}{Eq}
\creflabelformat{equation}{(#1)#2#3}
\crefname{enumi}{Step}{Steps}
\crefname{step}{Step}{Step}
\crefname{claim}{claim}{claims}

\newenvironment{abox}{
\begin{tcolorbox}[
    enlarge top by=5pt,
    enlarge bottom by=5pt,
    breakable,
    frame hidden,
    overlay broken = {
        \draw[line width=1pt, black]
        (frame.north west) rectangle (frame.south east);},
    overlay = {
        \draw[line width=1pt, black]
        (frame.north west) rectangle (frame.south east);},
    boxsep=0pt,
    left=4pt,
    right=4pt,
    top=10pt,
    arc=0pt,
    boxrule=1pt,toprule=1pt,
    colback=white
]}{\end{tcolorbox}}

\newenvironment{prop}{\begin{abox}\begin{ourProp}}{\end{ourProp}\end{abox}}

\DeclareMathOperator{\dom}{dom}
\DeclareMathOperator{\var}{var}
\DeclareMathOperator{\E}{\mathbb{E}}

\DeclareMathOperator{\vbl}{vbl}
\DeclareMathOperator{\Temp}{Temp}

\DeclarePairedDelimiter{\bracket}[]
\DeclarePairedDelimiter{\paren}()

\DeclarePairedDelimiter{\abs}{|}{|}

\DeclarePairedDelimiter{\set}{\{}{\}}

\DeclarePairedDelimiterXPP{\Ot}[1]{\widetilde{O}}(){}{#1}
\DeclarePairedDelimiterXPP{\Omgt}[1]{\widetilde{\Omega}}(){}{#1}
\DeclarePairedDelimiterXPP{\BigO}[1]{O}(){}{#1}
\DeclarePairedDelimiterXPP{\I}[1]{\boldsymbol{1}}(){}{#1}

\newcommand{\LOCAL}{\ensuremath{\mathsf{LOCAL}}\xspace}

\newcommand{\RCT}{\textsc{Rct}\xspace}
\newcommand{\MCT}{\textsc{Mct}\xspace}

\newcommand{\AL}{A_L}

\newcommand{\eps}{\varepsilon}
\newcommand{\poly}{\operatorname{\text{{\rm poly}}}}

\renewcommand{\Pr}{\mathbb{P}}
\NewDocumentCommand{\Prob}{sO{}E{_}{{}}m}{{\Pr}_{#3}
\IfBooleanTF{#1}{\bracket*{#4}}{\bracket[#2]{#4}}}

\newcommand{\Exp}{\mathbb{E}}

\newcommand{\dist}{\operatorname{dist}}

\newcommand{\lovasz}{Lov\'{a}sz\xspace}

\newcommand\given{ \nonscript\: \vert{} \nonscript\: \mathopen{} }

\renewcommand{\phi}{\varphi}

\newcommand{\alg}[1]{\textsf{#1}}

\renewcommand{\wp}{$\text{w.p.}$\xspace}
\newcommand{\whp}{$\text{w.h.p.}$\xspace}

\newcommand{\unhappy}{\mathrm{Unhappy}}
\newcommand{\swap}{\mathrm{Swappable}}
\newcommand{\All}{\mathrm{All}}
\renewcommand{\Big}{\operatorname{Big}}

\newcommand{\fT}{\mathscr{T}}
\newcommand{\fS}{\mathscr{S}}

\newcommand{\pious}{\text{$\Pi$-ous}\xspace}
\newcommand{\Gev}{G_{\operatorname{dep}}}

\newcommand{\cstB}{c_{\mathcal{B}'}}

\renewcommand{\leq}{\leqslant}
\renewcommand{\geq}{\geqslant}
\renewcommand{\le}{\leqslant}
\renewcommand{\ge}{\geqslant}

\renewcommand{\qed}{
\nobreak \ifvmode \relax \else
\ifdim\lastskip<1.5em \hskip-\lastskip
\hskip1.5em plus0em minus0.5em \fi \nobreak
\vrule height0.75em width0.5em depth0.25em
\fi
}

\AtBeginDocument{}

\newcommand{\trycolor}{\textsc{TryColor}\xspace}

\title{Sublogarithmic Distributed Vertex Coloring \\ with Optimal Number of Colors}
\author{
Maxime Flin \thanks{This work was supported in part by the Research Council of Finland, Grants 359104 and 363558. Part of this work was done while the author was working at Reykjavik University, funded by the Icelandic Research Fund, Grant 2310015-053.}\\
{\small Aalto University} \\ {\small maxime.flin@aalto.fi}
\and
Magn\'us M.\ Halld\'orsson \thanks{Supported by the Icelandic Research Fund, Grant 2511609.} \\
{\small Reykjavik University} \\ {\small mmh@ru.is}
\and 
Manuel Jakob \thanks{This research was funded in whole or in part by the Austrian Science Fund (FWF) \url{https://doi.org/10.55776/P36280}, \url{https://doi.org/10.55776/I6915}. For open access purposes, the author has applied a CC BY public copyright license to any author-accepted manuscript version arising from this submission.}\\
{\small TU Graz} \\ {\small m.jakob@tugraz.at}
\and 
Yannic Maus \footnotemark[3]\\
{\small TU Graz} \\ {\small yannic.maus@tugraz.at}
}
\date{}

\newcommand{\covd}{\text{CC}}

\renewcommand\qedsymbol{$\blacksquare$}
\let\oldnl\nl \newcommand{\nonl}{\renewcommand{\nl}{\let\nl\oldnl}}

\begin{document}
\thispagestyle{empty}
\maketitle

\begin{abstract}
    For any $\Delta$, let $k_\Delta$ be the maximum integer $k$ such that
    $(k+1)(k+2)\le \Delta$. 
    We give a distributed \LOCAL algorithm that, given an integer $k < k_\Delta$, computes a valid $\Delta-k$-coloring if one exists. The algorithm runs in $\Ot{ \log^4 \log n }$ rounds, which is within a polynomial factor of the $\Omega(\log\log n)$ lower bound, which already applies to the case $k=0$.
    It is also best possible in the sense that 
    if $k \ge k_\Delta$, the problem requires $\Omega(n/\Delta)$ distributed rounds [Molloy, Reed, '14, Bamas, Esperet '19].
    
     For $\Delta$ at most polylogarithmic, the algorithm is an exponential improvement over the current state of the art of $O(\log^{49/12} n)$ rounds. When $\Delta \ge (\log n)^{50}$, our algorithm achieves an even faster runtime of~$O(\log^* n)$ rounds.
\end{abstract}

\pagenumbering{roman}
\thispagestyle{empty}
\clearpage
\thispagestyle{empty}
\renewcommand{\cftbeforesecskip}{8pt}
\renewcommand{\cftbeforesubsecskip}{4pt}
\tableofcontents

\clearpage

\pagenumbering{arabic}
\setcounter{page}{1}

\section{Introduction}
Graph coloring is a problem of fundamental importance to combinatorics and computer science. Given a graph $G=(V,E)$ and an integer $c \geq 1$, a $c$-coloring of $G$ is a mapping from the vertices of $G$ to $\set{1, 2, \ldots, c}$ such that adjacent vertices of $G$ receive different colors. Determining the smallest $c$ for which a given graph admits a $c$-coloring -- a value known as the chromatic number denoted by $\chi(G)$ -- has been an enduring challenge in graph theory since its very beginning.
Computationally, it has long been known that computing exactly or even approximately the chromatic number of a graph is NP-hard \cite{Karp1972,Zuckerman2007}. In general, the chromatic number of a graph depends on the global structure of the graph rather than on purely local properties (see \cite[Chapter 3]{AlonSpencerBook}).

\paragraph{Distributed Graph Coloring.} In the \LOCAL\ model, 
introduced by Linial \cite{linial92}, the input graph $G$ is seen as a network in which $n=|V|$ nodes equipped with unique identifiers communicate with their neighbors in synchronous rounds and perform unlimited local computation, aiming to minimize the total number of rounds until each node has computed its output, e.g., its color in a graph coloring problem. 
In this setting, one classically aims to compute a $\Delta+1$-coloring where $\Delta$ is the maximum degree of the graph \cite{barenboimelkin_book}. Sequentially, the problem is trivial as one can greedily assign colors to nodes without ever getting stuck, or in other words, any valid partial coloring can always be completed to a coloring of the whole graph. Essentially all efficient distributed graph coloring algorithms exploit this greedy behavior of the problem, e.g., \cite{BEPSv3,CLP20,MT20,GK20,HKNT22,GG24}. The most important exceptions are several papers that aim for an efficient distributed implementation of Brooks' theorem \cite{brooks1941} for coloring with exactly $\Delta$ colors, e.g., \cite{PS95,GHKM18,FHM23,HM24,BBN26,BBN25}. 

\paragraph{Coloring with Fewer Colors.}
Coloring graphs whose chromatic number is smaller than~$\Delta$ using the optimal number of colors is challenging, even in the centralized setting. In a seminal paper, Molloy and Reed \cite{MR14} characterized the precise threshold at which the problem becomes NP-complete. Let $k_{\Delta}$ be the largest integer such that $(k_{\Delta}+1)(k_{\Delta}+2)\le \Delta$.\footnote{It can be verified that $k_\Delta=\lfloor
\sqrt{\Delta+1/4}-3/2\rfloor$ and thus 
$\sqrt{\Delta}-3<k_\Delta < \sqrt{\Delta}-1$ holds.} 
 They showed that for $c<\Delta-k_{\Delta}$ the decision problem whether a graph admits a $c$-coloring is NP-complete, while for $c\geq \Delta-k_{\Delta}$ the decision is based on whether the induced neighborhood $G[N[v]]$ of some node $v$ is $c$-colorable. Furthermore, for $c\geq \Delta-k_{\Delta}$ the coloring can be computed by a (centralized) polynomial-time algorithm when $\Delta$ is a constant. This result is highly non-trivial, spanning more than 60 dense pages \cite{MR14}.
 
 In the distributed setting, Bamas and Esperet showed that for $c=\Delta-k_{\Delta}$ computing a $c$-coloring of a graph with chromatic number $c$ requires $\Omega(n/\Delta)$ rounds \cite{BamasE19}. Based on the work of Molloy and Reed, they also gave a randomized algorithm for the \LOCAL model to $c$-color graphs when $c> \Delta-k_{\Delta}$ in a polylogarithmic number of rounds. In contrast, the state-of-the-art algorithms for $(\Delta+1)$-coloring and, recently also for $\Delta$-coloring, only require $\poly(\log \log n)$ randomized distributed time \cite{CLP20,RG20,GHKM18,FHM23,BBN25}. This discrepancy motivates the following question.  

\begin{quote}
    \emph{Can one color in a sublogarithmic number of rounds with fewer than $\Delta$ colors?}
\end{quote}

\paragraph{Our Results.}
In this paper, we answer the question in the affirmative. We show that one can compute a $c$-coloring in sublogarithmic time when one exists and $c\geq \Delta-k_{\Delta}+1$. Moreover, when $\Delta$ is asymptotically larger than some polylogarithm in $n$, our algorithm runs in $O(\log^*n)$ rounds, matching the state of the art for the significantly easier $(\Delta+1)$-coloring problem.  More precisely, we prove the following theorem:

\begin{restatable}{theorem}{thmmain}
\label{thm:col} \label{thm:main}
    For sufficiently large $\Delta$, and any $c \ge \Delta-k_{\Delta}+1$, there is a distributed randomized algorithm that
    takes a graph $G$ with maximum degree $\Delta$ as input, and does the
    following: either some vertex outputs a certificate
    that $G$ is not $c$-colorable, or the algorithm finds a $c$-coloring
    of $G$. The algorithm runs in $O(\log^* n)$ rounds when $\Delta \ge (\log n)^{50}$, and 
    in general in $\Ot{ \log^4\log n }$ rounds\footnote{In this paper, $\Ot{ f(n) }$ hides multiplicative factors of size $\poly \log f(n)$.}, with high probability.
\end{restatable}

\cref{thm:col} comes polynomially close to the lower bound of $\Omega(\log_{\Delta} \log n)$ rounds established for the $\Delta$-coloring problem \cite{LLL_lowerbound}. Moreover, this result uses the fewest number of colors possible due to the aforementioned lower bound from \cite{BamasE19}. 
While our approach builds on the framework of Molloy and Reed, it is conceptually and technically simpler in several key aspects; see \Cref{ssec:contributions,sec:tech-intro} for details.

By plugging the randomized algorithm from \Cref{thm:main} into the powerful distributed derandomization framework of \cite{GKM17,GHK18,RG20,GG24} we obtain more than a quadratic improvement for deterministic algorithms. Previously the best algorithm, using state-of-the-art subroutines \cite{GG24,RG20}, had complexity $O(\log^{49/12}n)$ \cite{BamasE19}.
\begin{restatable}{theorem}{TheoremDet}
    \label{thm:det-alg}
    For sufficiently large $\Delta$, and any $c \ge \Delta-k_{\Delta}+1$, there is a distributed deterministic algorithm that
    takes a graph $G$ with maximum degree $\Delta$ as input, and does the
    following: either some vertex outputs a certificate
    that $G$ is not $c$-colorable, or the algorithm finds a $c$-coloring
    of $G$. The algorithm runs in $\Ot{ \log^2 n }$ rounds. 
\end{restatable}
The runtime of \Cref{thm:det-alg} is only roughly quadratically slower than the lower bound of $\Omega(\log_\Delta n)$ rounds established for the $\Delta$-coloring problem \cite{Chang2016a} and almost matches the natural barrier of $\Omega(\log^2n)$ rounds for the deterministic complexity of the problem. Surpassing it is believed to require fundamentally new techniques.
Molloy and Reed's, Bamas and Esperet's, and also our algorithm for computing a coloring with fewer than $\Delta$ colors rely in multiple places on solving instances of the constructive \lovasz Local Lemma (LLL). The only known method to do this deterministically in a distributed setting is via the aforementioned derandomization framework that inherently comes with an $\Omega(\log^2 n)$ cost; see \cite{GG24} for more details. Breaking this barrier for general LLLs is considered one of the biggest open problems in the field. Chang and Pettie conjectured in \cite{Changhierarchy19} that this can be done on bounded-degree graphs, but this question remains unanswered.

\subsection{Further Background on Distributed Graph Coloring}

\paragraph{Historic focus on Greedy problems.} 
Coloring was the central topic of the paper establishing the \LOCAL model \cite{linial92}. However, since this model requires that nodes decide on their colors only based on their local view, even solving greedy problems such as maximal independent set, maximal matching or $(\Delta+1)$-vertex-coloring is highly non-trivial. Linial showed that even coloring an $n$-cycle with a constant number of colors requires at least $\Omega(\log^*n)$ rounds for deterministic algorithms and Naor later extended this lower bound to the randomized setting \cite{Naor91}. It took roughly 30 years to understand that on general graphs those greedy problems could be solved deterministically in polylogarithmic distributed time thanks to a breakthrough by Rohzon and Ghaffari \cite{RG20}. Over the last 5 years, follow-up work \cite{GK20,BBHORS21,FFGKR23,GG24} on deterministic algorithms coupled with earlier randomized techniques \cite{BEPSv3,ghaffari16_MIS,CLP20,HKNT22}  led to exciting progress on our understanding of the complexity of greedy problems in the \LOCAL model.

\paragraph{Recent interest in the locality of non-greedy problems.\footnote{Intuitively, a non-greedy problem is one where not every partial solution can be extended to a full solution. One distinguishing feature is that greedy problems can be solved in $O(\log^* n)$ rounds on constant-degree graphs, while non-greedy problems require at least $\Omega(\log \log n)$ rounds. See \cite{JMS25} for a more formal discussion on the topic. }}
Despite rising interest in the locality of non-greedy problems, e.g., for the seminal works on (hypergraph) sinkless orientation problems  \cite{LLL_lowerbound,GS17,BMNSU25}, for various degree splitting problems \cite{Splitting17,HMN22,Davies23}, and for the more amenable non-greedy edge coloring problems, e.g.,  \cite{GKMU18,SuVu19,HN21,HMN22,CHLPU20,Davies23,JMS25}, little is known for non-greedy vertex coloring. The exception is the aforementioned coloring with exactly $\Delta$ colors. A well-known theorem of Brooks \cite{brooks1941} shows that any connected graph admits a $\Delta$-coloring unless it is a $(\Delta+1)$-clique or an odd cycle. There have been several works designing faster and faster algorithms for the problem \cite{PS95,GHKM18,FHM23,HM24,BBN25} culminating in the very recent result of \cite{BBN26} obtaining an optimal round complexity when $\Delta=O(1)$.
The only known works attempting to color with fewer than $\Delta$ colors are for very sparse graphs, i.e., triangle-free graphs \cite{CPS17}, but their randomized complexities remain at least logarithmic.  In a different direction, Barenboim designed a randomized distributed algorithm that computes an $O(n^{1/2+\eps}\chi(G))$-coloring of any graph $G$ \cite{B12}. 

\paragraph{The importance of the \lovasz Local Lemma.} 
The \lovasz Local Lemma (LLL) provides a powerful probabilistic framework for showing that, even when many local bad events may occur, there exists a global configuration that avoids all of them simultaneously. Informally, it can be viewed as a localized analogue of the union bound: if each bad event occurs with sufficiently small probability and depends on only a limited number of other events, then one can still guarantee the existence of an outcome in which no bad event occurs. In the \emph{constructive} version of the LLL, the goal is not only to prove existence but also to efficiently compute such an assignment.

Molloy and Reed’s algorithm for coloring with $c \ge \Delta - k_{\Delta}$ heavily relies on multiple instances of the constructive LLL (see Section~11 in~\cite{MR14}). In essence, the LLL is used to ensure that partial colorings retain properties similar to those of a random coloring, thereby guaranteeing that they can always be extended without creating deadlocks. In their distributed adaptation, Bamas and Esperet significantly reduced the number of required LLL instances and solved each of them using a logarithmic-time distributed LLL solver based on the Moser--Tardos framework~\cite{BamasE19,MoserTardos10,CPS17}. To date, no non--LLL-based approach is known for computing such colorings.

In the distributed setting, random variables and bad events are naturally associated with nodes of the communication network, and the objective mirrors that of the constructive LLL: to compute a global assignment that avoids all bad events. Consequently, understanding the distributed complexity of the LLL has become a central challenge in the field. In a seminal result, Chang and Pettie~\cite{Changhierarchy19} showed that the LLL is complete for sublogarithmic-time computation of local graph problems on constant-degree graphs. As a consequence, any such problem with sublogarithmic complexity can be solved in $\poly(\log\log n)$ rounds using the distributed LLL algorithm of Fischer and Ghaffari~\cite{FG17}. 

For general graphs, however, it remains a major open question which problems admit sublogarithmic-time distributed algorithms. The complexity of the general distributed LLL on general graphs is also widely open, despite recent progress on special cases~\cite{GHK18,Davies23} and results showing sublogarithmic average-time complexity per node~\cite{D25}. Moreover, the constructive LLL admits an $\Omega(\log\log n)$ randomized lower bound~\cite{LLL_lowerbound}. Our work not only yields sublogarithmic-time algorithms for coloring with fewer than $\Delta$ colors, but also provides sublogarithmic-time solutions for several LLL instances that, prior to our work, were only known to admit logarithmic-time distributed solutions.

\subsection{The Challenges and Our Technical Contributions}
\label{ssec:contributions}
Our approach follows the same high-level structure as that of Molloy and Reed~\cite{MR14}. As in their work, we first transform the input graph to expose its structural regularities and then decompose it into five components: a sparse region, two dense regions consisting of cliques, and two intermediate regions linking them. These components are colored in a carefully chosen order while maintaining probabilistic invariants that guarantee that the dense cliques can be completed in the final stages. We refer to \Cref{sec:tech-intro} for a more detailed overview of this framework.

\paragraph{Challenges in achieving sublogarithmic distributed algorithms.}
While the Molloy--Reed framework provides a conceptually clean approach, several aspects pose major challenges in the distributed setting:
\begin{enumerate}
    \item The method of~\cite{MR14} for coloring the sparse and intermediate regions relies on an iterative semi-random process. This procedure is technically involved—requiring a lengthy analysis—and inherently slow, as it performs $\Delta^{\Theta(1)}$ iterations, each requiring the solution of an LLL instance. Although Bamas and Esperet reduced this to $O(\log^{13/12}\Delta)$ iterations, this complexity remains far too large for achieving sublogarithmic distributed time.
    \item The probabilistic arguments throughout~\cite{MR14} are based on the \lovasz Local Lemma. Known distributed LLL solvers require $\Omega(\log n)$ rounds in general~\cite{BamasE19,MoserTardos10,CPS17}. While some LLL instances can be solved faster when the criteria are sufficiently relaxed~\cite{HMN22,Davies23}, others -- most notably those arising from the dense regions -- satisfy only polynomial LLL criteria and thus resist existing fast techniques.
\end{enumerate}
As a consequence, a direct distributed implementation of the Molloy--Reed framework inevitably leads to high round complexity. Notably, this difficulty is most pronounced for high-degree graphs, despite the fact that such graphs admit ultrafast distributed algorithms for simpler problems such as $\Delta$-coloring~\cite{FHM23}.

\paragraph{Our Contributions.}
We present a sublogarithmic-time distributed algorithm for coloring with $\Delta - k_\Delta + 1$ colors by refining and simplifying the Molloy--Reed framework.

\begin{enumerate}
    \item \textbf{A central slack invariant.}
    Our main conceptual contribution is the introduction of a central invariant based on \emph{slack}, ensuring that each vertex consistently has more available colors than competing constraints. The key feature of this invariant is its monotonicity: as the algorithm progresses, the loss of available colors is matched by a corresponding reduction in uncolored neighbors, so that slack does not deteriorate.

    This viewpoint allows us to avoid tracking the detailed evolution of color lists and degrees over time. Rather than carefully controlling distributions at each step, we rely on the invariant to ensure that partial colorings remain extendable, leading to a simpler and more robust approach to coloring the sparse and intermediate regions.

    \item \textbf{Faster handling of structured LLL instances.}
    We develop sublogarithmic-time solutions for the Lov\'asz Local Lemma instances arising in our setting. Our approach builds on the shattering framework, augmented with additional mechanisms---such as guard events---that ensure the residual instances remain well-behaved and avoid cascading dependencies.

    While the individual ingredients have appeared in prior work, their combination allows us to handle LLL instances that were not previously known to admit such fast distributed solutions.

    \item \textbf{Parallel coloring of dense cliques.}
    For dense regions, we introduce a subsampling technique that enables cliques to be colored in parallel by reducing the problem to independent matching tasks within each clique. This avoids the need for global coordination and simplifies the treatment of dense structures compared to prior approaches.

    \item \textbf{Ultrafast algorithms in the high-degree regime.}
    In graphs of sufficiently large degree, we can replace the core iterative coloring step with a reduction to a form of list coloring that admits very fast distributed algorithms. As all our probabilistic claims hold with subexponential error bounds, we also bypass the application of LLL altogether.
    This leads to a $O(\log^* n)$-time solution when $\Delta \ge \log^{50} n$.
\end{enumerate}

Conceptually, our results expand the small but growing class of explicit distributed LLL formulations that can be solved in sublogarithmic time on general graphs. Despite this progress, designing a general distributed LLL algorithm that efficiently handles all relevant instances across all degree regimes remains a major open problem in the field.

\subsection{Organization of the Paper}
In \cref{sec:tech-intro}, we present the main technical ideas underlying our approach, while \cref{sec:prelim} contains complementary background material. The algorithm is described formally in \cref{sec:high-level}, where we also prove \cref{thm:col,thm:det-alg}. The subroutines and their analyses for coloring $\Pi$-ous subgraphs, sparse vertices, and dense vertices are presented in \cref{sec:colorwithmuchslack,sec:sparse,sec:cliques}, respectively. The high-degree case is treated separately in \cref{sec:hideg}.

The substantial technical detail throughout the paper is driven by the goal of achieving sublogarithmic distributed runtime. If one were only concerned with the existence of such colorings, a polynomial-time centralized algorithm, or a polylogarithmic-time distributed solution, many sections could be significantly simplified.

\section{Technical Overview}
\label{sec:tech-intro}
We give an overview of the main technical ideas behind \cref{thm:col}, focusing on how to construct a $(\Delta - k_\Delta + 1)$-coloring in sublogarithmic distributed time.
In \cref{sec:tech-intro-LLL}, we attempt to give a bird's eye view of the techniques allowing us to solve each instance of the \lovasz Local Lemma in sublogarithmic time.

\subsection{\texorpdfstring{Color Coverage (CC) \& $\Pi$-ous subgraphs}{Color Coverages (CC) \& Pious subgraphs} }
Due to the impossibility of completing all partial solutions in a non-greedy coloring problem, the hardest part of the problem is coloring the  \emph{last} vertices. In a non-greedy coloring problem, one has to construct a coloring while ensuring that the remaining nodes can be colored later. Following Molloy and Reed, we color the dense vertices last and begin by reviewing the conditions necessary for coloring the cliques at the end. 

\paragraph{Color Swapping in Cliques.}
Consider a $(\Delta - k_{\Delta} + 1)$-clique in which every vertex is adjacent to $k_{\Delta}$ nodes outside the clique and all but one vertex have been properly colored. To color this remaining vertex~$v$, which we call \emph{unhappy}, Molloy and Reed swap its color with that of some other vertex in the clique. In order to do so without creating new conflicts, the vertex $u$ with which $v$ swaps its color must have no external neighbors with $v$'s color (and vice versa). See \cref{fig:swap}. 
However, it may happen that the coloring outside the clique does not permit any such swap (as in \cref{fig:empty}). 

\begin{figure}[ht!]
    \centering
    \begin{subfigure}{.45\linewidth}
    \begin{center}
        \scalebox{0.75}{    
            \begin{tikzpicture}[scale=1]
                \def\radius{2}
		          \def\outlen{2} \def\numnodes{6}
		
                \foreach \i in {1,...,\numnodes}{
			        \coordinate (P\i) at ({360/\numnodes * (\i - 1)}:\radius);
		          }
		          \draw[gray, dashed, very thin, fill=gray!10] (0,0) circle (2.5);
		
                \foreach \i in {1,...,\numnodes}{
                	\foreach \j in {1,...,\numnodes}{
                		\ifnum\i<\j
                	       \draw[gray!50, very thin] (P\i) -- (P\j);
                		\fi
                	}
                }
                
                \pgfmathsetmacro{\angle}{360/\numnodes * (2 - 1)}
		          \node[circle, draw=white, fill=white, inner sep=5pt] at ($(P2)+({\angle-36}:\outlen)$) {};

                \pgfmathsetmacro{\angle}{360/\numnodes * (3 - 1)}
		          \node[circle, draw=white, fill=white, inner sep=5pt] at ($(P3)+({\angle-12}:\outlen)$) {};
                \foreach \i in {1,...,\numnodes}{
			        \pgfmathsetmacro{\angle}{360/\numnodes * (\i - 1)}
			        \foreach \shift in {-36, -12, 12, 36}{
				        \draw[black, very thin] (P\i) -- ++({\angle+\shift}:\outlen);
			        }
		          }
		
                \node[circle, draw=black, fill=green!60!black, inner sep=5pt, line width=0.3pt] at (P1) {};
		          \node[circle, draw=black, fill=yellow, inner sep=5pt, line width=0.3pt] at (P2) {};
		          \node[circle, draw=black, fill=cyan!70!white, inner sep=5pt, line width=0.3pt] at (P3) {};
        		\node[circle, draw=black, fill=purple!70!blue, inner sep=5pt, line width=0.3pt] at (P4) {};
        		\node[circle, draw=black, fill=blue, inner sep=5pt, line width=0.3pt] at (P5) {};
        		\node[circle, draw=black, fill=red, inner sep=5pt, line width=0.3pt] at (P6) {};
        
                \draw[red, very thick, dotted] ($(P2)!0.5!(P3)$) ellipse [x radius=1.5cm, y radius=0.8cm];
        
                \pgfmathsetmacro{\angle}{360/\numnodes * (1 - 1)}
                		\node[circle, draw=black, fill=blue, inner sep=5pt] at ($(P1)+({\angle-12}:\outlen)$) {};

                \pgfmathsetmacro{\angle}{360/\numnodes * (4 - 1)}
                		\node[circle, draw=black, fill=blue, inner sep=5pt] at ($(P4)+({\angle+12}:\outlen)$) {};
                
                \pgfmathsetmacro{\angle}{360/\numnodes * (5 - 1)}
                		\node[circle, draw=black, fill=blue, inner sep=5pt] at ($(P5)+({\angle+36}:\outlen)$) {};
                
                \pgfmathsetmacro{\angle}{360/\numnodes * (5 - 1)}
                		\node[circle, draw=black, fill=red, inner sep=5pt] at ($(P5)+({\angle-36}:\outlen)$) {};

                        \node at (-2.9,1.5) {\huge$A_i$};
                        \node at (3.2,1.5) {\large \texttt{Swappable}};
                        \node at (0,-2.8) {\large \texttt{unhappy}};
                
                        \draw[red, very thick, dotted] (P5) circle [radius=0.45cm];
                
                        \draw[red, thick, -{Stealth[length=2mm, width=2mm]}] (2.2,1.5) -- (1.5,1.7);
                        \draw[red, thick, -{Stealth[length=2mm, width=2mm]}] (-0.1,-2.65) -- (-0.5,-2.1);
	        \end{tikzpicture}
        }        
    \end{center}
    \caption{\label{fig:swap}}
    \end{subfigure}
    \hfill
    \begin{subfigure}{.45\linewidth}
    \begin{center}
        \scalebox{0.75}{    \begin{tikzpicture}[scale=1]
		
        \def\radius{2}
		\def\outlen{2} \def\numnodes{6}
		
        \foreach \i in {1,...,\numnodes}{
			\coordinate (P\i) at ({360/\numnodes * (\i - 1)}:\radius);
		}
		
\draw[gray, dashed, very thin, fill=gray!10] (0,0) circle (2.5);
		
\foreach \i in {1,...,\numnodes}{
			\foreach \j in {1,...,\numnodes}{
				\ifnum\i<\j
				\draw[gray!50, very thin] (P\i) -- (P\j);
				\fi
			}
		}
		
\foreach \i in {1,...,\numnodes}{
			\pgfmathsetmacro{\angle}{360/\numnodes * (\i - 1)}
			\foreach \shift in {-36, -12, 12, 36}{
				\draw[black, very thin] (P\i) -- ++({\angle+\shift}:\outlen);
			}
		}
		
\node[circle, draw=black, fill=green!60!black, inner sep=5pt, line width=0.3pt] at (P1) {};
		\node[circle, draw=black, fill=yellow, inner sep=5pt, line width=0.3pt] at (P2) {};
		\node[circle, draw=black, fill=cyan!70!white, inner sep=5pt, line width=0.3pt] at (P3) {};
		\node[circle, draw=black, fill=purple!70!blue, inner sep=5pt, line width=0.3pt] at (P4) {};
		\node[circle, draw=black, fill=blue, inner sep=5pt, line width=0.3pt] at (P5) {};
		\node[circle, draw=black, fill=red, inner sep=5pt, line width=0.3pt] at (P6) {};
        
\pgfmathsetmacro{\angle}{360/\numnodes * (1 - 1)}
		\node[circle, draw=black, fill=blue, inner sep=5pt] at ($(P1)+({\angle-12}:\outlen)$) {};

\pgfmathsetmacro{\angle}{360/\numnodes * (4 - 1)}
		\node[circle, draw=black, fill=blue, inner sep=5pt] at ($(P4)+({\angle+12}:\outlen)$) {};

\pgfmathsetmacro{\angle}{360/\numnodes * (5 - 1)}
		\node[circle, draw=black, fill=blue, inner sep=5pt] at ($(P5)+({\angle+36}:\outlen)$) {};

\pgfmathsetmacro{\angle}{360/\numnodes * (5 - 1)}
		\node[circle, draw=black, fill=red, inner sep=5pt] at ($(P5)+({\angle-36}:\outlen)$) {};

\pgfmathsetmacro{\angle}{360/\numnodes * (2 - 1)}
		\node[circle, draw=black, fill=blue, inner sep=5pt] at ($(P2)+({\angle-36}:\outlen)$) {};

\pgfmathsetmacro{\angle}{360/\numnodes * (3 - 1)}
		\node[circle, draw=black, fill=blue, inner sep=5pt] at ($(P3)+({\angle-12}:\outlen)$) {};

        \node at (-2.9,1.5) {\huge$A_i$};
        \node at (0,-2.8) {\large \texttt{unhappy}};

        \draw[red, very thick, dotted] (P5) circle [radius=0.45cm];

        \draw[red, thick, -{Stealth[length=2mm, width=2mm]}] (-0.1,-2.65) -- (-0.5,-2.1);
	\end{tikzpicture}

 }        
    \end{center}
    \caption{\label{fig:empty}}
    \end{subfigure}
    \caption{A clique $A_i$ in a graph with $\Delta = 9$ and $k = 3$. On \cref{fig:swap}, the clique $A_i$ contains one unhappy vertex, in blue. It cannot swap its color with the purple or green vertex because they have a blue external neighbor. It cannot swap its color with the red vertex because it has a red external neighbor. The remaining vertices (in cyan and yellow) make the $\swap$ set for this unhappy vertex. On \cref{fig:empty}, the $\swap$ set is empty because too many vertices on the outside adopted the blue color.}
    \label{fig:clique}
\end{figure}
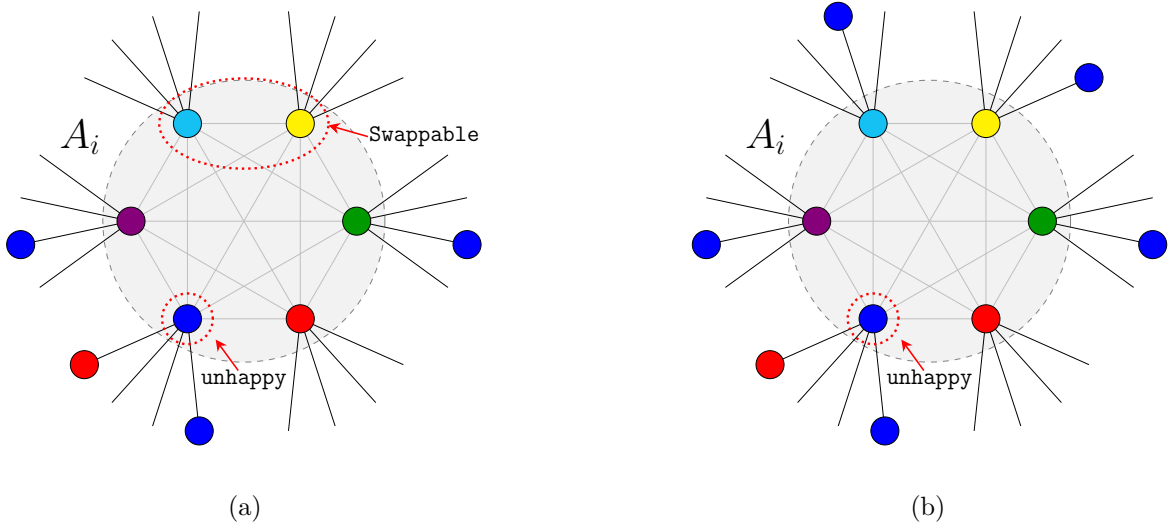

To circumvent this issue, Molloy and Reed constrain all earlier steps of the coloring algorithm so that the coloring outside each clique behaves sufficiently randomly.
Concretely, this means that for every clique and every color, only a small number of vertices are forbidden from using that color.
We formalize this requirement as the \emph{color coverage (henceforth CC) property}. Ensuring that CC holds until coloring the dense parts of the graph is essential to completing the coloring. See \Cref{ssec:CC} for the precise definition of CC. 

\paragraph{\texorpdfstring{$\Pi$-ous subgraphs}{Pious subgraphs}.}
The CC condition can be maintained if all earlier coloring steps are sufficiently random.
If each dense node has at most $U$ neighbors outside its clique, and each of them picks one random color out of $U\cdot \Delta^{0.22}$ colors, then few of these outside neighbors pick the same color, thereby maintaining the CC condition.
Thus, the key to ensuring the CC condition in each step of the algorithm is to maintain large lists of available colors $L(v)$ (colors that do not currently appear in the neighborhood of $v$) compared to the number of outside neighbors of the cliques. We capture this condition through \Cref{prop:Pi}, which we state in a simplified form as follows:

\renewcommand{\theourProp}{$\Pi$ (formalized in \cref{sec:proof-thm-main})} \begin{prop}
An induced subgraph~$H$ satisfies Property~$\Pi$ if 
\begin{enumerate}[label=(\alph*)]
    \item each uncolored dense vertex has at most~$U$ outside neighbors, and
    \item each vertex~$u \in H$ has a list $L(u)$ of at least $d_H(u) + U\cdot \Delta^{0.22}$ available colors. 
\end{enumerate}
\end{prop} 
Intuitively, Property~$\Pi$ guarantees that every vertex has enough slack compared to the number of neighbors outside its dense structure.

Subgraphs satisfying~\Cref{prop:Pi} are called \emph{$\Pi$-ous} (pronounced \emph{pi}ous). 
We can color a $\Pi$-ous subgraph $H$ while maintaining CC by a combination of iterated random color trials (\RCT) and multi color trials (\MCT). Condition (b) of \Cref{prop:Pi} ensures that the list of available colors for each vertex of $H$ has $U\cdot \Delta^{0.22}$ \emph{slack} that remains and ensures a sufficiently large list regardless of how the coloring of $H$ proceeds. 
Slack does not decrease because coloring a neighbor removes both one competing vertex and one available color, leaving their difference—the slack—unchanged or increased.

Property $\Pi$ is strictly stronger than assumptions of \cite{MR14}. Their assumption (8.2) on the list size involves a max of two terms, not the sum, which means that slack is only guaranteed when degrees are high. 

\subsection{Overview of the Coloring Algorithm}

Our algorithm centers on identifying $\Pi$-ous subgraphs that can be colored while maintaining \cref{def:notbig}, thereby ensuring that the dense vertices can be colored last.
We build on the structural decomposition of Molloy and Reed, which nearly\footnote{The subgraphs induced by $B_H$ and $B_L$ are $\pious$ by construction, however the subgraph induced by $S$ is not $\pious$ right away; see \cref{sec:sparse} for more details.} partitions the sparser regions of the graph into $\Pi$-ous subgraphs \cite{MR14}.
More precisely, they reduced the problem of finding a $c$-coloring of $G$ to the problem of finding a $c$-coloring of a graph $F$ with a very specific structure. Bamas and Esperet also showed that $F$ can be constructed from $G$ in $O(1)$ rounds in the \LOCAL model, and that one can efficiently recover a $c$-coloring of $G$ from a $c$-coloring of $F$ in the distributed setting \cite{BamasE19}. 

\begin{abox}
    \paragraph{Structural Decomposition (formalized in \cref{lem:structuralDecomposition}).}
    The vertices of the graph $F$ are partitioned into five sets as follows:    
    \begin{itemize}
        \item \textbf{Sparse vertices ($S$):} vertices whose neighborhoods are far from being cliques, either having noticeably lower degree or many missing edges among their neighbors,
        \item \textbf{High-slack intermediates ($B_H$):} vertices outside the cliques that have many neighbors in cliques,
        \item \textbf{Cliques with large external degree ($A_H$):} true cliques of size between $c-\Theta(\sqrt{\Delta})$ and $c$ and whose vertices have at most $O(\sqrt{\Delta})$ external neighbors,
        \item \textbf{Low-slack intermediates ($B_L$):} vertices that still have many neighbors in cliques, but not as many as the vertices in $B_H$,
        \item \textbf{Cliques with small external degree ($\AL$):} the final cliques to be colored, each with even smaller external degree bounded by $O(\Delta^{1/4})$.
    \end{itemize}
\end{abox}

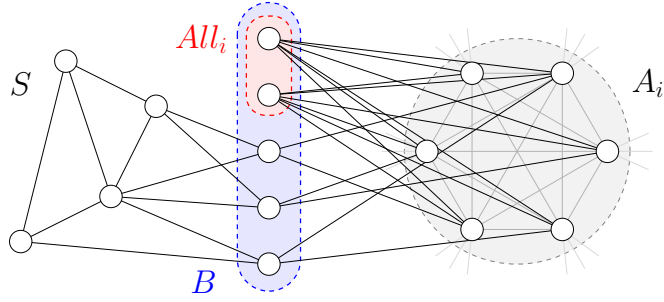
\begin{figure}[ht!]
    \begin{center}
        \scalebox{0.6}{\begin{tikzpicture}[scale=1]
\def\radius{2}
    \def\outlen{2}
    \def\numnodes{6}

\foreach \i in {1,...,\numnodes}{
        \coordinate (P\i) at ({360/\numnodes * (\i - 1)}:\radius);
    }

\def\xleft{-5.5}
    \foreach \j [count=\k from 1] in {2.5, 1.25, 0, -1.25, -2.5}{
        \coordinate (L\k) at (\xleft, \j);
    }

\coordinate (S1) at (-8, 1.0);
\coordinate (S2) at (-10, 2);
\coordinate (S3) at (-9, -1);
\coordinate (S4) at (-11, -2.0);

\draw[gray, dashed, very thin, fill=gray!10] (0,0) circle (2.5);

\foreach \i in {1,...,\numnodes}{
        \foreach \j in {1,...,\numnodes}{
            \ifnum\i<\j
                \draw[gray!60, very thin] (P\i) -- (P\j);
            \fi
        }
    }

\foreach \i in {1,...,\numnodes}{
    \pgfmathsetmacro{\angle}{360/\numnodes * (\i - 1)}
    \foreach \shift in {-24, 0, 24}{
        \draw[gray!30, thin] (P\i) -- ++({\angle+\shift}:{0.5*\outlen});
    }
}

\draw[blue, thick, dashed, rounded corners=20pt, fill=blue!10]
        (-6.2, 3.3) rectangle (-4.8, -3.1);

\draw[red, thick, dashed, rounded corners=12pt, fill=red!10]
        (-6.0, 3) rectangle (-5.0, 0.8);

\foreach \k in {1,2}{
        \foreach \i in {1,...,\numnodes}{
            \draw[black, thin] (L\k) -- (P\i);
        }
    }

\foreach \i in {2,5}{
        \draw[black, thin] (L3) -- (P\i);
    }
\foreach \i in {4,1}{
        \draw[black, thin] (L4) -- (P\i);
    }
\foreach \i in {2,6}{
        \draw[black, thin] (L5) -- (P\i);
    }
\draw[black, thin] (S1) -- (L3);
\draw[black, thin] (S1) -- (L4);
\draw[black, thin] (S3) -- (L3);
\draw[black, thin] (S3) -- (L5);
\draw[black, thin] (S3) -- (L4);

\draw[black, thin] (S4) -- (L5);
\foreach \i in {1,...,\numnodes}{
        \node[circle, draw=black, fill=white, inner sep=5pt] at (P\i) {};
    }

\foreach \k in {1,...,5}{
        \node[circle, draw=black, fill=white, inner sep=5pt] at (L\k) {};
    }

\foreach \i/\j in {1/2,2/3,3/4,1/3,2/4}{  \draw[black, thin] (S\i) -- (S\j);
}

\foreach \k in {1,...,4}{
    \node[circle, draw=black, fill=white, inner sep=5pt] at (S\k) {};
}
\node at (2.9,1.5) {\huge$A_i$};
    \node at (-7,2.5) {\color{red}\huge$All_i$};
    \node at (-6.95,-2.9) {\color{blue}\huge$B$};
    \node at (-11,1.5) {\huge$S$};
\end{tikzpicture}
 }        
    \end{center}
    \caption{A schematic illustration of the structural decomposition by Molloy and Reed. The vertices of $S$ are loosely connected, while the set $A_i$ is a clique. The intermediate set $B$ is too highly connected to $A$ to be considered sparse, but not densely enough to be itself part of the clique. For the coloring, we further divide the set $B$ into two subsets, $B_H$ and $B_L$ (see \Cref{lem:structuralDecomposition}); however, the figure shows $B$ as a single set for simplicity. Note that some vertices of $B$, here in the set $All_i$, can be connected to all the vertices of a clique (see \Cref{lem:structuralDecomposition} for details on the set $All_i$).}
    \label{fig:decomposition}
\end{figure}

While we follow the same overall structure as Molloy and Reed’s algorithm, the internal components differ significantly.
Most notably, the fact that we aim to color $\pious$ subgraphs means that we can use a simpler analysis.
First, we observe that $B_H$ and $B_L$ are $\pious$ by definition; we then provide a simple method to partition $S$ into two $\pious$ subgraphs; and finally design a highly parallelizable approach for coloring the cliques.
Coloring $S,B_H$ and $B_L$ is one of the most technically involved parts in \cite{MR14}. We shortly detail each of these steps.

\paragraph{Coloring nodes in $B_H$ and $B_L$.} By construction, nodes in $B_H$ and $B_L$ have many uncolored neighbors in cliques, which provides them with temporary slack; as a result, they satisfy \cref{prop:Pi} purely due to structural properties.
It is crucial that the vertices in $B_L$ are colored after the cliques in $A_H$. 
Although lists of vertices in $B_L$ are guaranteed to have size $\Omega(\sqrt{\Delta})$, this is insufficient for $B_L$ being $\Pi$-ous if $A_H$ is still uncolored.

\paragraph{Coloring the Sparse Vertices in $S$.}
It is well known that a single round of random color trials (\RCT) gives each vertex $v \in S$ slack $\Omega(\sqrt{\Delta})$. This process is called slack generation. Unfortunately, this is still insufficient for~\Cref{prop:Pi} requiring slack $U\cdot \Delta^{0.22}$ where $U=\Theta(\sqrt{\Delta})$.
We therefore split the remaining uncolored vertices~$S'$ into two groups $S_1$ and $S_2$. Coloring~$S_1$ first gives their vertices extra temporary slack (from neighbors in $S_2$), and when we color $S_2$ we can rely on the fact that dense vertices have few neighbors in $S_2$ making it easier to satisfy \Cref{prop:Pi}.

\paragraph{Coloring of Cliques.}
To achieve a sublogarithmic distributed runtime, we color cliques in $A_H$ and $\AL$ respectively in parallel.
For each clique~$A_i$, we assign a random permutation of the colors to its vertices avoiding any conflicts inside the clique.
With this process, only $O(\sqrt{\Delta})$ vertices per clique conflict with external neighbors.
Each such \emph{unhappy} vertex~$v$ seeks a  partner~$u$ in~$A_i$ so that both become happy if they swap their colors.
\cref{def:notbig} crucially guarantees a large supply of valid swap candidates.

To execute swaps in all cliques concurrently, we first independently downsample the candidate sets and remove those causing \emph{external conflicts}, i.e., swaps that cannot be executed safely in parallel with sampled candidate swaps in adjacent cliques. 
We show that, after this pruning step, sufficiently many candidates remain to also avoid \emph{internal conflicts}.

We set up a bipartite graph between unhappy vertices and their remaining candidates and show that it satisfies Hall’s condition, ensuring a perfect matching between unhappy nodes and swap candidates. We can then perform the corresponding swaps simultaneously without interference within the clique and due to the prior pruning also without  interference between different cliques. 

\subsection{Smaller Degrees and the \lovasz Local Lemma}
\label{sec:tech-intro-LLL}

All steps of our coloring algorithm are randomized and have a local failure probability of the form $\exp(-\Delta^\eps)$ for some universal constant $\eps \in (0,1]$.
Hence, when $\Delta \gg \log^{1/\eps} n$, our algorithm succeeds with high probability. 
However, when $\Delta \leq \poly(\log n)$, local failures can no longer be avoided.

The main tool at our disposal to handle local failures is the \lovasz Local Lemma. Since we aim for a sublogarithmic runtime, we cannot rely on the general framework by Moser and Tardos \cite{MoserTardos10,CPS17}. 
The only general approach currently known to achieve sublogarithmic round complexity is the \emph{shattering framework} \cite{BEPSv3}.
Unfortunately, there exists no general-purpose distributed LLL solver yielding sublogarithmic time on graphs of all degrees, 
in particular in the challenging regime where degrees are sublogarithmic but exceed $\poly(\log\log n)$.
\footnote{We cannot provably rule out that some of our LLLs could be solved via the sublogarithmic-time algorithm of \cite{GHK18}, but we do not know how to prove that they fall into the handled class of LLLs either.} 
For very low degree graphs, sublogarithmic-time algorithms are known~\cite{FG17,Davies23}. We therefore design dedicated constructions tailored to the different degree regimes.

\paragraph{The Shattering Framework.}
In this framework, a fast, randomized \emph{pre-shattering phase} resolves most of the graph -- for instance, fixing the colors of the majority of vertices -- leaving only small residual components of size $\poly(\log n)$.  
A subsequent \emph{post-shattering phase} then deterministically completes the solution by solving a smaller residual LLL problem, typically within $\poly(\log\log n)$ rounds.  

As one randomly colors vertices in the pre-shattering phase, one may be forced to uncolor certain vertices. For instance, if a color appears too frequently around a clique and \cref{def:notbig} is violated.
This may interfere with the progress measures of our algorithm, which are typically formulated as additional bad events.
For instance, uncoloring vertices may violate an event ensuring a constant-factor decrease in the uncolored degree. Such events are usually resolved during the post-shattering phase.

The key challenge, especially for non-greedy problems such as ours, is to design the pre-shattering phase in such a way that the post-shattering instance (1) is solvable and (2) has polylogarithmic size components, meaning that retractions (e.g., uncolorings) do not percolate to the large parts of the graph. 
We introduce several methods to deal with these challenges. 

\begin{compactitem}
    \item \emph{Fresh budgets} are used in the pre-shattering and the post-shattering whenever we do not want to see too many of something. By giving half of the total budget to each of the shattering steps, we remove the dependencies between the two phases.
    
    \item \emph{Marginal events} are added in the post-shattering phase to account for previously avoided bad events that could be reactivated by newly assigned variables.
    
    \item \textit{Guard events} impose additional constraints to preserve sufficient randomness for the post-shattering phase.
    They prevent the pre-shattering process from fixing too many variables or overcommitting to specific random outcomes, which would reduce the available randomness needed later. 

    \item \emph{Palette splitting} further separates randomness between pre- and post-shattering: one part of the color palette is used in pre-shattering, the other in post-shattering. Despite this restriction, the probabilistic guarantees of each step continue to hold.
\end{compactitem}

\paragraph{Example 1: Subsampling Swap Candidates.}
To illustrate those two ideas, consider the LLL that governs subsampling of swap candidates in cliques.  
Each unhappy vertex~$v$ independently samples potential partners with a fixed probability; with overwhelming probability, every vertex obtains sufficiently many candidates.  
The subtlety lies in showing that, after pruning candidates involved in conflicting swaps across adjacent cliques, enough valid partners remain.  
This is achievable, and the shattering analysis ensures that the components where the sampling fails are small.  

Since we only pick candidates in pre- and post-shattering using separate budgets, this increases the number of neighbors with a given color by a factor of at most two compared to having one single budget for both steps. Conveniently, the same analysis bounds the probability of breaking \cref{def:notbig} for the pre- and post-shattering.

Removing candidates does not make other candidates bad, however, resampling candidates to repair local failures can affect neighboring cliques, and potentially undo the progress made during pre-shattering.  To resolve this, we add an event in post-shattering for each clique adjacent to some other clique that resamples its candidates in post-shattering.

\paragraph{Example 2: Slack Generation.} To illustrate the latter two techniques, let us look at the slack generation LLL.  
Sparse nodes have many non-edges in their neighborhood, i.e., pairs of non-adjacent nodes. One can show that one round of random color trials provides a node with slack proportional to the number of non-edges by same-coloring the \emph{endpoints} of many of these non-edges. One can show that receiving sufficient slack forms an LLL. 

In a pre-shattering phase based on randomly coloring vertices, a node may also obtain too little slack. Now, if all of its neighbors are already colored,  it is impossible for it to obtain slack in the post-shattering phase. Uncoloring its neighbors would remove the slack of other nodes and in fact this may percolate  through the graph. Instead, we use a guarding event that ensures that in the pre-shattering phase few enough  neighbors of each node participate in the coloring procedure ensuring that enough neighbors (and also non-edges between them) remain uncolored for a tentative post-shattering phase. Of course, this event can also fail but one can verify that the guarding event prevents percolation of uncolorings. 

If enough neighbors (and also non-edges between them) remain uncolored in post-shattering, we can use the same random process with a disjoint palette to decouple the randomness of the pre-shattering and post-shattering phase. 

\paragraph{The Quest for a Sublogarithmic LLL Framework.}
We solve several LLLs by a combination of the mitigations above. Developing a general black-box framework that can solve all distributed LLLs of this form in sublogarithmic time remains an open challenge.  
Determining precisely which classes of LLL instances admit such runtimes is, in our view, one of the central open problems in the distributed complexity of the \lovasz Local Lemma.
\begin{openproblem}
Does there exist a sufficiently large constant $c$ such that all LLLs with dependency degree $d$ and  local failure probability $p \leq d^{-c}$ admit a $\poly(\log\log n)$ round algorithm? 
\end{openproblem}

\begin{openproblem}
Does there exist a constant $d_0$ such that all LLLs with dependency degree $d\geq d_0$ and  local failure probability $p \leq \exp{-\Omega(d^{\eps})}$ for some constant $\eps\in (0,1)$ admit a $\poly(\log\log n)$ round algorithm?
\end{openproblem}

\section{Preliminaries}
\label{sec:prelim}

\paragraph{Graphs.}
For a graph $G=(V,E)$, we denote by $\Delta(G)$ its maximum degree. For a set $S \subseteq V$, denote by $G[S]$ the subgraph of $G$ induced by $S$. The neighborhood of a vertex in a graph $G$ is denoted by $N_G(v) = \set{ u \in V \mid \set{u,v}\in E(G)}$, and, by extension, for any $X \subseteq V$, we write $N_G(X) = \bigcup_{v \in X} N(v)$ for the set of vertices adjacent to some vertex in $X$. We also use $N_G^{\leq t}(X)$ to denote the set of vertices within distance $t$ of a vertex in $X \subseteq V$. 
For a vertex $v\in G$ we denote the  degree of $v$ into $H$ by $\deg_H(v)=|N_G(v)\cap V(H)|$. When graphs are clear from the context we omit the respective indices. For $S \subseteq V$, we write $G - S$ for the graph $G[V \setminus S]$ induced by the vertices outside of $S$. For a vertex $v$ and set $S \subseteq V$, define $\dist(v, S) = \min_{u\in S} \dist(v, u)$. If $\mathcal{A}$ is a collection of sets of vertices, then let $\dist(v, \mathcal{A}) = \min_{S \in \mathcal{A}} \dist(v, S)$.

\paragraph{Colorings.}
For an integer $c \geq 1$ we denote $[c]=\{1,\ldots,c\}$. A partial $c$-coloring is a function $\phi : V \to [c] \cup {\bot}$ such that for all $\set{u,v}\in E$ are such that $\phi(u) \neq \phi(v)$ unless $\phi(u) = \bot$ or $\phi(v) = \bot$. The domain of $\phi$ is the set of colored vertices: $\dom\phi = \set{ u \in V : \phi(u) \neq \bot}$. Given a coloring $\phi$, a color is \emph{available} to $v \in V$ if none of its neighbors have that color under $\phi$. The set of available colors, or palette, is denoted by $L(v) = L_\phi(v) = [c] \setminus \set{ \phi(u) \mid u\in N(v) }$.
The \emph{slack} of a vertex $v$ with respect to a coloring $\phi$ and induced subgraph $H \subseteq G$ is the difference between the number of available colors and the number of uncolored neighbors in $H$, i.e., the slack is $s_{H,\phi}(v) =|L_\phi(v)| - |N_H(v) \setminus \dom\phi|$. 
We emphasize that the slack is always defined with respect to a subgraph $H$ and that it does not decrease as we color additional vertices of $H$.

\paragraph{Probabilities.}
When we say that an event holds ``with high probability in $n$'', abridged to \whp, we mean that it holds with probability at least $1 - n^{-c}$ for some desirably large constant $c > 0$. We often omit to mention $n$ explicitly as every such statement in this paper holds \whp in $n$, where $n$ is the number of vertices in the \LOCAL network.

\paragraph{Assumption on $\Delta$.}
As stated in \cref{thm:main}, we assume that $\Delta \geq \Delta_0$ where $\Delta_0$ is a sufficiently large universal constant; we frequently use this property throughout the paper to simplify inequalities. 

\subsection{\lovasz Local Lemma}

Let $(\Omega_i, \mathbb{P}_i)_{i=1}^n$ be probability spaces and $X_i : \Omega_i \to \mathbb{R}$ a family of mutually independent random variables for all $i\in [n]$. 
Let $\mathcal{B} = \{B_1, B_2, \dots, B_m\}$ be a family of undesirable (\emph{bad}) events in $\Omega = \prod_{i=1}^n \Omega_i$ and, for all $i\in[n]$, let $\var(B_i) \subseteq [n]$ be sets such that each indicator random variable $1(B_i)$ is a function of $\set{ X_j : j \in \var(B_i) }$.
Two events $B_i, B_j$ are adjacent in the \emph{dependency graph} if 
$\mathrm{var}(B_i) \cap \mathrm{var}(B_j) \neq \emptyset$. 
The classical \lovasz Local Lemma (LLL) \cite[Chapter 5]{AlonSpencerBook} states that if there exist parameters $p < 1$ and $d$ such that every event $B_i$ satisfies 
$\Prob{ B_i } \le p$, and
$ep(d+1) \le 1$, where $d$ is the maximum degree of the dependency graph, 
then there exists an assignment of the variables $X_i$ that avoids all events in~$\mathcal{B}$. 

In the \emph{distributed \lovasz Local Lemma} \cite{CPS17}, 
the input graph is the dependency graph: vertices are events of $\mathcal{B}$ and two events are connected by an edge iff they share a variable.
The objective is for all nodes to collaboratively find an assignment of the variables that avoids every event in~$\mathcal{B}$, 
using as few communication rounds as possible. In most applications, the input graph is not the dependency graph itself, but a round of communication on the dependency graph can be simulated by $O(1)$ rounds of \LOCAL on the input graph.

In the context of our coloring algorithms, the random variables correspond to each node's random choices such as its selected candidate color or whether it participates in a given trial. The bad events are the undesired outcomes of these choices, for example, for each clique there is a bad event that holds if its color coverage property is violated. 

\begin{theorem}[Deterministic LLL in \LOCAL, \cite{RG20,GG24}]
    \label{thm:deterministicLLLLOCAL}
    There is a constant $\eps > 0$ for which the following holds.
    There is a deterministic \LOCAL algorithm for the constructive Lovász Local Lemma with $n$ events under criterion $epd(1+\eps)<1$ that runs in $O(\log^* s)+ \Ot{\log^3 n}$ rounds if node IDs are from a space of size $s$. 
\end{theorem}

\noindent
\Cref{thm:deterministicLLLLOCAL} is proven by using the powerful general derandomization framework of \cite{RG20,GHK18,GKM17} for the algorithm of Moser-Tardos \cite{MoserTardos10} and using the fastest $\Ot{ \log^2 n }$ algorithm for computing the required network decompositions \cite{GG24}.

\subsection{Our Shattering Framework}
\label{sec:our-shattering-framework}

Throughout this paper, we use LLL instances to compute certain good vertex labels, e.g., colors.
To solve LLL instances in sublogarithmic time, we use the shattering technique \cite{BEPSv3}. We describe in this section the general approach that we follow throughout this paper.

\paragraph{Sublogarithmic LLL Algorithm.}
Every LLL that we solve in this paper follows the same general approach: each vertex $v$ of the graph to be colored has a random variable $X_v$ (typically a random color) and the vertex is involved in $\poly(\Delta)$ many events. Each event occurs with probability at most $\exp(-\Delta^{1/40})$, and hence when $\Delta \geq (\log n)^{50}$, it suffices to sample the $X_v$ uniformly to obtain a globally satisfying assignment with high probability. 

When $\Delta$ is small, we split the random process as follows:
\begin{enumerate}[label=\Roman*:]
    \item In a \emph{pre-shattering} random process, each vertex $v$ samples its random variable $X_v$ (typically a random color). We have a collection $\mathcal{A}$ of bad events that can then occur.
    \item We \emph{retract} the variables held by the vertices $v$ within $O(1)$ hops in $G$ (the variable graph) of a vertex used by an occurring event in $\mathcal{A}$. The exact meaning of retracting $X_v$ depends on the concrete LLL instance but generally means that the value of $X_v$ will be ignored for the final output. 
    \item In a second \emph{post-shattering} random process, we sample variables $Y_v$ for a subset of the vertices that includes at least all the retracted vertices. For this step, we also have a set of bad events $\mathcal{B}$ that must form an LLL. A satisfying assignment of the $Y_v$ avoiding all bad events of the LLL is computed using the deterministic algorithm of \cref{thm:deterministicLLLLOCAL}.
\end{enumerate}
Each vertex $v$ produces its final value (e.g., color) by combining its $X_v$ and $Y_v$ values, usually as the value of $Y_v$ if it was defined and $X_v$ otherwise.  For this scheme to succeed and run fast, it is essential that:
\begin{enumerate}
    \item Every event in $\mathcal{A}$ occurs with low probability,
    \item The connected components of the dependency graph of $\mathcal{B}$ have $\poly(\log n)$ size, and
    \item The events $\mathcal{B}$ form an LLL.
\end{enumerate}
Our analysis focuses on proving those three points.
In general, the  round complexity of our pre-shattering and retraction phases is constant and the overall round complexity is dominated by applying \Cref{thm:deterministicLLLLOCAL} on the connected components of the dependency graph of $\mathcal{B}$. The runtime of $\poly(\log\log n)$ then follows from point 2, as the components have size $N = \poly(\log n)$ and \cref{thm:deterministicLLLLOCAL} ends after $\poly(N)$ rounds. 
To prove point 2, we use a shattering argument.

\paragraph{The Shattering Lemma.}
The following lemma, initially introduced by \cite{BEPSv3}, shows that when vertices of a graph of maximum degree $\Delta$ get selected with probability $1/\poly(\Delta)$, the connected components of the selected vertices are small. Recall that this is only used when $\Delta \leq \poly(\log n)$, so connected components have size at most $\poly(\log n)$. We use the following general purpose shattering statement.

\begin{lemma}[Shattering Lemma~\cite{FG17}]\label{lem:Shattering}
Let $G=(V, E)$ be a graph with maximum degree $\Delta$. Consider a process that generates a random subset $B \subseteq V$ such that $\Prob{ v \in B } \leq \Delta^{-C_1}$, for some constant $C_1 \geq 1$, and such that the random variables $\I{ v\in B }$ depend only on the randomness of nodes within at most $C_2$ hops from $v$, for all $v\in V$, for some constant $C_2\geq 1$.
Then, for any constant $C_3\geq 1$, satisfying  $C_1>C_3+ 4C_2 + 2$,  we have that any connected component in $G[B]$ has size at most $O(  \Delta^{2C_2}\log_{\Delta} n)$ with probability at least $1- n^{-C_3}$.
\end{lemma}

To avoid technical redundancies, we use the following shattering lemma instead of \cref{lem:Shattering}. It takes advantage of the fact that, in all our application, our LLLs have one random variable per vertex.  So we consider two LLLs $\mathcal{A}$ and $\mathcal{B}$ whose variables are indexed by the vertices of a graph $G$ called the \emph{variable graph}. The set $\mathcal{A}'$ contains all the bad events that occur after the pre-shattering step and whose variables are retracted. Every event $B \in \mathcal{B}$ with a variable within distance $\cstB$ in $G$ of a retracted variable, i.e., in $\vbl(\mathcal{A}')$, is included in the post-shattering LLL which we call $\mathcal{B}'$. \cref{lem:our-shattering} states that the dependency graph $\Gev$ of the post-shattering LLL $\mathcal{B}'$ has small connected components if $\vbl(B)$ is \emph{locally embedded} in the variable graph $G$ for every event $B \in \mathcal{A} \cup \mathcal{B}$. This is formalized by \cref{part:our-shattering-diam,part:our-shattering-involved}. 

For simplicity, we state \cref{lem:our-shattering} in terms of collections of sets $\mathcal{A}, \mathcal{A}',\mathcal{B}, \mathcal{B}'$ instead of LLLs. When we use \cref{lem:our-shattering}, we replace every bad event $B$ with the set of vertices in $G$ that hold its variables. The fact that $\mathcal{A}'$ contain rare bad events from an LLL is formalized by \cref{part:our-shattering-pre-prob,part:our-shattering-vbl}.

\begin{lemma}
    \label{lem:our-shattering}
    Let $G=(V,E)$ be an $n$-vertex graph with maximum degree $\Delta$. 
    Consider $\mathcal{A}, \mathcal{B} \subseteq 2^V$ two collections of subsets of $V$. Suppose there exist constants $c_1, c_2, c_4 \geq 1$ and $\cstB \geq 0$ with $c_4 > 3c_1(4\cstB + 16) + c_2 + \cstB+1$ such that 
    \begin{enumerate}[label=(S\arabic*),leftmargin=1.5cm]
        \item \label[part]{part:our-shattering-diam}
        for all $S \in \mathcal{A} \cup \mathcal{B}$, the set $S$ has weak-diameter at most $c_1$ in $G$, and
        \item \label[part]{part:our-shattering-involved}
        for all $v\in V$, there are at most $\Delta^{c_2}$ sets $S \in\mathcal{A} \cup \mathcal{B}$ such that $v\in S$.
    \end{enumerate}
    Let $X_v$ be independent random variables indexed by the vertices of $G$.
    Consider a random process on the corresponding probability space that produces a collection $\mathcal{A}' \subseteq \mathcal{A}$ such that, for all $A\in \mathcal{A}$,
    \begin{enumerate}[start=3,label=(S\arabic*),leftmargin=1.5cm]
     \item \label[part]{part:our-shattering-vbl}
        the random variable $\I{ A \in \mathcal{A}' }$ is a function of the $\set{ X_v : v\in A }$, and 
        \item \label[part]{part:our-shattering-pre-prob}
        $\Prob{ A \in \mathcal{A}' } \leq \Delta^{-c_4}$~.
       
    \end{enumerate}
    Define
    \[
    \mathcal{B}' = \set{ B \in \mathcal{B} : \exists v\in B, \dist_G( v, \mathcal{A}' ) \leq c_{\mathcal{B}'} } \ .
    \]
     
    Let $\Gev$ be the (random) graph on vertex set $\mathcal{B}'$ with edges between $B$ and $B'$ iff $B \cap B' \neq \emptyset$.
    Then, the connected components of $\Gev$ have size at most $O(\Delta^{4c_1\cstB + 4c_1 + c_2}\log n)$, with high probability in $n$.
\end{lemma}

\begin{proof}
    Let $H = G^{2c_1}$, namely the graph on vertex set $V$ with an edge $\set{u,v}$ iff $\dist_G(u, v) \leq 2c_1$. 
    For every $B \in \mathcal{B}'$, let $f(B) \in B$ be the vertex such that $\dist_G(f(B), \mathcal{A}') \leq \cstB$, which exist by definition of $\mathcal{B}'$. Observe that $f$ is an adjacency-preserving map from $\Gev$ to $H[N_G^{\leq \cstB}(\mathcal{A}')]$: by \ref{part:our-shattering-diam}, if $B \cap B' \neq \emptyset$ (i.e., they are adjacent in $\Gev$), then $f(B)$ and $f(B')$ are connected by an edge in $H$. 
    Let $C$ be the largest connected component of $\Gev$. The image of $C$, denoted $f(C)$, is a connected subgraph of $H[N_G^{\leq \cstB}(\mathcal{A}')]$. On the other hand, by \ref{part:our-shattering-involved} each vertex of $H$ is the image of at most $\Delta^{c_2}$ vertices of $C$ through $f$, and so $|C| \leq \Delta^{c_2}|f(C)|$. It therefore suffices to upper bound the size of the largest connected component of $H[N_G^{\leq \cstB}(\mathcal{A}')]$.

    By \ref{part:our-shattering-involved}, \ref{part:our-shattering-pre-prob}, and the union bound over all sets $A$ with $v$ as a vertex, we have that 
    \[ 
    \Prob*{ v \in \bigcup_{A\in \mathcal{A}'} A } \leq \Delta^{c_2 - c_4} \ . 
    \]
    By the union bound over the vertices of $N_G^{\leq \cstB+1}(v)$, we deduce that 
    \[
    \Prob*{ \dist_G(v, \mathcal{A'}) \leq \cstB } 
    \leq \Delta^{\cstB + 1} \cdot \Delta^{c_2 - c_4} 
    \leq \Delta(H)^{(\cstB + 1 + c_2 - c_4)/(3c_1)} \ . 
    \]
    By \ref{part:our-shattering-vbl}, whether $A \in \mathcal{A}'$ depend only on the random variables $X_u$ of vertices in $A$, which has diameter one in $H$. Whether $v\in N_G^{\leq \cstB}( \mathcal{A}' )$ depends on whether $A \in \mathcal{A}'$ for all the $A$ with a vertex in $N^{\leq \cstB}_G(v)$, hence it depends only on the random values within $\cstB+1$ hops in $H$. Let $C_1 = (c_4 - \cstB - c_2 - 1)/(3c_1)$, $C_2 = \cstB+1$ and $C_3 = 10$; it is easy to verify that $C_1 > C_3 + 4C_2 + 2$ for our choice of $c_1, c_2, c_4, \cstB$. By \cref{lem:Shattering}, \whp, all the components of $H[N^{\leq \cstB}_G(\mathcal{A}')]$ have size at most $O( \Delta(H)^{2\cstB + 2} \log n ) = O( \Delta^{4c_1\cstB + 4c_1} \log n )$, which implies the claimed bound on the size of the largest connected component in $\Gev$.
\end{proof}

\section{Top-Level Algorithm and Proofs of \texorpdfstring{\cref{thm:main,thm:det-alg}}{Theorems~\ref{thm:main},~\ref{thm:det-alg}}}
\label{sec:high-level}
The goal of this section is to present the top-level algorithm for proving \Cref{thm:main,thm:det-alg}. 
\thmmain*

We assume that each vertex knows its incident edges and the values of $\Delta$ and $c$.

The general problem reduces to the coloring of a graph $F$ with a highly specific structure (see Theorem 5 in \cite{MR14}). In \Cref{sec:decomposition}, we present a structural decomposition of $F$ that forms the base of our coloring algorithm. In \Cref{ssec:CC}, we detail the crucial \emph{color coverage} property that we need to maintain throughout our algorithm in order to find suitable color swaps when completing the coloring for the cliques of $F$. 
In \Cref{sec:proof-thm-main}, we present the proof of \Cref{thm:main}, building on \Cref{lem:colorWithMuchSlack,lem:sparse,lem:coloringCliques} that state how fast we can color the various parts of $F$. These lemmas are proven in \Cref{sec:colorwithmuchslack,sec:sparse,sec:cliques}.

\subsection{Graph Decomposition}
\label{sec:decomposition}
The coloring problem of \Cref{thm:main} reduces to coloring a graph $F$ with a specific structure. This structural decomposition, introduced in \cite{MR14} and efficiently computable in \LOCAL\ by \cite{BamasE19}, is technically involved. Since later proofs rely on several additional structural properties, their original presentation becomes cumbersome. We therefore adapt the decomposition to our setting and restate it in the following (still technical) lemma, including all necessary properties. 

\begin{restatable}[Structural Decomposition]{lemma}{lemStructDecomposition}
    \label{lem:structuralDecomposition}
    There is a decomposition of $F$ into sets of vertices $S$, $B_H$, and $B_L$,  and into two collections of cliques $A_H$ and $\AL$ that are colored in the order presented below (see \cref{alg:highlevel}). Additional properties that the sets satisfy independently of the coloring in previous steps are stated afterwards.
    \begin{enumerate}
        \item \label[part]{part:decomp-S}
        $S$: Each $v\in S$ has $\deg_{S}(v)<\Delta-3\sqrt{\Delta}$ or it has $\deg_S(v)\leq \Delta$ and at least $9\cdot 10^5\cdot \Delta^{3/2}$ non-adjacent pairs of neighbors within $S$.     
        \item\label[part]{part:slack-BH}
        $B_H:$ Each $v\in B_H$ satisfies $|L(v)|\geq \deg_{B_H}(v)+\Delta^{3/4}$, regardless of how $S$ is colored.   
        \item \label[part]{part:A-ext}
        $A_H:$ The external degree (cf.~\cref{def:externalDegree}) of each $A_i\in A_H$ is bounded by $10^8\sqrt{\Delta}$.     
        \item \label[part]{part:slack-BL}
        $B_L:$ Each $v\in B_L$ satisfies $|L(v)|\geq \deg_{B_L}(v)+ \tfrac{1}{2}\sqrt{\Delta}$, regardless of how $S, B_H, A_H$ are colored.   
        \item \label[part]{part:AL-ext}
        $\AL$: The external degree of each $A_i\in \AL$ is bounded by $30\Delta^{1/4}$ .   
        \end{enumerate}
        Additionally we have: 
        \begin{enumerate}[label=(\alph*)]
        \item \label[part]{part:A-size}
        Every $A_i$ is a clique with $c-10^8\sqrt{\Delta}\le|A_i|\le c$.   
        \item \label[part]{part:All} 
        $\mathrm{All}_i\subseteq B_L\cup B_H$ Vertices in $\mathrm{All}_i \subseteq B$  are adjacent to all of $A_i$. $|\mathrm{All}_i|= c-|A_i|$ holds.  
        \item \label[part]{part:Big}
        $\Big_i^+\subseteq  B_L\cup B_H$: A vertex \( v \notin A_i \cup \All_i \) lies in \( \Big_i^+ \) if it has at least \( 2\Delta^{9/10} \) neighbors in \( A_i \).
        The set \( \Big_i^+ \) is a clique, and each \( v \in \Big_i^+ \) has at most \( \tfrac{3}{4}\Delta + 10^8\sqrt{\Delta} \) neighbors in \( A_i \).
    \end{enumerate}
\end{restatable}

\begin{definition}
\label{def:externalDegree}
The \emph{external degree} of a clique $A_i$ is the maximum number of neighbors of a $v\in A_i$ outside of $A_i\cup \mathrm{All}_i$. 
\end{definition}

The proof of \Cref{lem:structuralDecomposition} appears in \Cref{app:decomposition}.

\subsection{The Color Coverage (CC) Property}
\label{ssec:CC}
In this section, we formally introduce the \emph{color coverage} property, which we abbreviate \covd, that was outlined in \cref{sec:tech-intro}. Informally speaking we wish to ensure that the coloring outside of cliques $A_i$ looks random enough. Formally, we ask that each color is available to a constant fraction of the nodes on the inside of the cliques.

While $\All_i$ contains the nodes that are connected to all nodes in $A_i$, the nodes in $\Big_i^+$ are connected to most nodes (but still not too many) of $A_i$. Hence, coloring a single node in $\Big_i^+$ may bring $A_i$  close to not satisfying \cref{def:notbig}. But due to \Cref{lem:structuralDecomposition} the nodes in $\Big_i^+$ form a clique  for $A_i$; hence, no two nodes in $\Big_i^+$ can receive the same color.

\renewcommand{\theourProp}{CC} 
\begin{prop}
    \label{def:notbig} Consider some arbitrary step $j$ of a coloring algorithm. 
    Let $\covd_j(i,x)$ be the number of vertices in $A_i$ that have a neighbor $v \not\in A_i \cup \All_i \cup \Big_i^+$ that gets colored with $x$ in step $j$. We say that CC is \emph{within budget for a \underline{single coloring step}} if $\covd_j(i,x) < \Delta^{37/40}$ holds for every uncolored clique $A_i$ and for every color $x$. 
    
    We say that \emph{CC holds for a clique $A_i$} at a point in time if CC has been within budget for every coloring step applied so far.
\end{prop}

\noindent
Essentially, the $\Delta^{37/40}$ budget ensures that across all steps, each color remains sufficiently under-used around every clique, enabling valid swaps later. The under-usage across all steps is formalized in the following observation. 

\begin{observation}
\label{obs:CC}
If we maintain \cref{def:notbig} throughout the coloring process for a clique $A_i$, then for all colors $x$ and 
at any point in the coloring process, at most $4 \Delta/5$ vertices of $A_i$ have a neighbor outside of $A_i\cup \mathrm{All}_i$ with color $x$.
\end{observation}
\begin{proof}
Our algorithm has either $O(\log^* \Delta)$ or $O(\log \Delta)$ coloring steps\footnote{In several of our coloring steps we compute a partial coloring by setting up an LLL that we solve via the shattering framework. That means that we color vertices in a pre-shattering and a post-shattering phase, both of which count as a separate coloring step in the context of this observation. In other words, the pre-shattering and post-shattering phases in these LLLs obtain a separate CC budget which simplifies the analysis significantly. }, depending on the approach. 
Summing up the $\Delta^{37/40}$ budget of every individual step, results in $O(\log \Delta\cdot \Delta^{37/40})$ vertices in $A_i$ that have a $v \not\in A_i \cup \mathrm{All}_i \cup \Big_i^+$ that is permanently colored with color $x$ (in any previous step). 
Consider a fixed color $x$.
By \Cref{lem:structuralDecomposition}-\ref{part:Big}, $\Big_i^+$ forms a clique, and hence at most one node of $\Big_i^+$ can be colored $x$, and this vertex has at most $\tfrac34\Delta + 10^8\sqrt{\Delta}$ neighbors in $A_i$. Overall, the number of vertices $v \not\in A_i \cup \All_i$ adjacent to a vertex colored $x$ is bounded above by  
\begin{align*}
O(\log \Delta\cdot \Delta^{37/40})+\tfrac34\Delta + 10^8\sqrt{\Delta}\leq 4\Delta/5~.
& \qedhere\end{align*}
\end{proof}
 
\paragraph{Bounding $\covd_j(i,x)$.}
Suppose that we have a collection of at most $\Delta$ subsets of $V(F)$. Each set contains at most $Q$ vertices. No vertex lies in more than $2\Delta^{9/10}$ sets. We conduct a random experiment where each vertex is marked with probability at most $1/(Q \times \Delta^{1/5})$. The vertices are not necessarily marked independently, but the experiment has the following property\footnote{We remark that the exact value of the constant $1/5$ is not crucial. We can choose it as small as we like, but then the bound $\Delta^{37/40}$ inches closer to $\Delta$ and the exponent in the error probability gets smaller.}:
\begin{quote}
    (P7.1) For any set of $\ell\ge 1$ vertices, the probability that all are marked is at most $1/(Q \times \Delta^{1/5})^\ell$. 
\end{quote}
We use the following Lemma 32 of \cite{MR14} verbatim. Note that condition (P7.1) does not require independence between vertices of different sets, and also not for vertices inside the same set. 

\begin{lemma}[Lemma 32 of \cite{MR14}] \label{lem:bnd-notbig}
The probability that at least $\Delta^{37/40}$ sets contain at least one marked vertex is at most $\exp(-\Delta^{1/40})$.
\end{lemma}
The high level idea of using \Cref{lem:bnd-notbig} to show that CC is maintained for cliques in each step of our coloring procedure is as follows.  Most of our coloring steps are based on some variant of random color trials, in which each node picks one (or multiple) candidate colors and then permanently retains one of its candidate colors if none of its neighbors chose it as a candidate. To prove that such a process maintains CC, fix a color $x\in[c]$,  a clique $A_i$, and define a set $N_v$ for each vertex $v$ of the clique containing its external neighbors, except those in $\Big_i^+$. Let $Q$ be an upper bound on the sizes of the sets. As we exclude vertices in $\Big_i^+$ (those with more than $2\Delta^{9/10}$ neighbors in $A_i$) no vertex lies in more than $2\Delta^{9/10}$ sets. We consider a vertex $u\in \bigcup N_v$ marked if color $x$ is among its candidate colors. Clearly, the final colors picked by vertices are not independent, but the choices of candidate colors is.  Hence, if the list of each vertex is of size $Q\cdot \Delta^{1/5}$ the properties of \Cref{lem:bnd-notbig} are satisfied. We obtain that CC is maintained for color $x$ and clique $A_i$ with probability $\exp(-\Delta^{1/40})$. 

Note that in this process the events whether color $x$ appears in the neighborhood of two vertices  $v\neq v'$ may not be independent as $v$ and $v'$ may have a common external neighbor. The strength of \Cref{lem:bnd-notbig} is that it can still deal with this situation. 

\subsection{Full Algorithm and Proofs of \texorpdfstring{\cref{thm:main,thm:det-alg}}{Theorems~\ref{thm:main},~\ref{thm:det-alg}}}

\label{sec:proof-thm-main}

The proof of \Cref{thm:main} appears at the end of this section. We first focus on coloring the graph $F$ from which we can efficiently recover a coloring of $G$. See \Cref{alg:highlevel} for pseudocode and the order in which we process vertices of $F$. 
\begin{algorithm} \caption{Overall $\Delta-k$ Coloring Algorithm}
        \label{alg:highlevel}
        \Input{Graph $F$ from \cref{lem:structuralDecomposition} with the vertex-partition $S,B_H,A_H, B_L,\AL$}
        \Output{Coloring of $F$}
        \BlankLine
        \alg{ColorSparse}($S$)
        
        \alg{ColorWithMuchSlack}($B_H$) 
        
        \alg{ColorCliques}($A_H$)   
        
        \alg{ColorWithMuchSlack}($B_L$)
        
        \alg{ColorCliques}($\AL$)\end{algorithm}

The guarantees required and provided by subroutines \alg{ColorWithMuchSlack}, \alg{ColorSparse}, and \alg{ColorCliques} are stated in \Cref{lem:colorWithMuchSlack,lem:sparse,lem:coloringCliques} below, and proven in \Cref{sec:colorwithmuchslack,sec:sparse,sec:cliques} respectively. We begin with the definition of $\Pi$-ous subgraphs that are central for coloring all non-clique vertices.  It is used to show that any probabilistic coloring of these vertices is random enough compared to the external degree of the still uncolored cliques. 

\paragraph{$\Pi$-ous subgraphs.}
A subgraph $H$ is $\Pi$-ous if it satisfies the following property: 
\renewcommand{\theourProp}{$\Pi$} \begin{prop}\label{prop:Pi}
An induced subgraph~$H$ of~$F$ satisfies property~$\Pi$ if there exists $U \ge \Delta^{1/4}$ s.t.
\begin{enumerate}[label=(\alph*)]
\item every uncolored vertex in each clique~$A_i$ has at most~$U$ neighbors in $H - \mathrm{All}_i$, and
\item each vertex~$u \in H$ has a list of available colors satisfying $|L(u)| \ge \deg_H(u) + U\cdot \Delta^{0.22}$.
\end{enumerate}
\end{prop}
Property (b) ensures that every vertex in $H$ has slack $U\cdot \Delta^{0.22}$, and importantly this slack remains as other vertices of the subgraph $H$ are colored. This allows us to color $H$ with $O(\log\Delta)$ iterated random color trial while maintaining Property CC.
In \Cref{sec:colorwithmuchslack} we prove the following lemma for coloring $\Pi$-ous subgraphs.

\begin{restatable}[\alg{ColorWithMuchSlack}]{lemma}{LemColoringWithMuchSlack}
\label{lem:colorWithMuchSlack}
Let $H$ be a \pious subgraph of $F - \bigcup_i A_i$. There is a \LOCAL algorithm that, w.h.p., extends the coloring to all the vertices of $H$ in $\Ot{ \log\Delta \cdot \log^3 \log n }$ rounds while maintaining Property CC.
\end{restatable}

Note that when $\Delta = n^{\eps}$, \cref{lem:colorWithMuchSlack} does not guarantee a sublogarithmic runtime.
For $\Delta \geq (\log n)^{50}$, we adapt the state-of-the-art $(\deg+1)$-list-coloring algorithm of \cite{HKNT22} for coloring a $\Pi$-ous subgraph while maintaining CC; see \Cref{sec:hideg} for details. 
Using \cref{lem:colorWithMuchSlack} for low-degree graphs incurs an additional $O(\log\Delta) = O(\log\log n)$ factor in the runtime, but greatly simplifies the analysis.

\begin{restatable}{lemma}{ThmHighDegree}
    \label{thm:d1LC}
    Let $H$ be a $\pious$ subgraph of $F  - \bigcup_i A_i$ and assume $\Delta \ge (\log n)^{50}$.
    There is a \LOCAL algorithm that list-colors all vertices of $H$ in $O(\log^* n)$ rounds while maintaining Property CC.
\end{restatable}

\medskip\noindent
We cannot use \cref{lem:colorWithMuchSlack} directly to color the vertices of $S$ because they do not have enough initial slack to satisfy \cref{prop:Pi}. As in \cite{MR14}, we begin by solving an LLL to provide them with $\Omega( \sqrt{\Delta} )$ slack which is still not sufficient for \Cref{prop:Pi}. We then color the vertices of $S$ with \cref{lem:colorWithMuchSlack} (or \cref{thm:d1LC}) in two batches to ensure that \cref{prop:Pi} holds, thereby that CC is maintained. Details can be found in \cref{sec:sparse}.

\begin{restatable}[\alg{ColorSparse}]{lemma}{LemColorSparse}
\label{lem:sparse}
    Let $c \geq \Delta - k_\Delta + 1$ and $S$ be the set of vertices in $F$ as in \cref{lem:structuralDecomposition}-(\ref{part:decomp-S}).
    There is a $\Ot{ \log^4 \log n }$-round distributed algorithm that, \whp, $c$-colors the vertices of $S$ while maintaining CC. When $\Delta \geq (\log n)^{50}$, it runs in $O(\log^* n)$ rounds.
    \label{lem:tough-cookie}
\end{restatable}

\noindent
As explained above, to color the cliques, we require that they satisfy \cref{def:notbig}. Since the cliques from $\AL$ are colored later, we must be careful about maintaining \cref{def:notbig} for those as we color $A_H$. We prove \cref{lem:coloringCliques} in \cref{sec:cliques}.

\begin{restatable}[\alg{ColorCliques}]{lemma}{LemColoringCliques}
\label{lem:coloringCliques}
Let $c \geq \Delta - k_\Delta + 1$ 
and let $A'$ be a subset of the cliques $A_i$ from \cref{lem:structuralDecomposition} such that 
\begin{enumerate}
    \item all $A_i\in A'$ satisfy \cref{def:notbig}, and
    \item \label[part]{part:color-cliques-uncolored}
    every uncolored vertex in some $A_i \notin A'$ has at most $30\Delta^{1/4}$ external neighbors.
\end{enumerate}
Then, there is a $\Ot{\log^3 \log n }$-round \LOCAL algorithm that, \whp, $c$-colors all the cliques in $A'$ while maintaining \cref{def:notbig} for all the uncolored cliques outside $A'$. For $\Delta\geq (\log n)^{50}$, it runs in $O(1)$ rounds. 
\end{restatable}

\begin{proof}[Proof of \Cref{thm:main}] 
Building on \cite[Theorem~5]{MR14}, Bamas and Esperet \cite[Theorem 4.1]{BamasE19} show that if a graph $G$ is not $c$-colorable for $c\geq\Delta-k_{\Delta}+1$, then there exists a vertex $v\in V(G)$ such that the graph induced by $\{v\}\cup N(v)$ has chromatic number $>c$. This can be tested in $O(1)$ rounds in the \LOCAL model exploiting the unbounded local computations. 

Otherwise, we compute a $c$-coloring of $G$ as follows. First, we compute the graph $F$ with vertex partition $S$, $B_H$, $B_L$ and two sets of cliques $A_H$ and $\AL$ via \Cref{lem:structuralDecomposition}. Then we $c$-color $F$ (see below), and recover a $c$-coloring of $G$ from the $c$-coloring of $F$ in the same way as in \cite{BamasE19}.  They show (in \cite{BamasE19}) that this can be done in time proportional to the complexity of the \emph{degree+$\Omega(\sqrt{\Delta})$-list coloring} problem. This runs in $\Ot{ \log^{5/3}\log n }$ in general and in $O(\log^* n)$ rounds for $\Delta \ge (\log n)^{14})$, with the algorithms of \cite{GG24,HKNT22}.

For the rest of the proof we focus on computing a $c$-coloring of $F$ by coloring the nodes in  $S$, $B_H$, $A_H$, $B_L$, and $\AL$ in the prescribed order (see Algorithm~\ref{alg:highlevel}). We first focus on the case when $\Delta\leq (\log n)^{50}$; the case of larger $\Delta$ is handled thereafter. 

\medskip\noindent
\textbf{Coloring S.} The sparse nodes are colored via \Cref{lem:sparse}. 

\medskip\noindent
\textbf{Coloring $B_H$.} 
We argue that the graph induced by $B_H$ is \pious with $U = U_H = 10^8\sqrt{\Delta}$, and color it with \cref{lem:colorWithMuchSlack} when $\Delta \leq (\log n)^{50}$ and with \cref{thm:d1LC} otherwise.
By \Cref{lem:structuralDecomposition}-(\ref{part:slack-BH}), each vertex of $B_H$ has $\Delta^{3/4} \geq U\Delta^{1/4} = \Theta( \Delta^{0.72} )$ slack in $F[B_H]$ regardless of how nodes in $S$ are colored. All cliques are uncolored at this point and each vertex in each $A_i$ has at most $U_H$ neighbors in $B_H$ by \cref{lem:structuralDecomposition}-(\ref{part:A-ext},\ref{part:AL-ext}).

\medskip\noindent
\textbf{Coloring $A_H$.} The collection of cliques in $A_H$ is colored via \Cref{lem:coloringCliques}. We can apply the lemma as their \cref{def:notbig} was maintained by earlier coloring steps (\cref{lem:sparse,lem:colorWithMuchSlack,thm:d1LC}). We emphasize that, by \cref{lem:coloringCliques}, \cref{def:notbig} of cliques in $\AL$ still holds after coloring cliques in $A_H$.

\medskip\noindent
\textbf{Coloring $B_L$.} 
We argue that $F[B_L]$ is \pious subgraph with $U = U_L = 30\Delta^{1/4}$ and color it via \Cref{lem:colorWithMuchSlack} when $\Delta \leq (\log n)^{50}$ and via \cref{thm:d1LC} otherwise. 
By \cref{lem:structuralDecomposition}-(\ref{part:AL-ext}), the external degree of cliques in $\AL$ (cliques in $A_H$ are already colored at this point) is bounded above by $U_L$. 
By \Cref{lem:structuralDecomposition}-(\ref{part:slack-BL}), each node of $B_L$ has at least $\sqrt{\Delta}/2$ slack in $F[B_L]$, which is larger than the $U_L\cdot \Delta^{0.22} = 30\Delta^{0.47}$ slack required by \cref{prop:Pi}.

\medskip\noindent
\textbf{Coloring $\AL$.} The cliques in $\AL$ are also colored via \Cref{lem:coloringCliques}, which can be done as \cref{def:notbig} holds for each of these cliques. Note that \cref{lem:coloringCliques}-(\ref{part:color-cliques-uncolored}) vacuously holds at this step because all the uncolored cliques are in $A' = \AL$.

\medskip\noindent
\paragraph{Runtime.}
The runtime is dominated by the time required to color $S$, $B_H$ and $B_L$. 
For $\Delta \leq (\log n)^{50}$, it is $\Ot{ \log\Delta \cdot \log^3\log n} = \Ot{ \log^4 \log n }$ by \cref{lem:colorWithMuchSlack,lem:sparse}, 
while for  $\Delta \geq (\log n)^{50}$, it is $O(\log^* n)$ by \cref{thm:d1LC,lem:sparse}.
\end{proof}

\TheoremDet*
\begin{proof}[Proof of \Cref{thm:det-alg}]
    If the graph $G$ is not $c$-colorable, there exists a vertex $v$ for which $\set{v} \cup N(v)$ is not $c$-colorable \cite{MR14,BamasE19}, which can be deterministically detected in $O(1)$ rounds of \LOCAL. 
    When $G$ is $c$-colorable, we apply the derandomization framework developed in \cite{GKM17,GHK18,RG20}. It states that any randomized algorithm with runtime $T(n)$ for a problem whose solution can be verified in $l(n)$ rounds can be derandomized in $O((T(n)+l(n))\cdot T_{ND}(n))$ rounds where $T_{ND}(n)$ is the runtime to compute a so-called $(O(\log n), O(\log n))$-network decomposition, see any of these works for details. Here, it is applied to the randomized algorithm of \cref{thm:col} with $T(n) = \Ot{ \log^4 \log n }$ and $l(n) = 1$ because checking if the resulting coloring is proper only requires a single round. By using the algorithm of \cite{GG24} for computing a network decomposition in $\Ot{ \log^2 n }$ rounds, we obtain \Cref{thm:det-alg}.
\end{proof}

\section{Coloring With Much Slack}
\label{sec:colorwithmuchslack}
The goal of this section is to prove \Cref{lem:colorWithMuchSlack}, that is, we aim to color an induced subgraph $H$ of $F-\bigcup_iA_i$ that is $\Pi$-ous while at the same time maintaining CC for all uncolored cliques. 
Recall that $H$ is \pious if each node has sufficient slack in $H$ compared to a given upper bound $U\geq \Delta^{1/4}$ on the external degree of uncolored cliques. 
The coloring process consists of two primary components. The first one is a random color trial for all uncolored vertices, applied for $O(\log\Delta)$ iterations, where each iteration lowers the uncolored degree of each vertex by a constant factor in expectation. The properties of the resulting partial coloring after a single color trial are summarized in \Cref{lem:ColorWithMuchSlackIteration} below. The second component summarized in \Cref{lem:MCT} below is then a multi color trial that colors all remaining vertices of the graph instance to be colored. 
  
\begin{lemma}[Random Color Trial]
\label{lem:ColorWithMuchSlackIteration}
    Let $H$ be an uncolored \pious subgraph of $F - \bigcup_i A_i$ and let $R$ be the graph induced by the remaining uncolored nodes of $H$ after running \cref{alg:rct-single-iteration}.
    Then, \whp, 
    \begin{enumerate}
        \item \label{part:rct-degree-H} $\deg_R(v) \leq \max\set{ (1-1/180) \deg_{H}(v), \Delta^{1/10} }$ for all $v \in V(R) \cup \bigcup_i A_i$, and
        \item \label{part:rct-degree-A} \cref{def:notbig} is maintained for all uncolored cliques.
    \end{enumerate}
    \cref{alg:rct-single-iteration} runs in $\Ot{ \log^3 \log n }$ rounds.
\end{lemma}
\begin{restatable}[Multi Color Trial]{lemma}{lemMCT}
\label{lem:MCT}\label{lem:mct-nb}
Let $H$ be an uncolored induced subgraph in which every vertex has $\Delta^{9/20}$ slack and each uncolored vertex in each $A_i$ has at most $\Delta^{1/10}$ neighbors in $H$.
There is a $\Ot{ \log^3 \log n}$-round \LOCAL algorithm that, \whp, colors $H$ while maintaining \cref{def:notbig} for every uncolored clique $A_i$.
When $\Delta\geq (\log n)^{50}$, the algorithm runs in $O(1)$ rounds.
\end{restatable}

Finally, these two components are combined to establish the main lemma of this section. We start the coloring with $O(\log \Delta)$ iterations of the random color trial from \Cref{lem:ColorWithMuchSlackIteration}, which decrease the maximum uncolored degree to $\Delta^{1/10}$. Then, the coloring is completed by a single iteration of multi color trial from \Cref{lem:mct-coloring}.

\LemColoringWithMuchSlack*
\begin{proof}
    We color the graph $H$ which satisfies \Cref{prop:Pi}, by iteratively applying \Cref{alg:rct-single-iteration} for $T=c\log\Delta$ times, where $c>0$ is a sufficiently large constant. We then finish the coloring of $H$ using \Cref{lem:MCT}.
    To analyze the first iterative process, let $H_j$ be the subgraph induced by the uncolored vertices of $H$ after the $j$-th call of \cref{alg:rct-single-iteration}, and let $H_0 = H$. Consider a vertex $v \in V(H) \cup \bigcup_i A_i$. According to \cref{lem:ColorWithMuchSlackIteration}, for every iteration $j$, one of the two conditions must hold: either the uncolored degree drops immediately such that $|N_{H_j}(v)| \leq \Delta^{1/10}$ or the degree drops by a constant factor such that $|N_{H_j}(v)| \leq (1 - 1/180)|N_{H_{j-1}}(v)|$. If the first condition holds for some $j$, then $v$ has fewer than $\Delta^{1/10}$ uncolored neighbors in $H$ and already satisfies the required degree bound. We therefore assume that the second condition holds for all $j$. Since the maximum degree of $H$ (as a subgraph of $F$) is at most $10^9\Delta$ (by \cref{lem:structuralDecomposition}), after $T=c\log\Delta$ iterations the uncolored degree of $v$ is at most
    \[
        10^9\Delta\cdot(1-1/180)^{c\log\Delta} 
        \leq 10^9 \exp\paren*{ \ln\Delta - \frac{c\cdot\log\Delta}{180} }
        \leq 10^9 \Delta^{1 - c/180}
    \]
    For sufficiently large constant $c$, this is at most $\Delta^{1/10}$. Now let $H'$ be the subgraph of the initial graph $H$ that contains all the yet uncolored vertices after $T$ iterations of \cref{lem:ColorWithMuchSlackIteration}. Then the graph $H'$ fulfills the preconditions of \Cref{lem:MCT}, which require that every vertex has $\Delta^{9/20}$ slack and each uncolored vertex in each $A_i$ has at most $\Delta^{1/10}$ neighbors in the remaining uncolored graph. The application of \Cref{lem:MCT} to $H'$ completes the coloring of $H$. The \Cref{def:notbig} is preserved for all uncolored cliques as it is guaranteed by the outcome of \cref{lem:ColorWithMuchSlackIteration,lem:MCT}. The runtime consists of the $O( \log\Delta )$ iterations from \Cref{lem:ColorWithMuchSlackIteration} that take $\Ot{\log^3\log n}$ rounds each and $\Ot{\log^3\log n}$ rounds for the one application of \Cref{lem:MCT}. This completes the proof. 
\end{proof}

\subsection{Iterated Random Color Trial with \cref{def:notbig} (\Cref{lem:ColorWithMuchSlackIteration})}\label{sec:rct}

The goal of this section is to prove \Cref{lem:ColorWithMuchSlackIteration}. We present a $\poly(\log\log n)$ algorithm that partially colors the graph $H$ while maintaining \cref{def:notbig} and additionally reduces the uncolored degree of vertices in $H$ and $\bigcup_iA_i$ below $\Delta^{1/10}$. 

The analysis focuses on a single iteration of Random Color Trial (\RCT). We show that in each iteration of the \RCT procedure the uncolored degree of a vertex drops by a factor of $(1-1/180)$ compared to the previous iteration while maintaining \cref{def:notbig}. We set up an LLL that, if solved correctly, yields a partial coloring guaranteeing this drop in the uncolored degree. The random process underlying our LLL is described in \cref{alg:RCT}; we call the color $\psi(v)$ the candidate color of $v$. 

\begin{algorithm}
    \caption{Random Color Trial (\RCT)\label{alg:RCT}}
    \Input{$H'$ an induced uncolored subgraph of $H$}
    \Output{A partial coloring of $H'$}
    Each vertex of $H'$ gets activated independently with probability $p_a = 1/4$\;
    Each activated vertex $v$ independently samples a uniform available color $\psi(v)$\;
    \tcp{Let $N_<(v) = \set{ u\in N(v) : ID(u) < ID(v) }$}
    All vertices $v$ with $\psi(v) \in \psi(N_{<}(v))$ discard $\psi(v)$ (i.e., set $\psi(v)$ to $\bot$)\label{line:rct-retain} 
\end{algorithm}

\noindent
In \cref{alg:RCT}, every vertex of $H'$ has one random variable that consists of a pair with a random activation bit, set to one \wp $p_a$, and the uniform available color $\psi(v)$. So we henceforth abuse notation slightly and associate vertices with their random variables.
We introduce the following bad event, for every vertex $v$ with at least $\Delta^{1/10}$, define
\begin{quote}
    $E(v)$: the uncolored degree into $H$ of $v$ is reduced by a factor less than $(1 - 1/180)$, 
\end{quote}
and for each uncolored clique $A_i$, we have
\begin{quote}
    $E_{CC}(i)$: \cref{def:notbig} is not maintained for some color w.r.t.\ $A_i$.
\end{quote}
\begin{observation}
    \label{obs:rct-random-variables}
    $\vbl(E(v)) \subseteq N_H^{\leq 2}(v)$ and $\vbl(E_{CC}(i)) \subseteq N_F^{\leq 2}(A_i) \cap V(H)$.
\end{observation}
Let us argue now that no matter the coloring of $H$, since it is \pious those bad events are rare.

\begin{lemma}\label{lem:ev}
    Let $H'$ be an induced uncolored subgraph of a \pious subgraph $H$ of $F - \bigcup_i A_i$, and $v \in V(H) \cup \bigcup_i A_i$ with $\deg_{H'}(v) \geq \Delta^{1/10}$. 
    After \Cref{alg:RCT}, $E(v)$ occurs \wp at most $\exp(-\Omega( \Delta^{1/10} ) )$.
\end{lemma}
\begin{proof}
    Let $u\in \set{v} \cup N_{H'}(v)$ with $\deg_{H'}(u) \geq \Delta^{1/10}$.
    Let $X_u$ be the random variable that counts the number of activated vertices in the neighborhood $N_{H'}(u)$ of $u$. Its expected value is given by $\E[X_u] = p_a \cdot \deg_{H'}(u)$. We show the concentration of $X_u$ by Chernoff's bound (\cref{lem:chernoff}):
    \begin{align*}
        \Prob*{ |X_u - p_a\deg_{H'}(u)| \ge p_a\deg_{H'}(u)/2 } 
\leq 2\exp\paren*{ -\frac{p_a\cdot \deg_{H'}(u)}{12} } 
        \leq \exp\paren*{ -\Omega(\Delta^{1/10}) } \ .
    \end{align*}
    where the last inequality uses the lower bound on the degree of $u$.
    Let $E^{\mathrm{all}}(v)$ be the event that every vertex $u \in \set{v} \cup N_{H'}(v)$ with $\deg_{H'}(u) \geq \Delta^{1/10}$ has between $(p_a/2)\deg_{H'}(u)$ and $2p_a\deg_{H'}(u)$ active neighbors. Since $\overline{ E^{\mathrm{all}}(v) }$ is a low-probability event, we assume $E^{\mathrm{all}}(v)$ holds.

    Let $u_1, u_2, \ldots, u_d$ be the active neighbors of $v$ ordered by increasing IDs. Let $Z_i$ be the indicator random variable equal to one iff $u_i$ gets colored by \cref{alg:RCT}. It is uniquely determined by the random variables $R_k := \psi(N_<(u_{k}) \setminus \bigcup_{j < k} N_<(u_j))$ for all $k \leq i$. We claim that
    \[
    \E\left[ Z_i \given E^{\mathrm{all}}(v), R_1, R_2, \ldots, R_{i-1} \right] = \mu_i
    \geq 1/2 \ .
    \]
    There are two cases.
    If $\deg_{H'}(u) < \Delta^{1/10}$, then since $H'$ is \pious, $u$ has at least $\Delta^{1/4}$ available colors. Hence, it gets colored with probability at least $1 - \Delta^{-3/20} \geq 1/2$.
    If $\deg_{H'}(u) \ge \Delta^{1/10}$, then under $E^{\mathrm{all}}(v)$, it has at most $2p_a\deg_{H'}(u) = \deg_{H'}(u)/2$ active neighbors. Since $H'$ is \pious, $u$ has at least as many available colors as uncolored neighbors, and thus retains its color with probability at least $1/2$ by the union bound.
    
    Let $Z = \sum_i Z_i$ be the number of colored active neighbors of $v$ after \cref{alg:RCT}. By \Cref{lem:BEPS} with $\mu = \sum_i \mu_i \geq d/2$ for all $i\in[d]$ and $t = d/4$, we get that 
    \[  
        \Prob{ Z < d/4 \given E^{\mathrm{all}}(v) }
\leq \exp\paren*{ -\frac{d}{32} }
        \leq \exp\paren*{ -\Omega( \Delta^{1/10} ) } \ ,
\]
    where the last inequality uses that $d \geq (p_a/2)\deg_{H'}(v)$ under $E^{\mathrm{all}}(v)$ and that $\deg_{H'}(v) \geq \Delta^{1/10}$. Note that when $Z \geq d/4$, the uncolored degree of $v$ after \cref{alg:RCT} is $\deg_{H'}(v) - Z \leq (1 - p_a/8)\deg_{H'}(v)$. So, overall, we have that
    \[
    \Prob{ E(v) } 
    \leq \Prob{ \overline{ E^{\mathrm{all}}(v) } }
    + \Prob{ Z < d/4 \given E^{\mathrm{all}}(v) }
    \leq \exp\paren*{ -\Omega( \Delta^{1/10} ) }
    \ . \qedhere
    \]
\end{proof}

\begin{lemma}
    \label{lem:rct-cc-post-prob}
    Let $H'$ be an induced uncolored subgraph of a \pious subgraph $H$ of $F - \bigcup_i A_i$.
    For every uncolored clique $A_i$, after \cref{alg:RCT}, $E_{CC}(i)$ holds \wp at most $\Delta \cdot \exp(-\Delta^{1/40})$.
    \label{lem:cc-bad-event-for-rct}
\end{lemma}
\begin{proof}
    Consider any uncolored clique $A_i$ and any color $x$. The relevant collection of subsets is formed by the external neighborhoods of each vertex $v \in A_i$. By \Cref{lem:structuralDecomposition}-\ref{part:A-size}, we have $|A_i| < \Delta$. We ignore vertices of $\All_i \cup \Big_i^+$ for their contribution is not accounted for in \cref{def:notbig}. Hence, no relevant vertex of $H$ belongs to more than $2\Delta^{9/10}$ such subsets. Since $H'$ is \pious, each vertex $v$ of $H'$ picks color $x$ as candidate color \wp at most $1/(\deg_{H'}(v)+ U \cdot \Delta^{0.22})) \leq 1/(Q\cdot \Delta^{1/5})$ for $Q = U$. As candidate color choices are independent for different vertices, the probability that $\ell$ vertices pick color $x$ is at most $1/(Q\cdot \Delta^{1/5})^{\ell}$. Therefore, by \Cref{lem:bnd-notbig}, \cref{def:notbig} fails for clique $A_i$ and color $x$ \wp at most $\exp(-\Delta^{1/40})$. The lemma follows by the union bound over a total of at most $\Delta$ different colors. 
\end{proof}

Given this, \cref{alg:rct-single-iteration} follows our shattering framework (see \cref{sec:our-shattering-framework}) to produce a coloring in which degree in $H$ have dropped.

\begin{algorithm}[H]
    \caption{Single iteration of degree reduction in $H$ \label{alg:rct-single-iteration}}

    \Input{An uncolored \pious graph $H$}
    
    \Output{A proper partial list-coloring of $H$}
    
    Run \cref{alg:RCT} on $H$ and let $\mathcal{A}'$ be the set of events $E(v)$ or $E_{CC}(i)$ that occur.
    \label{line:it-rct-run-rct}\;
    
    Retract the colors of (i.e., uncolor) the vertices in $\vbl( \mathcal{A}' )$ \label{line:retraction}\;

    Solve the following LLL using the deterministic algorithm from \cref{thm:deterministicLLLLOCAL}: \label{line:rct-single-it-post-LLL}\;
    \nonl \hspace{0.3cm} \textbf{Random process: }Run \cref{alg:RCT} on $H'$ where $H'$ is the induced subgraph of $H$ of uncolored vertices $v$ with $\dist_F(v, \vbl( \mathcal{A}' ) ) \leq 2$\;
    \nonl \hspace{0.3cm} \textbf{Bad events: } 
    \begin{compactenum}
        \item[] $E(v)$ for all vertices $v \in N( \vbl( \mathcal{A}' ) ) \cup \vbl( \mathcal{A}' )$ with $\deg_{H'}(v) \geq \Delta^{1/10}$, and
        \item[] $E_{CC}(i)$ for all uncolored clique $A_i$ adjacent to a vertex in $H'$.
    \end{compactenum}

    Output the partial list-coloring $\phi$ where $\phi(v)$ is the color given in \cref{line:it-rct-run-rct} for vertices $v\notin H'$, and $\phi(v)$ is the color given in \cref{line:rct-single-it-post-LLL} for the vertices of $v\in H'$.
\end{algorithm}

\begin{remark}As detailed in \Cref{def:notbig}, there is a fresh budget of $CC(i,x)<\Delta^{\frac{37}{40}}$ for each color $x$, each clique $A_i$, and also for each coloring step of the algorithm. In particular, in \Cref{line:it-rct-run-rct,line:rct-single-it-post-LLL} of \Cref{alg:rct-single-iteration} the events $E_{CC}(i)$ get a separate budget of $\Delta^{\frac{37}{40}}$ each. 
\end{remark}

\begin{lemma}\label{lem:rtc_post_prob}
    Suppose $\Delta \leq (\log n)^{50}$.
    Let $\Gev$ be the dependency graph of the LLL of \cref{line:rct-single-it-post-LLL} in \cref{alg:rct-single-iteration}. With high probability, its largest connected component has size at most $\poly(\log n)$.
\end{lemma}

\begin{proof}
    In order to bound the size of the largest connected component of $\Gev$ we use \cref{lem:our-shattering}; next we argue that the lemma applies. We take the whole graph $F$ as graph $G$ from \cref{lem:our-shattering}, because every vertex $v$ of $H$ has one variable encoding its color choice.  The LLLs (formally sets of sets) $\mathcal{A}$ and $\mathcal{B}$ required to use the lemma are both given by all collections of $\vbl(E(v))$ and $\vbl(E_{CC}(i))$ for all $v\in V(H)$ and uncolored cliques $A_i\in A_L\cup A_H$. The set $\mathcal{A}'$ consists of the events $E(v)$ and $E_{CC}(i)$ that hold after \cref{line:it-rct-run-rct} of \cref{alg:rct-single-iteration}. Let $\mathcal{B}'$ with $c_{\mathcal{B}'} = 3$ be as in the statement of \Cref{lem:our-shattering}. Let $c_1=5, c_2=2$, and $c_4=3c_1(4\cstB+16)+c_2+\cstB+2$.
    We show that the preconditions \Cref{part:our-shattering-diam,part:our-shattering-involved,part:our-shattering-pre-prob,part:our-shattering-vbl} of \Cref{lem:our-shattering} hold. 

    \Cref{part:our-shattering-diam} holds with $c_1=5$ due to \cref{obs:rct-random-variables} (and because $F$ includes the cliques): a vertex $v$ holds variables used only by events $E(u)$ with $u\in N_H^{\leq 2}(v)$ (which has diameter at most 4) or events $E_{CC}(i)$ with $A_i$ with $\dist_G(v, A_i) \leq 2$ (which has diameter at most 5 in $F$). \Cref{part:our-shattering-involved} requires a bound $\Delta(F)^{c_2}$ for some constant $c_2$ on the number of events in which each variable appears and it holds with $c_2=2$ as each variable is only used by events with variables in its $2$-hop neighborhood.  \Cref{part:our-shattering-vbl} requires that for each event $A\in \mathcal{A}$ we can determine $1(A\in \mathcal{A}')$ by only evaluating the variables in $\vbl(A)$ which immediately holds by the definition of $\mathcal{A}'$. \Cref{part:our-shattering-pre-prob} requires an upper bound of $1/\Delta(F)^{c_4}$ on the probability for each bad event $A\in \mathcal{A}$ to be contained in $\mathcal{A}'$ where $c_1$ is an arbitrary constant satisfying $c_4>3c_1(4\cstB+16)+c_2+\cstB+1$. By \cref{lem:ev,lem:rct-cc-post-prob} this probability is at most $\Delta(G)\cdot \exp(-\Delta(G)^{1/40})$. As $\Delta(G)=\Theta(\Delta(F))$  \cref{part:our-shattering-pre-prob} holds for $c_4$ (as chosen above) as long as $\Delta(G)$ is larger than a sufficiently large absolute constant $\Delta_0$. The set of events included in the LLL of \cref{line:rct-single-it-post-LLL} are all included in $\mathcal{B}'$ (with $\cstB$ defined above) as vertices with distance at most $2$ to $\vbl(\mathcal{A}')$ participate in the post-shattering random process and any event $E(v)$ or $E_{CC}(i)$ adjacent to such a vertex participates in the post-shattering as well; note that the event $E_{CC}(i)$ for a clique $A_i$ that has the closest vertex $w$ in $H'$ in distance $2$ does not need to participate despite $w\in \vbl(E_{CC}(i))$ as the event  is avoided regardless the color choice of $w$.
    Hence by \cref{lem:our-shattering}, the connected components of $\Gev$ have size at most $O(\Delta(F)^{4c_1\cstB + 4c_1 + c_2}\log n) =\poly(\log n)$ with high probability.
\end{proof}
\begin{remark}
    Determining and verifying the minimal admissible value of the constant $\cstB$ in the above proof is somewhat tedious. The preconditions required for the shattering argument in \Cref{lem:our-shattering} remain valid even for larger choices of $\cstB$, since the probability that any given event is contained in $\mathcal{A}'$ decreases exponentially in $\Delta$. Consequently, \Cref{part:our-shattering-pre-prob} holds for any constant $c_4$ satisfying
    $c_4 > 3c_1(4\cstB + 16) + c_2 + \cstB + 1$,
    provided that $\Delta$ is larger than a sufficiently large constant $\Delta_0 = \Delta_0(\cstB)$.
    
    In the subsequent sections, we use several analogous applications of \Cref{lem:our-shattering}. For the sake of clarity and ease of verification, we do not always optimize the choice of $\cstB$, but instead adopt values that simplify the argument.
    
    Finally, we note that our retraction algorithm is deliberately somewhat aggressive. For instance, if $E_{CC}(i)$ holds for a clique $A_i$, we retract all color choices of vertices in $\vbl(E_{CC}(i))$, including those at distance two from $A_i$, even though it would suffice to retract only the color choices of vertices in $N(A_i)$. This uniform retraction rule helps streamline what is otherwise an already intricate procedure.
\end{remark}

\begin{proof}[Proof of \Cref{lem:ColorWithMuchSlackIteration}]
    If $\Delta \geq (\log n)^{50}$, then by \cref{lem:ev,lem:rct-cc-post-prob}, none of the bad events occur at \cref{line:it-rct-run-rct} of \cref{alg:rct-single-iteration}. Hence, (1) and (2) of \Cref{lem:ColorWithMuchSlackIteration} already hold after \cref{line:it-rct-run-rct} (i.e., $\mathcal{A}' = \emptyset$) and we skip the post-shattering phase. We henceforth assume that $\Delta \leq (\log n)^{50}$.

    For a vertex $v \in N_H(\vbl(\mathcal{A}'))$, note that each of its neighbors belong to $N_H^{\leq 2}( \vbl(\mathcal{A}') )$, hence every uncolored neighbor of $v$ in $H'$ samples a random color in the run of \cref{alg:RCT} in \cref{line:rct-single-it-post-LLL}. 
    By \cref{lem:ev,lem:rct-cc-post-prob}, after coloring some of the vertices of $H$ during \cref{line:it-rct-run-rct}, the set of events described in \cref{line:rct-single-it-post-LLL} of \cref{alg:rct-single-iteration} indeed forms an LLL whose dependency degree is bounded by $\poly \Delta$ via \Cref{obs:rct-random-variables}, and thus an assignment of colors for vertices of $\vbl( \mathcal{A}' )$ that avoids its bad events can be found by \cref{thm:deterministicLLLLOCAL}.
    Running \cref{alg:RCT} takes $O(1)$ rounds, and by \cref{obs:rct-random-variables} so does \cref{line:retraction}. By \cref{lem:rtc_post_prob}, the connected components of the dependency graph of the LLL solved in \cref{line:rct-single-it-post-LLL} have size at most $N = \poly(\log n)$. Hence, \cref{thm:deterministicLLLLOCAL} runs in $\Ot{ \log^3 N } = \Ot{ \log^3 \log n }$.

    Let us now conclude by verifying that both properties claimed by \cref{lem:ColorWithMuchSlackIteration} hold after nodes adopt colors as specified in the last line of \Cref{alg:rct-single-iteration}. If $v$ has fewer than $\Delta^{1/10}$ uncolored neighbors, then (1) vacuously holds, so let us assume the contrary. Consider a vertex $v \notin N( \vbl(\mathcal{A}') ) \cup \vbl( \mathcal{A}' )$. Then after \cref{line:it-rct-run-rct}, the degree of $v$ has decreased by a factor of $(1 - 1/180)$ or dropped below $\Delta^{1/10}$ (otherwise $E(v)$ holds and $v\in N( \vbl(\mathcal{A}') ) \cup \vbl( \mathcal{A}' )$) and none of the colors in $N(v)$ are retracted (otherwise $v$ has a neighbor in $\vbl(\mathcal{A}')$).
    Consider now a vertex $v\in N( \vbl( \mathcal{A}' ) ) \cup \vbl( \mathcal{A}' )$. 
    The post-shattering LLL (described \cref{line:rct-single-it-post-LLL}) includes the event $E(v)$, so the uncolored degree of $v$ in $H'$ drops by a $(1 - 1/180)$ factor, hence it at most $(1 - 1/180)\deg_H(v)$. Since $H'$ contains all the uncolored neighbors of $v$, the uncolored degree of $v$ at the end of the algorithm is upper bounded by the uncolored degree of $v$ in $H'$.

    After retracting colors in $\vbl(\mathcal{A}')$, none of the $E_{CC}(i)$ holds anymore and thus \cref{def:notbig} is maintained by the pre-shattering phase. During post-shattering, only vertices of $N_H^{\leq 2}( \vbl( \mathcal{A}' ) )$ get colored. Hence, \Cref{def:notbig} continues to hold for the $A_i$ that are not adjacent to some vertex in $N_H^{\leq 2}( \vbl( \mathcal{A}' ) )$. If $A_i$ is adjacent to such a vertex, then CC is maintained by the post-shattering coloring (with a fresh budget) because $E_{CC}(i)$ would otherwise hold.
\end{proof}

\subsection{Finishing Off The Coloring via MCT (Lemma~\ref{lem:MCT})}
In this section, we analyze the subsequent coloring step after computing the partial coloring in \Cref{sec:rct}. This step of computing the complete coloring of the considered vertices involves a single-round \MCT (equivalent to Lemma 35 of \cite{MR14}). In short, each vertex picks a set of colors and compares them to the sets of its neighbors. If there is a color in this set that was picked by none of the neighbors, then the vertex retains that color. See \cref{alg:MCT}.
\begin{algorithm}
    \caption{MultiColorTrial \label{alg:MCT}}
    \Input{An uncolored graph $H'$ and a parameter $T$}
    \Output{A partial coloring of $H'$}
    Every vertex $v\in V(H')$ samples a set $S(v)$ of $T$ uniform available colors with repetitions\;
    \lIf*{$\exists \chi \in S(v) \setminus S(N(v))$}{Color $v$ with $\chi$
    }\lElse{$v$ remains uncolored}
\end{algorithm}
We introduce the following events for every uncolored vertex $v \in V(H)$ and every uncolored clique $A_i\in A_H\cup A_L$:
\begin{quote}
    $E'(v)$: $v$ does not get colored, 

    $E_{CC}(i)$: \cref{def:notbig} is not maintained for some color w.r.t.\ $A_i$.
\end{quote}

\begin{observation}
    \label{obs:mct-random-variables}
    $\vbl(E'(v)) \subseteq N_H^{\leq 2}(v)$ and $\vbl(E_{CC}(i)) \subseteq N_F^{\leq 2}(A_i) \cap V(H)$.
\end{observation}

Similarly to \cref{sec:rct}, we argue that under the conditions of \cref{lem:MCT}, both events are unlikely to occur when we run \cref{alg:MCT}.

\begin{lemma}
    \label{lem:mct-coloring}
    Let $H$ be an uncolored graph with maximum degree at most $\Delta^{1/10}$. For every $v\in V(H)$ with $\Delta^{9/20}$ slack in $H$, after \cref{alg:MCT} with $T = \Delta^{1/10}$, the event $E'(v)$ holds \wp at most $\exp(-\Delta^{1/10})$.
\end{lemma}
\begin{proof}
    Since each vertex has at most $\Delta^{1/10}$ uncolored neighbors, there are also at most $T\Delta^{1/10} = \Delta^{1/5}$ different candidate colors appearing in its neighborhood. Fix $S(N(v))$ arbitrarily. Each of the $T$ colors that $v$ picks uniformly at random from its list of available colors belongs to this set \wp at most $\Delta^{1/5}/|L(v)|$. Since $L(v)$ contains at least $\Delta^{9/20}$ colors, the probability that $v$ is uncolored after \Cref{alg:MCT} is at most $\left(\Delta^{1/5}/|L(v)|\right)^{T}\leq \Delta^{-T/4}\leq \exp(-\Delta^{1/10})$.
\end{proof}

\begin{lemma}
    \label{lem:mct-cc}
    Let $H$ be an uncolored subgraph of $F - \bigcup_i A_i$ such that every uncolored $v\in A_i$ has at most $\Delta^{1/10}$ neighbors in $H$.
    For every uncolored clique $A_i$, after \cref{alg:MCT}, $E_{CC}(i)$ holds \wp at most $\Delta \cdot \exp(-\Delta^{1/40})$.
\end{lemma}

\begin{proof}
    We again apply \Cref{lem:bnd-notbig} to upper bound the probability.  Fix a clique $A_i$ and a color and form the sets consisting of external neighbors not in $\Big_i^+$ as before.  Set $Q=\Delta^{1/10}$. The probability that a specific color is picked by neighbor is bounded above by $T/|L(v)|=1/\Delta^{7/20}\leq 1/(Q\cdot\Delta^{1/5})$. Due to the independence of picking candidate colors between different vertices the probability that a specific color is picked by $\ell$ vertices is at most $1/(Q\cdot\Delta^{1/5})^{\ell}$. No vertex lies in more than $2\Delta^{9/10}$ sets, because in that case it would belong to some $\Big_i^+$. Now we can apply \Cref{lem:bnd-notbig} and get that for a fixed clique and color the probability that CC fails is at most $\exp(-\Delta^{1/40})$. With a union bound over all colors, $E_{CC}(i)$ holds \wp at most $\Delta\cdot \exp(-\Delta^{1/40})$.
\end{proof}

\begin{algorithm}
    \caption{\MCT Coloring}\label{alg:mctcoloring}
    \nonl\textbf{Pre-shattering:}\;
    Run \cref{alg:MCT} with $T = \Delta^{1/10}$.\;
    Let $\mathcal{A}'$ be the set of events $E'(v)$ or $E_{CC}(i)$ that occur.
    \label{line:MCT-retractionBase}\;
    Retract the colors of (i.e., uncolor) vertices in $\vbl(\mathcal{A}')$ \label{line:MCT-retraction}\;
    
    \BlankLine
    \nonl\textbf{Post-shattering:}\; 
    Solve the following LLL using the deterministic algorithm from \cref{thm:deterministicLLLLOCAL}:\; \label{line:MCTPostShattering}
    \nonl\hspace{0.3cm} \textbf{Random process:} Run \cref{alg:MCT} with $T = \Delta^{1/10}$ on $\vbl(\mathcal{A}')$\;
    \nonl\hspace{0.3cm} \textbf{Bad events: } 
    \begin{compactenum}
        \item[] $E'(v)$ for every uncolored vertex $v$, and
        \item[] $E_{CC}(i)$ for each uncolored clique $A_i$ adjacent to a vertex in $\vbl(\mathcal{A}')$
    \end{compactenum}
\end{algorithm}

\lemMCT*
\begin{proof}
    If $\Delta \geq (\log n)^{50}$, then by \cref{lem:mct-coloring,lem:mct-cc}, none of the bad events occur after running \cref{alg:MCT} in pre-shattering, \whp. Hence, $H$ is colored and CC is maintained \whp. We henceforth assume that $\Delta \leq (\log n)^{50}$.

    Let $\Gev$ be the dependency graph of the $E'(v)$ and $E_{CC}(i)$ and $B$ be the set of events of the post-shattering LLL in \Cref{line:MCTPostShattering} of \Cref{alg:mctcoloring}. We first reason how to apply \Cref{lem:our-shattering} to bound the size of each connected component of $\Gev$. We take the whole graph $F\supseteq H$ as graph $G$ from \cref{lem:our-shattering}, because every vertex of $H$ has one variable in \cref{alg:rct-single-iteration}. The LLLs (formally sets of sets) $\mathcal{A}$ and $\mathcal{B}$ required to use the lemma are both given by all collections of $\vbl(E'(v))$ and $\vbl(E_{CC}(i))$ for all $v\in V(H)$ and uncolored cliques $A_i\in A_L\cup A_H$. The set $\mathcal{A}'$ consists of the events $E'(v)$ and $E_{CC}(i)$ that hold after \cref{line:MCT-retractionBase} of \cref{alg:mctcoloring}. Let $\mathcal{B}'$ with $\cstB= 0$ be as in the statement of \Cref{lem:our-shattering}. Let $c_1=5, c_2=2$,  and $c_4=3c_1(4\cstB+16)+c_2+\cstB+2$.     
    We show that the preconditions 
   \Cref{part:our-shattering-diam,part:our-shattering-involved,part:our-shattering-pre-prob,part:our-shattering-vbl} of \Cref{lem:our-shattering} hold. 

    \Cref{part:our-shattering-diam} holds via \Cref{obs:mct-random-variables} and we also obtain that each variable is used by at most $\Delta(F)^{c_2}$ events yielding \Cref{part:our-shattering-involved}. \Cref{part:our-shattering-vbl} follows by the definition of $\mathcal{A}'$.   In order to prove \cref{part:our-shattering-pre-prob}, observe that by \cref{lem:mct-coloring,lem:mct-cc} the probability of each $A\in \mathcal{A}$ to be contained in $\mathcal{A}'$, i.e., $E'(v)$ or $E_{CC}(i)$ hold, is upper bounded by $\exp\paren{-\Omega(\Delta(G)^{1/10})}$ and $\Delta(G)\cdot \exp(-\Delta(G)^{1/40})$ respectively. Both of these values decrease exponentially in $\Delta(G)=\Theta(\Delta(F))$ and hence for $\Delta(G)\geq \Delta_0$ for an absolute constant $\Delta_0$ can be upper bounded by $1/\Delta(F)^{c_4}$ where $c_4$ is stated above. The set of events included in the LLL of \cref{line:MCTPostShattering} of \Cref{alg:mctcoloring} are all included in $\mathcal{B}'$ (with the value of $\cstB$ defined above). Hence by \cref{lem:our-shattering}, the connected components of $\Gev$ have size at most $N=O(\Delta(F)^{4c_1\cstB + 4c_1 + c_2}\log n) =\poly(\log n)$ with high probability.
    
    Since the pre-shattering only colors vertices and that the slack does not decrease, \cref{lem:mct-coloring} continues to apply to bound the probability of $E'(v)$ in the post-shattering LLL. Likewise, \cref{lem:mct-cc} applies to bound the probability for $E_{CC}(i)$ in the post-shattering LLL. \Cref{obs:mct-random-variables} bounds the dependency degree and hence, for sufficiently large $\Delta$, \Cref{line:MCTPostShattering} poses an LLL that we can solve with \Cref{thm:deterministicLLLLOCAL} in $\Ot{\log^3 N}=\Ot{\log^3\log n}$ rounds; all other parts of the algorithm are executed in $O(1)$ rounds. At the end, all nodes of $H$ are colored as all events $E'(v)$ are either avoided already after the pre-shattering phase or after the post-shattering phase; note that any uncolored node after the pre-shattering phase is contained in $\vbl(\mathcal{A}')$. After retracting colors in $\vbl(\mathcal{A}')$, none of the $E_{CC}(i)$ holds anymore and thus also \cref{def:notbig} is maintained by the pre-shattering phase. During post-shattering, only vertices of $\vbl( \mathcal{A}')$ get colored. Hence, \Cref{def:notbig} continues to hold for the $A_i$ that are not adjacent to some vertex in $\vbl( \mathcal{A}')$. If $A_i$ is adjacent to such a vertex, then \Cref{def:notbig} is maintained by the post-shattering coloring (with a fresh budget) as it is implied by avoiding $E_{CC}(i)$.
      
\end{proof}

\section{Coloring the Sparse Nodes}
\label{sec:sparse}
In this section,  we show how to color the sparse nodes $S$ while maintaining $CC$. Recall that this is the first step of \Cref{alg:highlevel}.

\LemColorSparse*

\subsection{Algorithm for Sparse Nodes}
The first step is similar to \cite{MR14}: to generate \emph{slack} for the nodes of $S$ in the form of sufficient color reuse.
Initially, every vertex in $S$ tries a random color, resulting in about $\sqrt{\Delta}$ colors repeated in each neighborhood. A color is repeated in $N(v)$ if at least two vertices have it. When $\Delta \geq (\log n)^{50}$, every vertex gains $(1 + \Omega(1))\sqrt{\Delta}$ slack \whp.
When $\Delta \leq \log^{50} n$, our algorithm differs significantly from \cite{MR14}, with an involved shattering argument; see \Cref{lem:alg-slackgen} and \Cref{sec:slackGeneration}  for details. 
We henceforth call $S' \subseteq S$ the set of vertices uncolored after slack generation.

If the remaining subgraph $F[S']$ satisfied property $\Pi$, we could complete the task with algorithm \alg{ColorWithMuchSlack}.
However, though the slack obtained is significant, it can still be small relative to the external degree of cliques
(see \cref{lem:structuralDecomposition}(\ref{part:A-ext})).

The solution is to 
color the nodes of $F[S']$ in two batches. We split $S'$ into subgraphs $S_1$ and $S_2$ via a \emph{degree-splitting} LLL that guarantees that degrees into $S_1$ and $S_2$ are similar to what a random split would provide.
In \Cref{lem:whittling}, we show $F[S_1]$ is $\pious$ (thanks to the uncolored neighbors in $S_2$), and that $F[S_2]$ is $\pious$ regardless of the coloring in $S_1$ (as long as it extends the coloring produced by slack generation).
We can then color $S'$ by 
using $\alg{ColorWithMuchSlack}$ twice: first to color $S_1$ and then $S_2$. 

\begin{algorithm}
    \caption{\alg{ColorSparse} \label{alg:tough-cookie}} 

    \alg{SlackGeneration} \label{line:color-sparse-slackgen} \;
        
Partition $S'$ into sets $S_1$ and $S_2$ using the degree splitting algorithm of \cref{lem:ds} when $\Delta \leq (\log n)^{50}$ and by sampling every vertex of $S'$ into $S_2$ \wp $p = 2\Delta^{-1/4}$ otherwise 
\label{line:whittling} \;

    \alg{ColorWithMuchSlack}($S_1$) \;
    \alg{ColorWithMuchSlack}($S_2$)
\end{algorithm}

\noindent
We begin by stating the properties of the \alg{SlackGeneration} step.
We defer the proof of \cref{lem:alg-slackgen} to \cref{sec:slackGeneration} to preserve the flow of the paper.
\begin{restatable}{lemma}{LemShatteringSlackGen}
    \label{lem:alg-slackgen}
There is a randomized \LOCAL algorithm that computes in $\Ot{ \log^3 \log n }$ rounds a partial coloring of the subgraph $F$ induced by $S$ such that:
\begin{enumerate}[label=\emph{(\alph*)}]
    \item \label[part]{part:slackgen-repeated-color}
     For every $v \in S$, either $v$ has fewer than $\Delta - 3\sqrt{\Delta}$ neighbors in $S$ or at least $1.05\sqrt{\Delta}$ colors appear at least twice in $N(v) \cap S$;   
    \item \cref{def:notbig} is maintained; and
    \item Every vertex in $S$ has at most $19\Delta/20$ colored neighbors.
\end{enumerate}
If $\Delta \geq (\log n)^{50}$, then the algorithm ends after $O(1)$ round.
\label{lem:slackgen}
\end{restatable}

\noindent
For the rest of this section, $L(v)$ denotes the set of colors available after slack generation (\cref{line:color-sparse-slackgen} in \cref{alg:tough-cookie}). The main implication of \cref{lem:slackgen} is the following lower bound on the number of available colors.  We emphasize that for \cref{obs:listsize}, it is crucial that the number of repeated colors guaranteed by \cref{lem:slackgen}(a) is $(1 + \Omega(1))\sqrt{\Delta}$.

\begin{observation}
    All $v\in S'$ have $|L(v)| \geq \deg_{S'}(v) + 0.05\sqrt{\Delta}$.
    \label{obs:listsize}
\end{observation}
\begin{proof}
  There are $c \geq \Delta - k_\Delta + 1 \ge \Delta - \sqrt{\Delta}+1$ colors. 
  If $\deg_S(v) < \Delta - 3\sqrt{\Delta}$, then $c \geq \deg_S(v) + 2\sqrt{\Delta}$ and the claim follows as each color lost in $L(v)$ corresponds to a colored neighbor (thus not in $S'$). If $\deg_S(v) \geq \Delta - 3\sqrt{\Delta}$, then by \cref{lem:slackgen}(a), it has at least $1.05\sqrt{\Delta}$ repeated colors in $N(v)$. Simple accounting shows that $L(v)$ contains at least 
  \[
  c - \deg_{S \setminus S'}(v) + 1.05\sqrt{\Delta} 
  \geq \Delta - \deg_{S \setminus S'}(v) + 0.05\sqrt{\Delta}
  \geq \deg_{S'}(v) + 0.05\sqrt{\Delta} \ ,
  \]
  where the first inequality uses the definition of $c$ and the latter uses that $F[S]$ has maximum degree $\Delta$.
\end{proof}

\medskip\noindent
The following \emph{degree splitting} result is used as a subroutine to implement line 3 of \cref{alg:tough-cookie}. Similar results were given in \cite{HMN22}, but not the exact claim we need.

\begin{restatable}{lemma}{LemDegreeSplitting}
    There is a universal constant $\alpha > 0$ for which the following holds.
Let $p \in (0,1)$, let $H$ be a graph with maximum degree $\Delta \leq (\log n)^{50}$ and $S' \subseteq V(H)$.
    There is a $\Ot{\log^3\log n}$-round algorithm that partitions $S'$ into $S_1$ and $S_2$ such that, \whp, every vertex $v\in V(H)$ with $\deg_{S'}(v) \ge \alpha p^{-1}\log \Delta$ has
    \begin{enumerate}
        \item at most $4p\deg_{S'}(v)$ neighbors in $S_2$, and
        \item at least $p\deg_{S'}(v)/2$ neighbors in $S_2$. 
    \end{enumerate}
    \label{lem:ds}
\end{restatable}

We emphasize that the guarantee of \cref{lem:ds} applies to all vertices of $F$ with sufficiently many neighbors in $S'$. In particular, it splits the degree of vertices in $A_i$ towards $S'$ as well.
The proof is deferred to \cref{sec:proof-ds} and follows the shattering framework. 

\begin{lemma}
\label{lem:whittling}
There is a $\Ot{ \log^3 \log n }$ round \LOCAL algorithm that partitions $S'$ into sets $S_1$ and $S_2$ such that 
\begin{enumerate}
    \item $F[S_1]$ is $\pious$, and
    \item $F[S_2]$ is $\pious$ regardless of how the coloring is extended to the vertices of $S_1$.
\end{enumerate}
When $\Delta \geq (\log n)^{50}$, the algorithm ends after $O(1)$ rounds.
\end{lemma}
\begin{proof}
    Let $p = 2\Delta^{-1/4}$.
    When $\Delta \geq (\log n)^{50}$, we sample vertices of $S'$ into $S_2$ with probability $p$, by the Chernoff Bound and Union Bound, the guarantees of \cref{lem:ds}~(1,2) hold with high probability.
    When $\Delta \leq (\log n)^{50}$, we split the vertices of $S'$ into $S_1$ and $S_2$ using the algorithm of \cref{lem:ds} with $H = F$, $S'$ as itself, and $p = 2\Delta^{-1/4}$. We henceforth focus on proving that (1) and (2) hold.

    Let $U_1=10^8\cdot \sqrt{\Delta}$, and $U_2 = \alpha \Delta^{1/4}\log\Delta$ where $\alpha$ is the constant from \Cref{lem:ds}.
    We claim that $F[S_1]$ is \pious with $U = U_1$ and $F[S_2]$ is \pious with $U = U_2$, even after we extend the coloring to $S_1$.
    
    \medskip\noindent
    \textbf{Proof of 1.}  
    A node with $\deg_{S'}(v)<\Delta/40$, has a list of size $\Omega( \Delta )$. Indeed, it has at most $19\Delta/20$ colored neighbors (by \cref{lem:slackgen}(c)) thus $|L(v)| \geq c - 19\Delta/20 \geq \Delta/20 - k_\Delta \geq \deg_{S'}(v) + \Delta/40 - k_\Delta \geq \deg_{S'}(v) + 4\Delta^{3/4} \geq \deg_{S'}(v) + 4U_1\Delta^{0.22}$. A node with $\deg_{S'}(v)\geq \Delta/40$ has at least $\deg_{S'}(v)/\Delta^{1/4} \geq \Delta^{3/4}/40$ neighbors in $S_2$ by \Cref{lem:ds}(2). 
    Thus, we obtain by \Cref{obs:listsize} that
    \begin{align*}
        |L(v)| 
        \geq \deg_{S'}(v)
        = \deg_{S_1}(v) + \deg_{S_2}(v)
        \geq \deg_{S_1}(v) + \Delta^{3/4}/40.
    \end{align*}
    For sufficiently large $\Delta$, it holds that $\Delta^{3/4}/40 > U_1\cdot\Delta^{0.22}$, thereby proving Part (b) of \cref{prop:Pi}.
    As for Part (a), recall that a vertex in $A$ has at most $U_1 = 10^8\sqrt{\Delta}$ neighbors in $S$ (by \cref{lem:structuralDecomposition}(\ref{part:A-ext})). 
    
    \medskip\noindent
    \textbf{Proof of 2.}
    Consider any extension of the coloring to $S_1$ and let $L_2(v)$ denote the list of colors still available to $v\in S_2$.
    By \Cref{obs:listsize}, after coloring the nodes in $S_1$ in an arbitrary manner, the list size of a node $v\in V(S_2)$ is at least 
    \begin{align*}
        |L_2(v)|
        \geq|L(v)|-\deg_{S_1}(v)
        \geq \deg_{S_2}(v)+0.05\sqrt{\Delta}
        \geq \deg_{S_2}(v)+U_2\cdot \Delta^{0.22}~, 
    \end{align*}
    where the last inequality holds as for sufficiently large $\Delta$. Part (a) of \cref{prop:Pi} trivially holds for a vertex $u\in \bigcup_i A_i$ with $\deg_{S'}(u) < \alpha\Delta^{1/4}\log\Delta = U_2$. Otherwise, the splitting lemma applies to $u$ and since $\deg_{S'}(u) \leq 10^8 \sqrt{\Delta}$, it has at most $4p\deg_{S'}(u)\leq 8\deg_{S'}(u)/\Delta^{1/4} \leq 8\cdot 10^8\Delta^{1/4}< U_2$ neighbors in $S_2$.
\end{proof}

\medskip\noindent
We conclude by putting all the statements together to prove the main result about coloring $S$.

\begin{proof}[Proof of \cref{lem:tough-cookie}]
By \cref{lem:whittling}, we can use \cref{lem:colorWithMuchSlack} to color $F[S_1]$ and $F[S_2]$. It takes $\Ot{ \log\Delta \cdot \log^3\log n }$ rounds when $\Delta \leq (\log n)^{50}$. When $\Delta \geq (\log n)^{50}$, we use \cref{thm:d1LC} instead to color in $O(\log^* n)$ rounds.
\end{proof}

\subsection{Degree Splitting (Proof of \texorpdfstring{\cref{lem:ds}}{Lemma~\ref{lem:ds}})}
\label{sec:proof-ds}

We partition $S'$ into $S_1$ and $S_2$ with \cref{alg:ds} and argue that the claims of \cref{lem:ds} are verified with high probability. The proof follows the shattering framework: first, we analyze the probability of bad events occurring during pre-shattering, deduce that the connected components of the dependency graph are $\poly(\log n)$-sized in post-shattering, and finally verify that the post-shattering instance is indeed an LLL. 

\begin{algorithm}
    \caption{Degree Splitting (\cref{lem:ds})}
    \Input{a graph $H$ with maximum degree $\Delta$, and a set $S' \subseteq V(H)$}
    \Output{a partition $S_1, S_2$ of $S'$ \label{alg:ds}}

    \BlankLine
    \nonl\textbf{Pre-Shattering}\;
    Each vertex of $S'$ joins $S_2$ independently \wp $p$.\;
    For each $v$ with $\deg_{S'}(v) \geq \alpha p^{-1}\log\Delta$, consider the event
    {\medskip\begin{enumerate}[leftmargin=1.5cm]
    \item[$E_d(v)$:] $\abs{ \deg_{S_2}(v) - p\deg_{S'}(v) } > (p/2)\deg_{S'}(v)$\;
    \end{enumerate}}\label{line:degreeSplittingPreShattering}

    \BlankLine
    \nonl\textbf{Retractions}\;
    Let $X$ be the set of vertices $v$ for which $E_d(v)$ hold\;
    Remove from $S_2$ every vertex adjacent to $X$, i.e., $S_2 \gets S_2 \setminus (X \cup N(X))$\;

    \BlankLine
    \nonl\textbf{Post-Shattering}\;
    We solve the following LLL using \label{line:dspostshattering}\cref{thm:deterministicLLLLOCAL}\;
    \nonl Every vertex $v\in N^{\leq 4}(X) \cap (S' \setminus S_2)$ independently joins $S_2$ \wp $p$\;
    \nonl For each $v\in N^{\leq 5}(X)$ with $\deg_{S'}(v) \geq \alpha p^{-1}\log\Delta$, we have the bad event
    {\begin{enumerate}[leftmargin=1.5cm]
        \item[$E_d'(v)$:] $\deg_{S_2}(v) > 4p\deg_{S'}(v)$ or $\deg_{S_2}(v) < (p/2)\deg_{S'}(v)$.
    \end{enumerate}}
\end{algorithm}

\medskip\noindent

The events $E_d(v)$ and $E'_d(v)$ are determined by the random choices of their neighbors. 
\begin{observation}
    \label{obs:ds-random-variables}
    $\vbl(E_d(v)) \subseteq N_H^{\leq 1}(v)$ and $\vbl(E'_{d}(v)) \subseteq N_H^{\leq 1}(v)$.
\end{observation}

The probability that a bad event $E_d(v)$ occurs during pre-shattering is low by a direct application of the Chernoff Bound.

\begin{claim}
    \label{claim:ds-Ev}
    For all $v$, the event $E_d(v)$ occurs \wp at most $2\exp( - p \deg_{S'}(v)/12 )$.
\end{claim}

\medskip\noindent
From this, we argue that the connected components of the post-shattering dependency graph are small.

\begin{claim}
    \label{claim:degreeSplittingShattering}
    Suppose $\Delta \leq (\log n)^{50}$.
    The connected components of the dependency graph $\Gev$ of \Cref{line:dspostshattering} of \cref{alg:ds}  have size at most $\poly(\log n)$ with high probability.
\end{claim}
\begin{proof}
    In order to bound the size of the largest connected component of $\Gev$ we use \cref{lem:our-shattering}; next we argue that the lemma applies. We take the whole graph $F$ as graph $G$ from \cref{lem:our-shattering}, because every vertex $v$ of $S'\subseteq V(F)$ has one variable encoding whether it is contained in $S_2$ or not. The LLLs (formally sets of sets) $\mathcal{A}$ and $\mathcal{B}$ required to use the lemma are both given by all collections of $\vbl(E_d(v))$  for all $v\in V(H)$ satisfying $\deg_{S'}(v) \geq \alpha p^{-1}\log\Delta$ and $\vbl(E'_d(v))$ for all $v\in V(H)$ satisfying $\deg_{S'}(v) \geq \alpha p^{-1}\log\Delta$ respectively. The set $\mathcal{A}'$ consists of the events $E_d(v)$ that hold after \cref{line:degreeSplittingPreShattering} of \cref{alg:ds}. Let $\mathcal{B}'$ with $\cstB = 5$ be as in the statement of \Cref{lem:our-shattering}. Let $c_1=2$ and $c_2=1$.
    We show that the preconditions 
       \Cref{part:our-shattering-diam,part:our-shattering-involved,part:our-shattering-pre-prob,part:our-shattering-vbl} of \Cref{lem:our-shattering} hold. 
        
    \Cref{part:our-shattering-diam} holds due to \cref{obs:ds-random-variables} with $c_1=2$. \Cref{part:our-shattering-involved} requires a bound $\Delta(F)^{c_2}$ for some constant $c_2$ on the number of events in which each variable appears. It  holds with $c_2 = 1$ as a vertex $v$ holds variables used only by events $E_d(u)$ or $E'_d(u)$ with $u\in N_H^{\leq 1}(v)\subseteq N_F^{\leq 1}(v)$.   \Cref{part:our-shattering-vbl} requires that for each event $A\in \mathcal{A}$ we can determine $1(A\in \mathcal{A}')$ by only evaluating the variables in $\vbl(A)$ which immediately holds by the definition of $\mathcal{A}'$. \Cref{part:our-shattering-pre-prob} requires an upper bound of $1/\Delta(F)^{c_4}$ on the probability for each bad event $A\in \mathcal{A}$ to be contained in $\mathcal{A}'$ where $c_1$ is an arbitrary constant satisfying $c_4>3c_1(4\cstB+16)+c_2+\cstB+1$. By \cref{claim:ds-Ev} this probability is exponentially small in $\Delta(G)=\Theta(\Delta(F))$ and hence \cref{part:our-shattering-pre-prob} holds for $c_4$ (as chosen above) as long as $\Delta(G)$ is larger than a sufficiently large absolute constant $\Delta_0$. The set of events included in the LLL of \cref{line:rct-single-it-post-LLL} are all included in $\mathcal{B}'$. Hence by \cref{lem:our-shattering}, the connected components of $\Gev$ have size at most $O(\Delta(F)^{4c_1\cstB + 4c_1 + c_2}\log n) =\poly(\log n)$ with high probability.
\end{proof}

\medskip\noindent
And we now prove that the post-shattering events indeed form an LLL.

\begin{claim}
    \label{claim:degreeSplittingpostShatteringLLL}
    In post-shattering, the probability that $E_d'(v)$ occurs is at most $2\Delta^{-\alpha/24}$, and dependency degree of the $\set{E_d'(v): v \in N^{\leq 5}(X)}$ is at most $\Delta^2$.
\end{claim}
\begin{proof}
    Every vertex $v \in N^{\leq 5}(X)$ has at most $\deg_{S'}(v)$ neighbors taking part in the post-shattering sampling. In expectation at most $p\deg_{S'}(v)$ of them join $S_2$ and, by the Chernoff Bound, the probability to have more than $2p\deg_{S'}(v)$ join $S_2$ is at most $\exp(-p\deg_{S'}(v)/12)$. If $v\in X$, then none of its neighbors remain in $S_2$ after retractions. Otherwise, $E(v)$ did not hold, and hence $v$ has at most $2p\deg_{S'}(v)$ neighbors in $S_2$ before the post-shattering sampling. Thus, the probability that a vertex has more than $4p\deg_{S'}(v)$ neighbors in $S_2$ is upperbounded by $\exp(-p\deg_{S'}(v)/12) \leq \Delta^{-\alpha/12}$.

    A vertex $v \notin N^{\leq 3}(X)$ does not see any retraction in its neighborhood and has enough neighbors in $S_2$ already (it would otherwise belong to $X$).
    Let $v\in N^{\leq 3}(X)$. Observe that $N(v) \cap S' \subseteq N^{\leq 4}(X) \cap S'$, hence all its neighbors in $S' \setminus S_2$ participate in the post-shattering sampling. Call $d$ its number of neighbors in $S_2$ before the post-shattering (and after retractions). It has at $p(\deg_{S'}(v) - d)$ neighbors that join $S_2$ in the post-shattering, hence by the Chernoff Bound, it gains fewer than $p(\deg_{S'}(v) - d)/2$ neighbors in $S_2$ during the post-shattering with probability at most $\exp(-p(\deg_{S'}(v) - d)/12)$. Hence, the degree in $S_2$ is smaller than $d + p(\deg_{S'}(v) - d)/2 = p\deg_{S'}(v)/2 + (1 - p/2)d$ \wp at most $\exp(-p(\deg_{S'}(v) - d)/12)$. If $d \geq p\deg_{S'}(v)/2$, then we already have that $v$ has enough neighbors in $S_2$. Otherwise, after the post-shattering sampling it has fewer than $p\deg_{S'}(v)/2$ \wp at most $\exp(-p\deg_{S'}(v)/24) \leq \Delta^{-\alpha/24}$.

    By adding up the error probability of both ways, the event $E_d'(v)$ holds \wp at most $2\Delta^{-\alpha/24}$.
    Whether $E'(v)$ is a function of the samplings of all $u\in N(v)$. Hence $E_d'(v)$ is independent from every $E_d(u)$ where $u \notin N^{\leq 2}(v)$. So the dependency degree is bounded above by $\Delta^2$.
\end{proof}

\medskip\noindent
Now we can put all of this together to prove our degree splitting lemma.

\LemDegreeSplitting*

\begin{proof}[Proof of \cref{lem:ds}]
    Pre-shattering and retractions require $O(1)$ rounds of \LOCAL.
    \Cref{claim:degreeSplittingShattering} shows that each connected component in the post-shattering phase is \whp of size at most $\poly(\log n)$. \Cref{claim:degreeSplittingpostShatteringLLL} shows that it is an LLL with a sufficiently good polynomial criterion to apply \Cref{thm:deterministicLLLLOCAL} when $\alpha$ and $\Delta$ are large enough. Thus, the post-shattering LLL is solved in $\Ot{ \log^3\log n }$ rounds.
    
    It remains to show that the solution meets the requirements of the lemma. 
    
    \medskip\noindent
    \textit{Sufficiently few neighbors in $S_2$ (Part 1).} 
    Let $v\in N^{\leq 5}(X)$. Since $E'(v)$ does not hold after post-shattering, it has at most $4p\deg_{S'}(v)$ neighbors in $S_2$. If $v \notin N^{\leq 5}(X)$, then none of its neighbors participate in the sampling during post-shattering (recall that only nodes of $N^{\leq 4}(X)$ participate). Since $E(v)$ did not hold after pre-shattering, it has at most $2p\deg_{S'}(v)$ neighbors in $S_2$ after post-shattering as well.

    \medskip\noindent
    \textit{Enough neighbors in $S_2$ (Part 2):} 
    Let $v\in N^{\leq 5}(X)$. Since $E'(v)$ does not hold after post-shattering, it has at least $p\deg_{S'}(v)/2$ neighbors in $S_2$.
    Let $v\notin N^{\leq 5}(X)$. Since $v\notin X$, after the pre-shattering sampling, $v$ has at least $p\deg_{S'}(v)/2$ neighbors in $S_2$. Retractions only occur in $X \cup N(X)$, hence $v$ is not affected. The number of neighbors in $S_2$ can only increase during post-shattering, hence it has enough neighbors in $S_2$ at the end of the algorithm.
\end{proof}

\subsection{Slack Generation (Proof of \texorpdfstring{\cref{lem:alg-slackgen}}{Lemma~\ref{lem:alg-slackgen}})}
\label{sec:slackGeneration}

This section describes the algorithm that is executed in \cref{line:color-sparse-slackgen} of \cref{alg:tough-cookie}. The guarantees it provides are described by \cref{lem:alg-slackgen}.

At the heart of \cref{lem:alg-slackgen} is the following key observation about random color trials, formally stated in \cref{lem:prob-slackgen}: the coloring is such that $\Omega( \overline{m}/\Delta )$ colors are repeated in the neighborhood of a vertex with $\overline{m}$ missing edges. Crucially, as long as $\overline{m} \geq c\Delta\log \Delta$, this fails with probability at most $\Delta^{-\Theta(c)}$, giving rise to an LLL. This was already used by \cite{MR14} (see Lemma 41). We say that a color is \emph{repeated} if it is used at least two times.

We use the following variant of the result in which a set $A$ of active nodes try colors from a custom set of $q$ colors, which we prove in \cref{app:slackgen}.

\begin{restatable}{proposition}{lemSlackGeneration}
    \label{lem:prob-slackgen}
    There exists a universal constant $c_0$ for which the following holds.
    Let $\Delta$ and $q$ be positive integers with $q \ge \Delta/3$. 
Let $A \subseteq V(G)$ be a subset of nodes such that $F[A]$ has maximum degree at most $\Delta$, and $C$ a set of $q$ colors.  
    Consider a node $v \in V$ with at least $\overline{m}$ non-edges in $F[N(v) \cap A]$ and such that $\overline{m} \geq c_0 \cdot q$. If every vertex of $A$ samples a random color $\chi(v) \in C$ and retains it only if $\chi(v) \notin \chi(N(v) \cap A)$, then
    \begin{quote}
        node $v$ has at least  $\frac{ \overline{m} / q }{ 3 \cdot 10^{4} }$ repeated colored in $N(v)$
    \end{quote}
    with probability at least $1-\exp(-\Omega(\overline{m}/q))$.
\end{restatable}

\noindent
To implement \cref{lem:alg-slackgen} when $\Delta$ is small, we resort to the shattering technique. To ensure that \cref{lem:prob-slackgen} applies after shattering, we use two disjoint sets of colors $C_1 = \set{1, 2, \ldots, \Delta/2}$ and $C_2 = \set{\Delta/2+1, \ldots, c}$ for pre-shattering and post-shattering respectively. In post-shattering, some vertices will only have part of their neighbors sample colors. We thus need to ensure that the post-shattering instance retains enough of the original sparsity. This can be achieved by activating nodes randomly as \cref{lem:prob-sparsity-subsampling} describes.

\begin{restatable}{proposition}{PropSparsityPreserving}\label{lem:prob-sparsity-subsampling}
    Suppose $H$ is a graph of maximum degree $\Delta$ and let $v$ be a vertex such that $H[N(v)]$ contains $\overline{m}$ anti-edges. Sample every vertex into a set $A$ with probability $p \geq 8/\Delta$. The number of anti-edges induced by $H[N(v) \cap A]$ is at least $p^2 \overline{m}/2$ w.p.\ at least $1 - \exp(- \Omega( p\overline{m}/\Delta ))$.
\end{restatable}
\noindent
It follows immediately 
that \cref{lem:prob-sparsity-subsampling} holds in expectation. One obtains concentration through, e.g., Janson's inequality. 
The proof of \cref{lem:prob-sparsity-subsampling} is deferred to \cref{app:slackgen}.

\medskip\noindent
We can now describe the algorithm for \cref{lem:alg-slackgen} run on the set of vertices $S$. It begins by activating nodes by sampling them into a set $R\subseteq S$ \wp 1/2 , and runs the random experiment of \cref{lem:prob-slackgen} on activated nodes only.
When $\Delta \geq (\log n)^{50}$, the claims of \cref{lem:alg-slackgen} already hold with high probability. When $\Delta$ is small, the algorithm begins by identifying failures (see below), it then retracts some colors (\cref{line:slackgen-retracts}), and finally use a deterministic algorithm for fixing the remaining parts of the graph (\cref{line:slackgen-post-shattering}). 

To identify the aforementioned failures we define the following events. 
Define the following event for each $v\in S$: 
 {\medskip\begin{enumerate}[leftmargin=1.5cm]
        \item[$E_1(v)$:] $v$ has more than $\max\set{(11/20)\deg_S(v), \Delta/20}$  neighbors in $R$, 
    \end{enumerate}
    For each $v\in S$ with $\deg_S(v) \geq \Delta - 3\sqrt{\Delta}$ consider the following events
    \begin{enumerate}[leftmargin=1.5cm]
        \item[$E_2(v)$:] it has fewer than $3\sqrt{\Delta}$ repeated colors in $F[N(v) \cap R]$, or
        \item[$E_3(v)$:] it has fewer than $9 \cdot 10^5\Delta^{3/2}/8$ anti-edges in $F[N(v) \cap (S \setminus R)]$.
    \end{enumerate}
For each clique $A_i\in A_L\cup A_H$ define the following event:
   \begin{enumerate}[leftmargin=1.5cm]  
       \item[$E_{CC}(i)$:] \cref{def:notbig} is violated for clique $A_i$
   \end{enumerate} 
   }

\begin{algorithm}[H]
    \caption{Implementation of Slack Generation (\cref{lem:alg-slackgen}) \label{alg:slackgen}}
    \nonl\textbf{Pre-Shattering:}\;
    Each vertex of $S$ joins $R$ \wp $1/2$ \label{line:slackgen-activate}\;
    Each vertex in $R$ tries a random color in $C_1=[\Delta/2]$ 
    \label{line:slackgen-RCT}\;

\nonl\textbf{Retractions:}\;
    Let $\mathcal{A}'$ be the set of events $E_1(v)$, $E_2(v)$, $E_3(v)$, or $E_{CC}(i)$ (defined above) that occur.\label{line:slackgen-bad-events}
    
    Retract the colors of (i.e., uncolor) the vertices in $\vbl( \mathcal{A}' )$ and remove them from $R$\label{line:slackgen-retracts}\;

    \BlankLine
    \nonl\textbf{Post-Shattering:}\;
    Solve the following LLL using the deterministic algorithm of \cref{thm:deterministicLLLLOCAL}:\; 
    \nonl Every vertex in $\{v\mid \dist_F(v, \vbl( \mathcal{A}' ) ) \leq 2\}\cap (S \setminus R)$ tries a color in $C_2 = \set{\Delta/2+1, \ldots, \Delta-k}$ with probability $3/4$. \newline For each $v\in N^{\leq 3}(\vbl(\mathcal{A}'))$, we have the bad event {
    \begin{enumerate}[leftmargin=1.5cm]
      
        \item[$E_1'(v)$:] $v$ has more than $\max\set{ (4/5) \cdot \deg_{S \setminus R}(v), \Delta / 20 }$ colored neighbors in $S \setminus R$.
    \end{enumerate}
    For every $v\in N^{\leq 1}(\vbl(\mathcal{A}'))$ with $\deg_{S}(v) \geq \Delta - 3\sqrt{\Delta}$, we also have the following bad event
    \begin{enumerate}[leftmargin=1.5cm]
    \item[$E'_2(v)$:] $v$ has fewer than $1.05\sqrt{\Delta}$ colors repeated in $F[N(v) \cap (S \setminus R)]$
        \end{enumerate}
        For each  $A_i\in A_L\cup A_H$ with $\vbl(E_{CC}(i))\cap N^{\leq 1}(\vbl(\mathcal{A'}))\neq \emptyset$ define the following event:
   \begin{enumerate}[leftmargin=1.5cm]  
       \item[$E_{CC}(i)$:] \cref{def:notbig} is violated for clique $A_i$.
   \end{enumerate} 
        
        } \label{line:slackgen-post-shattering}
\end{algorithm}

\medskip\noindent
Let us begin by proving that events from the pre-shattering are unlikely.

\begin{claim}
    \label{claim:slackgen-pre}
    For all $v$ in $S$, each event $E_i(v)$, $i \in [3]$ occurs \wp at most $\exp(- \Omega( \Delta^{1/40} ) )$. For each $A_i\in A_L\cup A_H$ the event $E_{CC}(i)$ occurs with probability at most $\Delta\cdot \exp(- \Omega( \Delta^{1/40} ) )$.
\end{claim}
\begin{proof}
    For the same reason as in \cref{lem:cc-bad-event-for-rct}, 
    $\Prob{ E_{CC}(i)}  \leq \Delta\exp(- \Omega( \Delta^{1/40} ) )$. Also,  since vertices try colors with probability $1/2$, the Chernoff Bound directly implies that $\Prob{ E_1(v) } \leq \exp(-\Omega(\Delta))$. 

    Let us now consider some fixed $v$ with $\deg_S(v) \geq \Delta - 3\sqrt{\Delta}$ and bound the probability of $E_2(v)$ and $E_3(v)$.
    By \cref{lem:structuralDecomposition}(\ref{part:decomp-S}), the vertex $v \in S$ has $\overline{m} \geq 9 \cdot 10^5 \Delta^{3/2}$ anti-edges in its neighborhood. By \cref{lem:prob-sparsity-subsampling} on the subgraph of $F$ induced by $S$ and $p=1/2$, there are fewer than $9 \cdot 10^5 \Delta^{3/2} / 8$ of anti-edges in $F[N(v) \cap R]$ \wp at most $\exp(-\Omega(\Delta^{1/2}))$. Suppose that $F[N(v) \cap R]$ has $9 \cdot 10^5 \cdot \Delta^{3/2}/8$ induced anti-edges. By \cref{lem:prob-slackgen} (with $C = C_1$, $q \geq \Delta/3$ and $\overline{m} = 9 \cdot 10^5 \Delta^{3/2} / 8$), after vertices of $R$ try a random color from $C_1$, the number of pairs of non-adjacent neighbors colored the same in $N(v) \cap R$ is smaller than 
    $\frac{9 \cdot 10^5 \Delta^{3/2}/8}{ 3\cdot 10^4 \cdot \Delta } \geq 3\sqrt{\Delta}$ \wp at most $\exp(-\Omega(\Delta^{1/2}))$. Hence, the event $E_2(v)$ occurs w.p.\ at most $\exp(-\Omega(\Delta^{1/2}))$.

    To bound the probability of $E_3(v)$, observe that each vertex of $S$ joins the complement of $R$ independently \wp $1/2$. By \cref{lem:prob-sparsity-subsampling} on $H=F[S]$ where nodes are sampled in $A = S \setminus R$ \wp $p = 1/2$, the graph $F[N(v) \cap (S \setminus R)]$ contains fewer than $9 \cdot 10^5 \cdot \Delta^{3/2}/8$ anti-edges \wp at most $\exp(-\Omega(\sqrt{\Delta}))$. 
\end{proof}

\medskip\noindent
Next, we argue about the size of the connected components in the post-shattering instance. 

\begin{claim}
    \label{claim:slackgen-shattering}
    Suppose $\Delta \leq (\log n)^{50}$.
    Let $\Gev$ be the dependency graph of the LLL of \cref{line:slackgen-post-shattering} in \cref{alg:slackgen}. With high probability, its largest connected component has size at most $\poly(\log n)$.
\end{claim}

\begin{proof}
   In order to bound the size of the largest connected component of $\Gev$ we use \cref{lem:our-shattering}; next we argue that the lemma applies. We take the whole graph $F$ as graph $G$ from \cref{lem:our-shattering}, because every vertex $v$ of $H$ has one variable encoding whether it is contained in $R$ and its color choice in \cref{alg:slackgen}. The LLLs (formally sets of sets) $\mathcal{A}$ required to use the lemma are given by the collections of $\vbl(E_1(v)), \vbl(E_2(v)), \vbl(E_3(v))$ and $\vbl(E_{CC}(i))$ for the respective $v\in V(H)$ and uncolored cliques $A_i\in A_L\cup A_H$. The LLL for $\mathcal{B}$ is given by $\vbl(E'_1(v)), \vbl(E'_2(v))$ and $\vbl(E_{CC}(i))$ for the respective nodes $v\in V(H)$ and uncolored cliques in $A_i\in A_L\cup A_H$. The set $\mathcal{A}'$ consists of the events $\vbl(E_1(v)), \vbl(E_2(v)), \vbl(E_3(v))$ and $\vbl(E_{CC}(i))$ that hold after \cref{line:slackgen-RCT} of \cref{alg:slackgen}. Let $\mathcal{B}'$ with $\cstB = 3$ be as in the statement of \Cref{lem:our-shattering}. Let $c_1=5, c_2=2$, and $c_4=3c_1(4\cstB+16)+c_2+\cstB+2$.     
    We show that the preconditions 
   \Cref{part:our-shattering-diam,part:our-shattering-involved,part:our-shattering-pre-prob,part:our-shattering-vbl} of \Cref{lem:our-shattering} hold. 

    \Cref{part:our-shattering-diam} holds with a similar reasoning as in \cref{obs:rct-random-variables} (and because $F$ includes the cliques) for $c_1=5$. \Cref{part:our-shattering-involved} holds for $c_2=2$ as a variable is only contained in events with variables in its $2$-hop neighborhood.   \Cref{part:our-shattering-vbl} follows directly with the definition of $\mathcal{A}'$. \Cref{part:our-shattering-pre-prob} requires an upper bound of $1/\Delta(F)^{c_4}$ on the probability for each bad event $A\in \mathcal{A}$ to be contained in $\mathcal{A}'$ where $c_1$ is an arbitrary constant satisfying $c_4>3c_1(4\cstB+16)+c_2+\cstB+1$. By \cref{claim:slackgen-pre} this probability is at most $\Delta(G)\cdot \exp(-\Delta(G)^{1/40})$. As $\Delta(G)=\Theta(\Delta(F))$  \cref{part:our-shattering-pre-prob} holds for $c_4$ (as chosen above) as long as $\Delta(G)$ is larger than a sufficiently large absolute constant $\Delta_0$. The set of events included in the LLL of \cref{line:slackgen-post-shattering} are all included in $\mathcal{B}'$ with the $\cstB$ defined above; a smaller choice of $\cstB$ is possible, but it is more tedious to verify that all events are then included in $\mathcal{B}'$. Hence by \cref{lem:our-shattering}, the connected components of $\Gev$ have size at most $O(\Delta(F)^{4c_1\cstB + 4c_1 + c_2}\log n) =\poly(\log n)$ with high probability.

\end{proof}

\medskip\noindent
The technical heart of this proof is the following claim, which argues that \cref{lem:prob-slackgen} still applies in the post-shattering instance. This comes from the introduction of $E_3(v)$ to ensure that 
pre-shattering preserved enough sparsity for the post-shattering step.

\begin{claim}
    \label{claim:slackgen--post-LLL}
    In the post-shattering phase, events $E'_1(v)$ and $E'_2(v)$ occur \wp at most $\exp(-\Omega(\Delta^{1/2}))$.
\end{claim}

\begin{proof}
The bound on $\Prob{E'_1(v)}$ follows directly with a Chernoff bound as each node only picks a color with probability $3/4$, that is, in expectation $3/4 \deg_{S\setminus R}(v)$ neighbors participate in the coloring process of the post-shattering phase, and the constant factor deviation (lower bounded by $\Delta/20)$ forbidden by $E'_1(v)$ is exponentially small in $\Omega(\Delta)$. 
    
    Let $v \in N^{\leq 1}(\vbl(\mathcal{A}')) \cap S$ be a vertex with $\deg_{S}(v) \geq \Delta - 3\sqrt{\Delta}$.
    We claim that, after retractions, vertex $v$ has at least $9 \cdot 10^5 \cdot\Delta^{3/2}/8$ anti-edges in $F[N(v) \cap N^{\leq 2}(\vbl(\mathcal{A}')) \cap (S \setminus R)]$, where $\mathcal{A}'$ is the set of the occurring events previously defined. If $E_3(v)$ occurred during pre-shattering, we have that $N(v) \cap R =\emptyset$ because its neighbors are removed from $R$ during the retraction step. Otherwise, then it had $9 \cdot 10^5 \cdot\Delta^{3/2}/8$ anti-edges in $F[N(v) \cap (S \setminus R)]$ before retractions because $E_3(v)$ does not occur. Removing vertices from $R$ can only increase the number of anti-edges in $F[N(v) \cap (S \setminus R)]$; hence, it also has $9 \cdot 10^5 \cdot\Delta^{3/2}/8$ anti-edges in $F[N(v) \cap N^{\leq 2}(\vbl(\mathcal{A}')) \cap (S \setminus R)]$ after retractions as $N(v) \cap S \subseteq N^{\leq 2}(\vbl(\mathcal{A}')) \cap S$.
    
    Consider now some vertex $v\in N^{\leq 1}(\vbl(\mathcal{A}')) \cap S$ and let us upper bound the probability of $E_2'(v)$.
    As we just argued, every such vertex $v$ has $9 \cdot 10^5\Delta^{3/2}/8$ anti-edges in $F[N(v) \cap (S \setminus R)]$ (and all vertices of $N(v) \cap (S \setminus R)$ are trying colors if activated). By \cref{lem:prob-sparsity-subsampling} (with $p = 3/4$), node $v$ has fewer than $(3/4)^2/2 \cdot 9 \cdot 10^5 \cdot\Delta^{3/2}/8$ \wp at most $\exp(-\Omega(\Delta^{1/2}))$.
    So, by \cref{lem:prob-slackgen}, the node $v$ has fewer than 
    \[ 
    \frac{(3/4)^2/2 \cdot 9 \cdot 10^5 \Delta^{3/2} / 8}{3\cdot 10^4 \cdot \Delta} \geq 1.05 \sqrt{\Delta}
    \] repeated colors in $N(v)$ \wp at most $\exp(-\Omega(\Delta^{1/2}))$. So the event $E'_2(v)$ occurs \wp at most $\exp(-\Omega(\Delta^{1/2}))$.
\end{proof}

And we can now conclude with the proof of our main result, which we restate here.
\LemShatteringSlackGen*
\begin{proof}
    Both coloring steps (\cref{line:slackgen-RCT,line:slackgen-post-shattering}) produce a proper coloring with disjoint sets of color, and hence the resulting partial coloring is proper.
    By \Cref{claim:slackgen-pre}, when $\Delta \geq (\log n)^{50}$, the set $\mathcal{A}'$ is empty with high probability. Hence, \cref{alg:slackgen} ends after $O(1)$ rounds for such values of $\Delta$ and the partial coloring produced verifies the claims of \cref{lem:alg-slackgen}. When $\Delta \leq (\log n)^{50}$, by \Cref{claim:slackgen-shattering}, the connected components of the dependency graph in the post-shattering instance have size at most $\poly(\log n)$. The probability of $E_{CC}(i)$, $E_1'(v)$, and $E'_2(v)$ is $\Delta^{-\omega(1)}$ respectively from \cref{lem:bnd-notbig}, the Chernoff Bound and \Cref{claim:slackgen--post-LLL}. Since $\vbl(E_i(v)) \subseteq N^{\leq 2}(v)$ and $\vbl(E_{CC}(i))$ consists of every vertex of $S$ within 2 hops from $A_i$ in $F$, the dependency degree of the LLL is $\poly(\Delta(F)) = \poly(\Delta)$. Hence, they form an LLL that can be solved in $\Ot{ \log^3 \log n }$ by \cref{thm:deterministicLLLLOCAL}.

    It remains to verify that the claimed properties indeed hold.
    \cref{def:notbig} is respected during both pre-shattering and post-shattering so Part (b) holds as well.

    \medskip\noindent\textit{Part (a).}
    All the vertices in $V \setminus N^{\leq 1}(\vbl(\mathcal{A}'))$ have at least $3 \sqrt{\Delta}$ colors that appear at least twice in $N(v) \cap S$ because $E_2(v)$ does not hold for such vertices and none of their neighbors retract their color. All vertices in $v\in N^{\leq 1}(\vbl(\mathcal{A}'))$ with $\deg_S(v) \geq \Delta - 3\sqrt{\Delta}$ have $1.05\sqrt{\Delta}$ repeated colors because $E'_2(v)$ does not hold after the post-shattering step. Overall, Part (a) of \cref{lem:alg-slackgen} holds for all vertices of $S$. 
    
    \medskip\noindent\textit{Part (c).}
    Vertices $v\in V \setminus N^{\leq 3}(\vbl(\mathcal{A}'))$ have at most $11\Delta/20$ colored neighbors because $E_1(v)$ does not hold and none of their neighbors gets colored in post-shattering. A vertex $v$ with $E_1(v) \in \mathcal{A}'$ has at most $4\Delta/5$ colored neighbors after post-shattering because all of its neighbors are uncolored at the beginning of \cref{line:slackgen-post-shattering} and $E_1'(v)$ does not hold. For every other vertex, i.e., $v\in N^{\leq 3}(\vbl(\mathcal{A}')) $ but $E_1(v)$ did not occur during pre-shattering, the number of colored neighbors is at most
    \[
    \deg_R(v)
    + \max\set{ (4/5)\deg_{S \setminus R}(v), ~\Delta/20 }
    \]
    because $E_1'(v)$ does not hold and only vertices of $R$ are colored during pre-shattering. If we have that $(4/5)\deg_{S \setminus R}(v) \leq \Delta/20$, then since $\deg_R(v) \leq (11/20)\Delta$, the number colored neighbors is at most $(11/20 + 1/20)\Delta < 19\Delta/20$. Otherwise, using once again that $\deg_R(v) \leq 11\Delta/20$, the colored degree is at most 
    \begin{align*}
    \deg_R(v) + (4/5)(\deg_{S}(v) - \deg_{R}(v)) 
    &\leq (4/5)\deg_S(v) + (1/5)\deg_R(v) \\
    &\leq (4/5 + 1/5 \cdot 11/20)\Delta 
    \leq (19/20)\Delta \ . \qedhere
    \end{align*}
\end{proof}

\section{Coloring Cliques}
\label{sec:cliques}
The goal of this section is to extend the coloring to the cliques in $A_H$ and $A_L$, respectively. In \cref{alg:highlevel}, the cliques of $A_H$ are colored first, and thus when coloring these nodes we need to ensure that CC is satisfied for all cliques in $A_L$. 
We prove the following:

\LemColoringCliques*

At a high level, we follow the outline of the approach of \cite{MR14}. In the first phase, we compute a defective coloring (i.e., a coloring with monochromatic edges) of the cliques to be colored by assigning a random permutation of the colors \emph{not} used in $\All_i$ to the vertices of $A_i$. About $\sqrt{\Delta}$ vertices of each clique have a conflict in that they are assigned the same color as an external neighbor
(see \cref{lem:sct}(a)).
In the second phase, we perform color swaps: each node with a conflict finds a partner within its own clique with whom it can swap colors to remove the conflict without introducing new ones.
This is where the CC property is crucially used.

Our Step~1 is largely similar to the first phase, with the difference of making the conflicts one-way directed in order to simplify the formulation in the shattering framework. 
We implement the second phase quite differently and do so in two steps (Steps 2 and 3). 
In our Step 2, we find for each conflicted node a set of candidates for swapping. These sets are designed so that the task of selecting a final swapping partner from the candidate sets is independent across different cliques.
In Step 3, we then choose the swap partners locally based on a bipartite matching and apply the swaps in parallel.

\subsection{Step 1: Synchronized Color Trial}
We compute for each clique $A_i$ in $A'$ a random permutation of the 
$|A_i|=c-|\All_i|$ colors not used by $\All_i$ (following \cite{MR14}) and assign them to the nodes of $A_i$. This introduces a limited amount of conflicts, or monochromatic edges.
We orient these monochromatic edges for easier shattering formulation: 
if $\{u,v\}$ is oriented from $u$ to $v$, it means that $v$ should change its color in Step~2.  

\begin{definition}
\label{def:unhappy}
    For a defective coloring $\gamma$ and an orientation of the monochromatic edges of $\gamma$, 
    we define $\unhappy_i$ as the set of vertices of each $A_i$ in $A'$ with an incoming monochromatic edge
    Let $\unhappy=\bigcup_i \unhappy_i$ and call a node \emph{unhappy} if it belongs to $\unhappy$.
\end{definition}

Given the upper bound on external degrees, we expect each clique to have $O(\sqrt{\Delta})$ unhappy nodes. We can achieve that \whp via an LLL formulation.

\begin{restatable}{lemma}{LemSct}
\label{lem:sct}
    \label{lem:SCT-cliques}
    There is a \LOCAL algorithm that computes a \textit{defective} coloring $\gamma$ of $A'\subseteq A$ and an orientation of its monochromatic edges such that
    \begin{enumerate}[label=(\alph*)]
    \item $|\unhappy_i| \le (10^8 + 1) \sqrt{\Delta}$ for all $A_i \in A'$, and 
    \item All cliques of $A'$ and all uncolored cliques not in $A'$ satisfy the CC property.
    \end{enumerate}
    The algorithm runs in a single round when $\Delta \ge (\log^{50} n)$ and in $\Ot{\log^3\log n}$-rounds otherwise.
\end{restatable}
The proof is deferred to \cref{sec:proof-sct}.

\subsection{Step 2: Finding Safe Swaps}

We now compute a pool of potential partners for each unhappy vertex.
A linear fraction of the nodes in the clique satisfies a minimal requirement. We then cull the pool by probabilistic subsampling. Any candidate that could cause a conflict if some other swap were performed is then eliminated. This eliminates potential inter-clique conflicts.
In Step 3, we then select the actual swaps to avoid intra-clique conflicts.

A node $u$ is a \emph{swappable} candidate for a node $v$ (in the same $A_i$) if swapping their colors leaves them both conflict-free when the coloring of the rest of the graph is unchanged.
Let $\swap_v$ be the set of swappable candidates for $v$.

\begin{definition} 
    \label{def:swap}
    For each $v\in A_i\in A'$ the set $\swap_v$ consists of all nodes $u\in A_i$ that satisfy
    \begin{enumerate}[label=(\alph*)]
        \item $u \not\in \unhappy_i$,
        \item $\gamma(u)$ does not appear on an external neighbor of $v$, and
        \item $\gamma(v)$ does not appear on an external neighbor of $u$.
    \end{enumerate}
\end{definition}

\noindent
We next show that a constant fraction of the nodes of $A_i$ are swappable (for any given node $v$).
This lemma is the \emph{only place} where we use the CC property. Its proof is purely deterministic, given the CC property of the clique and the guarantees of \Cref{lem:sct}.

\begin{lemma}
\label{lem:swappable}
    All $v\in A_i\in A'$ have $|\swap_v|\geq \Delta/10$.
\end{lemma}
\begin{proof}
    We bound separately from above the number of nodes violating (a), (b) and (c) in \cref{def:swap}.
    \emph{Nodes violating (a):} By \Cref{lem:sct}(a) at most $4\sqrt{\Delta}$ vertices of $A_i$ are in $\unhappy_i$.
    
    \emph{Nodes violating (b):} The external degree of $v$ is bounded by $\sqrt{\Delta}$ (by \Cref{lem:structuralDecomposition}(3,5)), and hence at most $\sqrt{\Delta}$ vertices in $A_i$ have a color appearing in the external neighborhood of $v$.
    
    \emph{Nodes violating (c):} Note that CC holds for the clique $A_i$, as it held before Step 1 (recall \cref{lem:coloringCliques}(1)) and was maintained by Step 2 (by \cref{lem:sct}(b)). Hence, by \Cref{obs:CC} each color (here applied to $\gamma(v)$) appears on an external neighbor (recall, one outside of $A_i\cup \mathrm{All}_i$) of at most $4\Delta/5$ vertices in $A_i$.  
    
    \medskip\noindent
    Putting everything together,
    \[ 
    |\swap_v| \ge |A_i|-4\Delta/5 - 5 \sqrt{\Delta}\geq \Delta/10,
    \]
    using the lower bound on $|A_i|$ from \Cref{lem:structuralDecomposition}(a).
\end{proof}

While individual swaps can be done safely (with the sets $\swap_v$), pairs of swaps could introduce new conflicts.
Instead, we seek in this step to truncate the sets to obtain  
a \emph{safe} candidate system. 
Intuitively, a candidate is safe with respect to the candidates of other cliques if it can swap its color obliviously to what happens in other cliques without creating monochromatic edges. 
\begin{definition}
\label{def:badCandidate}
Given a collection $\fT = \set{ T_v \subseteq \swap_v : v\in \unhappy }$  of \emph{candidate sets} for each unhappy node, we call a candidate $u \in \swap_v$ \emph{unsafe for $v$} with regard to $\fT$ if one of the following holds:
\begin{enumerate}[label=$\roman*)$]
     \item $v$ has an external neighbor $w$ that has a candidate $w'\in T_w$ with $\gamma(w') = \gamma(u)$;
     \item $v$ has an external neighbor $w\in T_{w'}$ that is a candidate for a node $w'$ with $\gamma(w') = \gamma(u)$;
\item $u$ has an external neighbor $w\in T_{w'}$ that is a candidate for a node $w'$ with $\gamma(w') = \gamma(v)$.     
 \end{enumerate}
A candidate that is not unsafe is \emph{safe}. 

We say that $\fT$ is a \emph{safe candidate system} if the sets $T_v$ contain only safe candidates for $v$ for all $v\in\unhappy$. 
\end{definition}

\noindent
\Cref{def:badCandidate} has three instead of four cases, because if $u$ has an external neighbor $w$ that has a candidate $w'$ with $\gamma(w') = \gamma(v)$, then $w'$ is also an unsafe candidate for $w$. 
Note that whether a candidate is safe or not depends on the candidate sets of nodes in adjacent cliques. 

The most involved step of our algorithm for coloring cliques is the construction of a safe candidate system as described in \cref{lem:subsampling}. We defer the proof to \cref{sec:proof-clique-subsampling} to preserve the flow of the paper.

\begin{restatable}[Subsampling for safe candidates]{lemma}{LemSubsampling}
\label{lem:subsampling}
    There is a $\Ot{ \log^3\log n }$ \LOCAL algorithm that, \whp, computes a safe candidate system $\fT$ such that 
    \begin{enumerate}\item $|T_v| \ge \Delta^{17/40}/40$ for each $v\in \unhappy$,
\item Each $u\in A_i \setminus \unhappy_i$ belongs to at most $\Delta^{17/40}/80$ sets $T_v$ with $v\in \unhappy_i$, and      
        \item (Strong version of CC) For each uncolored $A_j \not\in A'$ and color $x$, at most $2\cdot \Delta^{37/40}$ vertices in $A_j$ have a neighbor outside of $A_j\cup \All_j\cup \Big_j^+$ that has a candidate of color $x$ or is a candidate for a node of color $x$.
    \end{enumerate}
\end{restatable}

\subsection{Step 3: Local Matching}
We finally show how to turn a safe candidate system (\cref{lem:subsampling}) into a proper coloring, concluding the proof of \cref{lem:coloringCliques}. 
We first observe that safeness implies that the color selection becomes a fully local problem within each clique.

\begin{observation}
    Suppose an unhappy node $v$ performs a swap of colors with one of its safe candidates $u \in T_v$. Afterwards, there are no incoming monochromatic edges to $v$ (nor $u$), independent of any swaps performed in other cliques.

    Further, if all unhappy nodes perform swaps with safe candidates and if the swaps are disjoint, then the resulting coloring is proper.
    \label{obs:safe}
\end{observation}
\begin{proof}
    For the first claim, consider an external neighbor $w$ of $v$ and let $A_j$ be its clique. Since $u$ was swappable for $v$, $\gamma(u) \ne \gamma(w)$, so if $u$ did not perform a swap, there is no conflict. Suppose then that $w$ swapped its color with one of its safe candidates $w'$. Then $\gamma(w') \ne \gamma(u)$, by \cref{def:badCandidate}(i), so again there is no conflict. The case of conflicts with $u$ is similar.

    For the second claim, observe that by the first claim and since all unhappy nodes perform swaps, all monochromatic edges between cliques are eliminated and no new ones are introduced. The disjointness criteria -- that no node participates in more than one swap -- ensures that no conflicts occur within the cliques.
\end{proof}

The remaining task -- deciding on swap partners -- then reduces to a local bipartite matching problem.
Consider the bipartite graph with vertex bipartition $L = \unhappy_i$ and $R = A_i \setminus \unhappy_i$, with an edge between each $v\in L$ and $u\in R$ if and only if $u \in T_v$. 
By \cref{lem:subsampling}(1,2), each node in $L$ has at least $\Delta^{17/40}/40$ incident edges, while each node in $R$ has at most $\Delta^{17/40}/80$ incident edges. It follows that for every subset $S \subseteq L$, $|N(S)| \ge |S|$. Then, by Hall's theorem there is a matching $M$ that saturates $L$. 
We perform exactly the color swaps given by the matching $M$, which yields a proper coloring by \cref{obs:safe}. 

To conclude with the proof of \cref{lem:coloringCliques}, it remains to argue that the CC property is maintained for the remaining uncolored cliques.
This follows for Step 1 by \cref{lem:sct}(b) and for Step 3 by \cref{lem:subsampling}(3).

\subsection{Proof of \texorpdfstring{\cref{lem:sct}}
{Lemma~\ref{lem:sct}}}
\label{sec:proof-sct}

\LemSct*

\begin{algorithm}
    \caption{Synchronized Color Trial on $A_i$\label{alg:sct}}
    \Input{A collection of uncolored cliques $A'' \subseteq A'$}
    \Output{A coloring $\gamma$ of all the $A_i \in A''$ and an orientation of its monochromatic edges}
    
    Each $A_i \in A''$ samples a permutation $\pi_i$ of the colors \emph{not} already used on $\All_i$\;

    Order the vertices of $A_i$ as $v_1, v_2, \ldots, v_{|A_i|}$ arbitrarily and let $\gamma(v_j)$ be the color $\pi_i(j)$

    Orient the resulting monochromatic edges $\set{u,v}$ as follows: {\begin{itemize}[noitemsep]
        \item if $v\in A_i\in A'$ and $u$ is not in a clique of $A'$, orient the edge towards the vertex $v$, and vice versa;
        \item if $v\in A_i\in A'$ and $u\in A_j\in A'$, orient the edge from $u$ to $v$ if $i < j$ and in the opposite direction otherwise.
    \end{itemize}}
\end{algorithm}
Note that due to \Cref{lem:structuralDecomposition}(b), $|\All_i|=c-|A_i|$, and hence a color is available for each node in $A_i$.
In this process, each clique corresponds to exactly one random variable. Recall that a vertex is unhappy (\cref{def:unhappy}) if it has at least one incoming monochromatic edge. We consider two bad events for each $A_i$: define
\begin{enumerate}[noitemsep,topsep=0em,leftmargin=1.5cm]
    \item[$E_a(i)$] : $\unhappy_i$ contains more than $10^8 \sqrt{\Delta}$ vertices; and 
    
    \item[$E_b(i)$] : $A_i$ is adjacent to a clique of $A'$ or to an uncolored clique not in $A'$ for which the CC property was broken.
\end{enumerate}
Clearly, the events $E_a(i)$ and $E_b(i)$ depend on the random permutation of $A_i$ and of every $A_j$ adjacent to $A_i$. Let us argue that they are rare events, no matter what the coloring outside $A_i$ is. We emphasize that \Cref{claim:sct-ev} bounds the probability of both events for pre-shattering and post-shattering because it assumes the coloring of $F - A_i$ is adversarial.

\begin{claim}
    \label{claim:sct-ev}
    Let $A_i \in A$ be an uncolored clique, and let the coloring of $F - A_i$ be arbitrary. 
    When we run \cref{alg:sct} with an $A''$ that contains $A_i$, then events $E_a(i)$ and $E_b(i)$ occur with probability at most $\exp(-\Omega(\Delta^{1/40}))$.
\end{claim}
\begin{proof}
    Let $\Temp_i\subseteq A_i$ consist of the nodes of $A_i$ that have a color conflicting with an  external neighbor. Clearly, $\unhappy_i\subseteq \Temp_i$.
The colors are drawn from a set of size $|A_i| \ge \Delta/2$ and at most $10^8\sqrt{\Delta}$ of them conflict with those of external neighbors (\cref{lem:structuralDecomposition}(3,5)). Hence, each node has conflict w.p.\ at most $2 \cdot 10^8/\sqrt{\Delta}$, so 
the expected size of $\Temp_i$ is at most $2 \cdot 10^8\sqrt{\Delta}$.
    As Molloy and Reed argue in Lemma 39 of \cite{MR14} by applying McDiarmid's inequality (\cref{lem:mcdiarmid}, with $c=r=1$), $|\Temp_i|$ is highly concentrated. 
Specifically, $E_a(i)$ occurs \wp at most $4\exp(-(\sqrt{\Delta})^2/(128\cdot (10^8+1) \sqrt{\Delta})) \leq \exp(-\Delta^{1/3})$ for $\Delta$ sufficiently large. 

    Let us now bound $\Prob{E_b(i)}$. Fix an arbitrary color $x$ and a clique $A_j$. 
    Due to \Cref{lem:structuralDecomposition} the external degree of $A_i$ is bounded by $10^8\sqrt{\Delta}$. 
    For a fixed external neighbor $u \in A_\ell \in A''$ of $A_j$, the probability that it receives color $x$ in the random permutation of $A_j$ is $1/|A_j| \leq 2/\Delta$ (by \cref{lem:structuralDecomposition}(b)). Furthermore, at most one vertex in each clique is assigned color $x$ and the random permutation for different cliques are independent. Hence, we can apply \Cref{lem:bnd-notbig} with $Q=10^8\sqrt{\Delta}$ as it bounds external degrees from above (by \cref{lem:structuralDecomposition}(3,5)), and get that $A_j$ does not satisfy CC for color $x$ with probability at most $\exp(-\Delta^{1/40})$. By union bound on all colors and cliques adjacent to $A_i$, the event $E_b(i)$ occurs \wp at most $\Delta^{3/2+1}\cdot \exp(-\Delta^{1/40}) \leq \exp(-\Omega(\Delta^{1/40}))$.
\end{proof}

To ensure that \cref{lem:sct}(a,b) hold even when $\Delta$ is small, we employ the shattering framework (\cref{lem:our-shattering}).

\begin{algorithm}
    \caption{Coloring Algorithm for \cref{lem:sct}\label{alg:sct-shattering}}
    Run \cref{alg:sct} on every $A_i \in A'$\;

    Let $A''$ be the set of cliques $A_i \in A'$ for which $E_a(i)$ or $E_b(i)$ occurs\;

    Uncolor every vertex in cliques $A_i \in A''$\label{line:sct-retract}\;

    Solve the following LLL using \cref{thm:deterministicLLLLOCAL}:\;\label{line:sct-post}
    
    \nonl \textbf{Random Process:} Run \cref{alg:sct} on every $A_i \in A''$ and orient monochromatic edges between $A' \setminus A''$ and $A''$ toward $A''$

    \nonl \textbf{Bad Events:} $E_a(i)$ and $E_b(i)$ for all $i\in A''$
    
\end{algorithm}

\begin{proof}[Proof of \cref{lem:sct}]
    We run \cref{alg:sct-shattering} and show that properties (a) and (b) hold after the post-shattering phase. When $\Delta \geq (\log n)^{50}$, it follows directly from \Cref{claim:sct-ev} and the union bound that $A'' = \emptyset$, and hence the algorithm ends in $O(1)$ rounds and is correct with high probability.

    We argue now that the random process in \cref{line:sct-post} of \cref{alg:sct-shattering} along with events $E_a(i)$ and $E_b(i)$ form an LLL whose dependency graph $\Gev$ has $\poly(\log n)$-sized connected components. The events $E_a(i)$ and $E_b(i)$ depend on the random choices of at most $10^9\Delta^{3/2}$ neighboring almost-cliques, hence the dependency degree of the LLL is at most $10^{18}\Delta^3$. By \Cref{claim:sct-ev}, every event occurs \wp $\Delta^{-\omega(1)}$, regardless of the coloring produced by the earlier steps of \cref{alg:sct-shattering}. Thus, the set of events described in \cref{line:sct-post} indeed describes an LLL. We next want to apply \cref{lem:our-shattering} to bound the size of connected components.

    Consider the variable graph $H$ of this random process: it has a vertex (i.e., a variable) $v_i$ for each $A_i \in A'$ and vertices $v_i$ and $v_j$ are connected iff an edge connects $A_i$ and $A_j$. The collections $\mathcal{A}$ and $\mathcal{B}$ from \cref{lem:our-shattering} are all 1-hop neighborhoods in $H$, i.e., all the $\set{v_i} \cup N_H(v_i)$ for $A_i\in A'$. Define $\cstB = 0$ so that the set $\mathcal{B}'$ in \cref{lem:our-shattering} includes the 1-hop neighborhood of every $v_i$ with $A_i \in A''$. This is a superset of $A''$ (the set of cliques for which we have an event in post-shattering), so it suffices to bound from above the sizes of the connected components induced by the events corresponding to sets of $\mathcal{B}'$.
    \cref{part:our-shattering-diam,part:our-shattering-involved} with $c_1 = 2$ and $c_2 = 1$ are direct from the choice of sets $\set{v_i} \cup N_H(v_i)$ because it has diameter at most 2 and every $v_i$ is included only in the sets of its neighbors. \cref{part:our-shattering-vbl} is direct from \cref{alg:sct}. By \Cref{claim:sct-ev}, each $A_i$ is included in $A''$ (i.e., a bad event occurred and $A_i$ uncolored its vertices) with probability at most $\exp(-\Omega(\Delta^{1/40}))$.
    The degree of a $v_i$ is at most $10^9 \Delta^{3/2}$, because every vertex in $A_i$ may be adjacent to up to $10^9\sqrt{\Delta}$ vertices in other cliques. It implies that for $\Delta$ sufficiently large, we can pick a constant $c_4 > 3c_1(4\cstB + 16) + c_2 + \cstB + 1 = O(1)$ such that each set of $\mathcal{A}$ is kept in $\mathcal{A}'$ with probability at most $\Delta(H)^{-c_4}$, i.e., \cref{part:our-shattering-pre-prob}. Therefore \cref{lem:our-shattering} applies and, \whp, we can bound the size of the largest component in the dependency graph of the LLL of \cref{line:sct-post} by some $\poly(\log n)$.
    Let us conclude the proof by proving that the resulting coloring indeed satisfies (a) and (b). For (b), it simply follows from the fact that after retractions (\cref{line:sct-retract}), the \cref{def:notbig} is preserved, and since events $E_b(i)$ are avoided in post-shattering, \cref{def:notbig} is also preserved by this coloring step (with a fresh budget). To see why (a) holds, observe that every $A \notin A''$ satisfies (a) after \cref{line:sct-retract} (otherwise $E_a(i)$ occurs and $A_i$ is included in $A''$).
    Importantly, if a monochromatic edge $\set{u,v}$ appears where $v$ was colored by \cref{line:sct-post} and $u$ belongs to a clique of $A'$ that did not participate in post-shattering (not in $A''$), we orient the edge from $u$ to $v$. So, the coloring of a clique $A_j$ that participates in the post-shattering phase can never increase the size of $\unhappy_i$ of an adjacent clique $A_i$ that does not participate in the post-shattering phase. Finally, a clique $A_i \in A''$ colored by \cref{line:sct-post} satisfies (a) at the end of the algorithm, as the event $E_a(i)$ would otherwise occur in the post-shattering LLL.
\end{proof}

\subsection{Proof of \texorpdfstring{\cref{lem:subsampling}}{Lemma~\ref{lem:subsampling}}}
\label{sec:proof-clique-subsampling}
\LemSubsampling*

\noindent
Our approach is based on the following random procedure: 
for each node $v\in \unhappy$, 
\begin{enumerate}
\item each node $u\in \swap_v$ samples itself into a set $S_v$ with probability $p = \Delta^{-23/40}$; then
\item all candidates $u\in S_v$ that are safe for $v$ w.r.t.\ $\fS = \set{ S_v: v\in \unhappy}$ join $T_v$,
\item output the safe candidate system $\fT = \set{ T_v : v\in \unhappy }$.
\end{enumerate}

\noindent
The sampling is done independently for different nodes $v, v' \in \unhappy$, that is, a node $u$ may be contained in some $S_v$ but not in $S_{v'}$. We emphasize that, when $\Delta \leq \poly(\log n)$, that process does not succeed everywhere; hence, we use (again) the shattering technique. The pre-shattering step will produces large enough sets $T_v^{pre}$ of safe candidates for most cliques. The post-shattering step then computes sets $T_v$ for the remaining cliques based on the same LLL and \cref{thm:deterministicLLLLOCAL}.

Before proving \cref{lem:subsampling}, we prove the following technical claim --- that we will use both in pre- and post-shattering --- to bound the number of unsafe candidates.
The set of cliques $A'' \subseteq A'$ will be all cliques of $A'$ in pre-shattering, and, in post-shattering, all the cliques for which we made retractions. \Cref{claim:b1Detail} shows that a set $S$ of candidates contains only a constant fraction of unsafe candidates with regard to candidates sampled by neighboring cliques.

Note that the random process in \Cref{claim:b1Detail} depends only on the outside randomness, i.e., the set $S$ is fixed deterministically. We also emphasize that the success probability depends on the degree $k$ of the outside vertices into $S$. We abuse notation slightly and write $A' \setminus A_i$ for $A' \setminus \set{A_i}$.

\begin{claim}
\label{claim:b1Detail}
    Let $A'' \subseteq A'$ be a subset of the cliques we wish to color. Suppose that every node in $Swappable_w$ for some $w \in \unhappy_i$ with $A_i \in A''$ joins the candidate set $S_w$ w.p.\ $p < 10^{-10}\Delta^{-1/2}$. Consider a fixed $v \in A_i$ and $S \subseteq \swap_v$, and call $k$ the maximum number of neighbors in $S$ that vertices $w \in A_j \in A' \setminus A_i$ can have, i.e., $k = \max_{ w \in A_j \in A' \setminus A_i } |N(w) \cap S|$. 
    Then, 
    $S$ contains fewer than $|S|/20$ unsafe candidates for $v$ w.r.t.\ $\fS = \set{S_w : w\in\unhappy}$ 
    \wp at least $1 - \exp\paren*{ - \Omega\paren*{ |S| / k } }$.
\end{claim}

\begin{proof}
Recall from \cref{def:badCandidate} that a node $u \in S$ can be an unsafe candidate for $v$ for three reasons. We first bound the number of candidates that are bad due to $(i)$ or $(ii)$. Then, we analyze the number of unsafe candidates due to $(iii)$.

\medskip\noindent
\textit{Counting unsafe candidates due to rule $(i)$ or $(ii)$:} 
For $u\in S$, let $X_u$ be the indicator random variable equal to one iff either $(i)$ or $(ii)$ in \Cref{def:badCandidate} holds w.r.t.\ the sets of candidates $\set{ S_w }_w$. 
To have $X_u = 1$ by rule $(i)$, node $v$ must have an external neighbor $w\in\unhappy$ with a candidate colored $\gamma(u)$ and sampled in $S_w$. To have $X_u = 1$ by rule $(ii)$, node $v$ must have an external neighbor in some $S_w$ where $w$ is colored $\gamma(u)$. For a given external neighbor, at most one of those can occur, each w.p.\  at most $2p$; thus, $\Exp[X_u] \leq 10^8\sqrt{\Delta} \cdot p < 1/100$ (recall that by \Cref{lem:structuralDecomposition}, node $v$ has at most $10^8\sqrt{\Delta}$ external neighbors). The $\set{ X_u }_{u\in A_i}$ are independent because each $X_u$ depends only on the randomness of nodes colored $\gamma(u)$, and all of the nodes of $A_i$ are colored differently. We apply the Chernoff Bound to the sum of the $X_u$ over $u \in S$ and obtain that the probability of having more than $|S|/20$ nodes removed is bounded from above by $\exp(- \Omega( |S| ))$.

\medskip\noindent
\textit{Counting unsafe candidates due to rule $(iii)$:} 
For $u\in S$, let $Y_u$ be the indicator random variable equal to one iff $(iii)$ in \Cref{def:badCandidate} holds for $u$ w.r.t.\ sets $\fS$. We have that $Y_u = 1$ only if there exists an external neighbor $w \in A_j \in A''$ of $u$ that was sampled in the candidate set of the node colored $\gamma(v)$ in $A_j$ (if any). Hence, by union bound, $\Exp[ Y_u ] \leq p \cdot 10^8\sqrt{\Delta} < 1/100$. Note that $Y_u$ and $Y_{u'}$ can be dependent for $u \neq u'\in S$ as $u$ and $u'$ may have the same external neighbor $w$ that is a candidate for a node colored $\gamma(v)$. However, the variables $\set{Y_u: u\in S}$ form a read-$k$ family for $k = \max_{w \in A_j \in A'' \setminus A_i } |N(w)\cap S|$, w.r.t.\ the independent boolean variables $Z_w$ equal to one iff $w$ was sampled in the candidate set of the node colored $\gamma(v)$ in its clique. Each variable $Z_w$ influences only the $Y_u$ for which the edge $\set{u,w}$ exists, which amounts to at most $k$ variables definition of $k$. Hence, the read-$k$ bound implies that
\[
\Prob*{ \sum_{u\in S} Y_u > |S|/20 } 
\leq \Prob*{ \abs*{ \sum_{u\in S} Y_u - \Exp\left[ \sum_{u\in S} Y_u \right] } > |S| / 30 }
\leq \exp\paren*{ - \Omega( |S| / k ) }
\]

\medskip\noindent
Overall, the number of unsafe candidates for $v$ in $S$ is bounded from above by the sum of $X_u + Y_u$ over $u\in S$. By union bound, the sum of the $X_u$ and the sum of the $Y_u$ are both smaller than $|S|/20$ w.p.\ at least $1 - 2\exp(- \Omega( |S| / k ) )$; hence the claim.
\end{proof}

\noindent
We are now ready to stick the analysis above together to prove \Cref{lem:subsampling}.

\paragraph{Pre-shattering.}
For each node $v\in \unhappy$, each node $u\in \swap_v$ samples itself into the set $S_v$ with probability $p = \Delta^{-23/40}$.
Denote by $\fS = \set{S_v : v\in \unhappy}$ the randomly generated collection of candidate sets.
Let us define a some bad events for every $A_i \in A'$:
\begin{enumerate}[leftmargin=1.5cm]
    \item[$B_1(i)$:] Some $v \in \unhappy_i$ has fewer than $\Delta^{17/40}/20$ safe candidates w.r.t. $\fS$;

    \item[$B_1'(i)$:] Some $w\notin A_j \in A' \setminus A_i$ has at least $\Delta^{1/10}$ neighbors in some $S_v$ where $v\in \unhappy_i$;

    \item[$B_2(i)$:] Some $u \in A_i \setminus\unhappy$ belongs to more than $\Delta^{17/40}/80$ sets in $\fS$;

    \item[$B_3(i)$:] For some uncolored $A_j \not\in A'$ adjacent to $A_i$ and color $x$, more than $\Delta^{37/40}$ vertices in $A_j$ have a neighbor outside of $A_j \cup \All_j \cup \Big_j^+$ that either has a candidate of color $x$ or is a candidate for a node of color $x$;
    
    \item[$B_4(i)$:] The set $\swap_v$ for some $v\in \unhappy_i$ contains $\Delta/20$ unsafe candidates w.r.t. $\fS$.
\end{enumerate}

\noindent
It should be clear that those events depend on the colors and random decisions within $O(1)$ distance. Let us now argue that they occur with probability at most $\exp(-\Delta^{1/40})$.

\begin{claim}
\label{claim:b1}
    Event $B_1(i)$ and $B_1'(i)$ occur with probability at most $\exp(-\Delta^{1/40})$.
\end{claim}
\begin{proof}
We argue about a given $v\in \unhappy_i$. The claim then follows by union bound over all such $v$.

We first argue that the probability on the randomness of $S_v$ that $S_v$ either (1) contains fewer than $\Delta^{17/40}/20$ nodes, or (2) has $\Delta^{1/10}$ edges to some $w\in A_j \in A' \setminus A_i$ is small. 
Fix a node $v\in A_i$, and recall that $|\swap_v|\geq \Delta/10$ (\Cref{lem:swappable}). Hence, the expected size of $S_v$ is $p|\swap_v|=\Delta^{17/40}/10$. By Chernoff, we obtain that $|S_v|$ has fewer than $\Delta^{17/40}/19$ w.p.\ at most $\exp(-\Omega(\Delta^{17/40}))$. Meanwhile, a node $w\notin A_j\in A' \setminus A_i$ has at most $10^8\sqrt{\Delta}$ neighbors in $A_i$ and thus $o(1)$ expected neighbors in $S_v$. By Chernoff, node $w$ has more than $\Delta^{1/10}$ neighbors in $S_v$ w.p.\ at most $\exp(-\Omega( \Delta^{1/10}))$. By union bound on all $v\in\unhappy_i$ and $w$ adjacent to $A_i$, we get that $\Prob{ B_1'(i) } \leq 10^{16}\Delta^{2} \exp\paren*{ -\Omega(\Delta^{1/10}) } \leq \exp\paren*{ - \Delta^{1/40} }$. 

To bound the probability of $B_1(i)$, it suffices to bound the probability of $B_1(i) \cap \overline{ B_1'(i) } \cap \overline{E(i)}$, where $E(i)$ is the event that some $S_v$ with $v\in \unhappy_i$ contains fewer than $\Delta^{17/40}/19$ nodes. Consider any given realization of the sets $S_v$ for $v\in \unhappy_i$ such that $\overline{ B_1'(i) } \cap \overline{E(i)}$ holds; then the random process is the one described in \Cref{claim:b1Detail} with $A'' = A' \setminus A$ and $S = S_v$ for some fixed $v\in \unhappy_i$. Since $S_v$ is such that $\overline{B_1'(v)} \cap \overline{E(i)}$ holds, it verifies the assumption of \Cref{claim:b1Detail} with $k = \Delta^{1/10}$, and thus $S_v$ contains at most $|S_v|/20$ unsafe candidates. Under $\overline{E(i)}$, it implies that $S_v$ contains at least $\Delta^{17/40}/20$ safe candidates. By union bound, the probability (on the randomness of $S_w$ for $w \notin \unhappy \setminus \unhappy_i$) that some $v$ has too few safe candidates is 
\[ |\unhappy_i| \exp\paren*{ - \Omega\paren*{ \frac{ \Delta^{17/40} }{ \Delta^{1/10 } } } } \le \exp(-\Omega( \Delta^{13/40} )) \ . \]
Since this holds for any realization of $\set{S_v: v\in \unhappy_i}$ where $\overline{B_1'(i)} \cap \overline{E(i)}$ holds, we get a bound on the probability of $B_1(v)$ as
$
\Prob{ B_1(i) } 
\leq \Prob{ B_1(i) \cap \overline{B_1'(i)} \cap \overline{E(i)} } + \Prob{ B_1'(i) } + \Prob{ E(i) }
\leq \exp( - \Delta^{1/40}  )$.
\end{proof}

\begin{claim}
\label{claim:b2}
    Event $B_2(i)$ occurs with probability at most $\exp(-\Delta^{1/40})$.
\end{claim}
\begin{proof}
    In expectation, a vertex is sampled into  $\Delta^{-23/40}|\unhappy_i|=o(1)$ sets of $\fS$. The claim follows via a Chernoff bound.
\end{proof}
\begin{claim}
\label{claim:b3}
    Event $B_3(i)$ occurs with probability at most $\exp(- \Omega(\Delta^{1/40}) )$
\end{claim}
\begin{proof}
    Fix a clique $A_j$ (either in $A'$ or uncolored) and a color $x$. For a fixed external neighbor $u\in A_\ell \neq A_j$, the vertex with color $x$ in $u$'s clique (if any) gets sampled into $u$'s candidate set $S_u$ or samples $u$ into its candidate set with probability at most $p \leq \Delta^{-23/40}$. On the other hand, by assumption in \cref{lem:coloringCliques}, the external degrees of uncolored cliques not in $A'$ is at most $30\Delta^{1/4}$ (recall that in this step, we do not need to guarantee that CC is maintained by cliques of $A'$, in particular, cliques of $A_H$ do not need to verify \cref{lem:subsampling}(3)). And since the sampling is independent, it is easy to verify that (P7.1) holds with $Q = 30\Delta^{1/4}$, thus that $B_3(i)$ occurs because of clique $A_j$ with probability at most $\exp(-\Delta^{1/40})$ by \cref{lem:bnd-notbig}. By union bound on the colors and the at most $10^8\Delta^{3/2}$ cliques adjacent to $A_i$, we obtain the desired probability bound.
\end{proof}

\begin{claim}
    Event $B_4(i)$ holds w.p. at most $\exp(-\Omega(\Delta^{1/2}))$.
\end{claim}
\begin{proof}
    Fix some $v \in \unhappy_i$. Recall that every $w\in A_j \in A' \setminus A_i$ has at most $k := 10^8\sqrt{ \Delta }$ external neighbors, thus at most $k$ neighbors in $\swap_v$.
    Hence \Cref{claim:b1Detail} with $S = \swap_v$ and $k$ as above implies that $\swap_v$ contains more than $|\swap_v|/20$ unsafe candidates w.r.t.\ sets $\set{ S_w : w \in \unhappy }$ with probability at most $\exp(-\Omega(|\swap_v|/k)) \leq \exp( -\Omega(\Delta^{1/2}))$, where the last inequality uses that $|\swap_v| \geq \Delta/10$ (\cref{lem:swappable}). The claim follows by union bound on all $v\in \unhappy_i$.
\end{proof}

\paragraph{Retractions \& Candidates $T^{pre}_v$.}
If $\Delta \geq (\log n)^{50}$, w.h.p., none of the bad events occur and we are done. We henceforth assume that $\Delta \leq (\log n)^{50}$. To find enough safe candidates even in cliques where bad events occur, we retract some sets of candidates.
Let $A'' \subseteq A'$ be the set of cliques $A_i\in A'$ for which
\begin{enumerate}
    \item some bad event $B_j(i)$ for $j\in[4]$ occurred, or
    \item $A_i$ is incident to some $A_{i'}$ for which a bad event $B_j(i')$ occurred.
\end{enumerate}
For all the cliques $A_i \in A' \setminus A''$ and $v\in\unhappy_i$, let $T^{pre}_v$ be the subset of $S_v$ containing all the safe candidates w.r.t.\ the sets $S_v$ sampled during pre-shattering.
Recall that removing candidates does not turn safe candidates into unsafe ones; hence the sets $T^{pre}_v$ contain only safe candidates w.r.t.\ $\fT^{pre} = \set{T_v^{pre}: v\in\unhappy_i, A_i \in A' \setminus A''}$. We say that the sets $S_v$ for $v\in A_i \in A''$ were retracted as we henceforth ignore them. 

\paragraph{Post-shattering.}
To compute candidates for sets of $A''$ we consider the LLL induced by the same random sampling process as before except that it runs only in $A''$: every vertex $u\in \swap_v$ for $v\in \unhappy_i$ with $A_i \in A''$ samples itself into a set $S_v'$ with probability $p = \Delta^{-23/40}$. We call the sets of candidates $S_v'$ rather than $S_v$ to emphasize that they were sampled during post-shattering. The random collection of candidate sets thereby produced is called $\fS'$. Consider the following bad events: for each $A_i \in A''$ or $A_i$ adjacent to some $A_j \in A''$,
\begin{enumerate}[leftmargin=1.5cm]
    \item[$P_1(i)$:] Some $v\in \unhappy_i$ has fewer than $\Delta^{17/40}/40$ safe candidates w.r.t. $\fS' \cup \fT^{pre}$;
\end{enumerate}
and for each $A_i \in A''$, let
\begin{enumerate}[leftmargin=1.5cm]
    \item[$P_2(i)$:] Some $u \in A_i \setminus \unhappy_i$ is a candidate for more than $\Delta^{17/40}/80$ nodes in $\unhappy_i$;

    \item[$P_3(i)$:] For some uncolored $A_j \not\in A'$ adjacent to $A_i$ and color $x$, more than $2 \cdot \Delta^{37/40}$ vertices in $A_j$ have a neighbor outside of $A_j \cup \All_j \cup \Big_j^+$ that either has a candidate of color $x$ or is a candidate for a node of color $x$.
\end{enumerate}

\noindent
We emphasize that the safety of the candidates is w.r.t.\ $\fT^{pre}$ obtained from pre-shattering and $\fS'$ sampled in post-shattering.
Let us first argue that $P_1(i)$, $P_2(i)$ and $P_3(i)$ are rare events even in the presence of the fixed pre-shattering sets $S_v$.

\begin{claim}
    For $A_i \in A''$ or $A_i$ adjacent to some $A_j \in A''$, the event $P_1(i)$ occurs w.p.\ at most $\exp(- \Omega( \Delta^{1/40}) )$.
\end{claim}

\begin{proof}
    Suppose first that $A_i \in A''$. If $B_4(i)$ holds, then $A_i$ and all the adjacent cliques belong to $A''$. Hence the same argument as in \Cref{claim:b1} (for $B_1(i)$) implies that $P_1(i)$ holds with probability at most $\exp(-\Omega(\Delta^{1/40}))$. Otherwise, suppose that $\overline{ B_4(i) }$ holds, hence for all $v\in\unhappy_i$, the sets $\swap_v$ contain at least $\Delta/20$ safe candidates w.r.t.\ $\fT^{pre}$. We follow the same analysis as for \Cref{claim:b1}, but we only have half as many vertices to pick from. For a fixed $v\in \unhappy_i$, w.p. at least $1 - \exp(-\Omega(\Delta^{17/40}))$, at least $\Omega(\Delta^{17/40})$ nodes $u\in \swap_v$ that are safe candidates w.r.t.\ $\fT^{pre}$ and get sampled into $S_v'$. Also, w.p. at least $1 - 10^8\Delta^{3/2}\exp(-\Omega( \Delta^{1/10} ))$, every $w \in A_j \in A'' \setminus A_i$ has fewer than $\Delta^{1/10}$ neighbors in $S_v'$. Now, \Cref{claim:b1Detail} with $S$ the set of safe candidates (w.r.t. $\fT^{pre}$) in $S_v'$ and $k = \Delta^{1/10}$ implies that w.p.\ $1 - \exp( - \Omega(\Delta^{17/40}/k) ) \geq 1 - \exp(-\Omega(\Delta^{1/40}))$. Overall, when $A_i \in A''$, the event $P_1(i)$ holds w.p. at most $\exp( -\Omega( \Delta^{1/40} ) )$.

    \medskip\noindent
    Suppose now that $A_i \notin A''$ but is adjacent to some $A_j\in A''$. Fix some $v\in \unhappy_i$. Since $A_i \notin A''$, the event $B_1'(i)$ does not hold, thus every $w$ in some $A_j \in A''$ has fewer than $\Delta^{1/10}$ neighbors in $T^{pre}_v$. By \Cref{claim:b1Detail} with $S = T_v^{pre}$, the set $T_v^{pre}$ contains at most $|T_v^{pre}|/20$ unsafe candidates w.r.t.\ $\fS'$ w.p. at least $1 - \exp(-\Omega(\Delta^{17/40}/\Delta^{1/10})) = 1 - \exp(-\Omega(\Delta^{13/40}))$. Since $T_v^{pre}$ contained at least $\Delta^{17/40}/20$ nodes (otherwise $B_1(i)$ would hold and $A_i\in A''$), the set $T_v^{pre}$ contains at least $(1 - 1/20)\Delta^{17/40}/20 \geq \Delta^{17/40}/40$ safe candidates w.r.t. sets $\fT^{pre} \cup \fS'$ with probability at least $1 - \exp(-\Omega( \Delta^{13/40} ))$. 
    The claim follows by union bound over all $v\in\unhappy_i$.
\end{proof}

The bound on the probability of $P_2(i)$ and $P_3(i)$ follows the same argument as in pre-shattering. Note that we allow the post-shattering phase to have a fresh budget.

\begin{claim}
    Let $A_i\in A''$. Then $P_2(i)$ and $P_3(i)$ each occur w.p.\ at most $\exp(-\Delta^{1/40})$.
\end{claim}

\begin{proof}
    See \Cref{claim:b2,claim:b3}.
\end{proof}

As for the synchronized color trial, we can see the post-shattering random process as having one random variable for each clique $A_i \in A''$.
For each event $P_1(i)$, $P_2(i)$, and $P_3(i)$, its set of variables corresponds to the clique $A_i$ and each neighboring clique of $A''$. Therefore the dependency degree of those events is at most $10^{16}\Delta^{3}$. Since each event occurs with probability at most $\exp(-\Omega(\Delta^{1/40}))$, it indeed defines an LLL that can be solved by \cref{thm:deterministicLLLLOCAL}.

The shattering argument showing that the connected components of the dependency graph have size at most $\poly(\log n)$ is identical to that of \cref{lem:sct}, so we do not repeat it here. It follows that computing the sets $\fS'$ such that none of the events $P_1(i)$, $P_2(i)$, or $P_3(i)$ occur takes $\Ot{ \log^3 N } = \Ot{ \log^3 \log n }$ rounds. 

The algorithm outputs the safe candidate system $\fT = \set{ T_v : v\in\unhappy }$ where $T_v$ is the set of safe candidates for $v$ in $T_v^{pre}$ w.r.t.\ $\fT^{pre} \cup \fS'$ if $v\in A_i \in A' \setminus A''$ (it succeeded in pre-shattering), and $T_v$ is the set of safe candidates for $v$ in $S'_v$ w.r.t. $\fT^{pre} \cup \fS'$ if $v\in A_i \in A''$ (it retracted its candidates after pre-shattering). This choice of candidate sets satisfies \cref{lem:subsampling}(3) because after retractions in the pre-shattering, none of the $B_3(i)$ holds, and after post-shattering, none of the $P_3(i)$ hold. \cref{lem:subsampling}(1,2) hold after retractions for every clique $A_i \notin A''$ because $B_1(i)$ and $B_2(i)$ do not occur and retracting candidates in neighboring cliques cannot make them occur. If $A_i$ has neighboring cliques in $A''$, the event $P_1(i)$ is included in post-shattering so \cref{lem:subsampling}(1) continues to hold after selecting sets $\fS'$. For $A_i \in A''$, \cref{lem:subsampling}(1,2) hold after post-shattering because neither of $P_1(i)$ nor $P_2(i)$ can occur.
\qedsymbol

 \section{Ultrafast Coloring High-Degree Graphs}
\label{sec:hideg}

In this section, we show that the coloring can be computed in $O(\log^* n)$ rounds when $\Delta = \Omega(\log^{50} n)$, rather than $O(\log \Delta \cdot \poly\log\log n)$ many rounds when using the iterative approach described in \cref{sec:colorwithmuchslack} \footnote{We have not tried to optimize the constant exponent "50"}.

Several parts of our algorithm actually run in $O(1)$ rounds when $\Delta$ is large (e.g., when $\Delta \geq \log^{50} n$) simply because the probability that each bad event from our LLLs occurs is $\poly(n)\exp(-\Omega(\Delta^{1/40}) \leq \poly(n)\exp(-\Omega(\log^{5/4} n)) \leq 1/\poly(n)$.
This applies to the coloring of the cliques (see \alg{ColorCliques} from \cref{sec:cliques}), slack generation (see \alg{SlackGeneration} from \cref{sec:slackGeneration}), and vertex splitting (\cref{lem:ds} from \cref{sec:sparse}).

The remaining hurdle is that the algorithm \alg{ColorWithMuchSlack} uses $O(\log \Delta)$ iterations.
In this section, we explain how it can be replaced by
an algorithm of \cite{HKNT22}
for the $(\deg+1)$-list-coloring problem. 
It runs in $O(\log^* n)$ rounds when $\Delta > (\log n)^3$. 
Recall that the input to \alg{ColorWithMuchSlack} is a $\pious$ subgraph $H$. 

\ThmHighDegree*

If it was not for the CC constraint, one could simply run the $(\deg+1)$-list-coloring algorithm of \cite{HKNT22} to extend the coloring to $H$. Given an instance of $(\deg+1)$-list-coloring, it is not possible in general to ensure that the CC constraint will be maintained. Here, we use that, thanks to \cref{prop:Pi}, the nodes of $H$ have lists significantly larger than their degree.

In more detail, we observe that all the coloring steps of \cite{HKNT22} are based on three forms of randomized color trials that we have already considered:
\begin{itemize}
    \item Random color trials (RCT), when nodes try a random color from their palette,
    \item Multi-color trials (MCT), when nodes try multiple colors from their palette, and
    \item Synchronized color trials (SCT), when the nodes of an almost-clique are assigned a permutation of the colors in a palette. 
\end{itemize}
Using \cref{prop:Pi} and \cref{lem:bnd-notbig}, we can show that each of the $O(\log^* n)$ steps of \cite{HKNT22} maintain the CC property with high probability.

\begin{proof}[Proof of \cref{thm:d1LC}]
One technical issue needs to be resolved before applying the algorithm of \cite{HKNT22}. We have lower bounds on the palettes of nodes but not explicit lower bounds on their \emph{degrees}. 
To address this, we construct the graph $H'$ from $H$ by adding $|L(v)| - (\deg_{H}(v)+1)$ dummy vertices incident to each $v$ each with two arbitrary colors. 
Note that every vertex of $H$ has degree at least $U \cdot \Delta^{0.22}$ in $H'$ (by \cref{prop:Pi}(b)).

The first step of \cite{HKNT22} computes (deterministically) a vertex partition of $H'$ into sets $V_{sp}, C_1, C_2, \ldots, C_q$ with the properties described in \cite[Lemma 4.2]{AA20}. The algorithm then colors $H'$ in two steps (see Algorithm 8 in \cite{HKNT22}): first color $H'[V_{sp}]$ with \cite[Algorithm 4]{HKNT22}, and the color $H'[C_1 \cup \ldots \cup C_q]$ with \cite[Algorithm 5]{HKNT22}.

Both algorithms begin with a random color trial --- for the sparse vertices, only some vertices try a random color. See \cite[Algorithm 3]{HKNT22}.
Since every vertex of $H$ samples one color out of a list of $U\Delta^{0.22}$ while the external degree of uncolored cliques is upper bounded by $U$ (by \cref{prop:Pi}), \cref{lem:bnd-notbig} and union bound implies that CC is not maintained for some clique with probability at most $n\Delta\exp(-\Delta^{1/40}) < 1/\poly(n)$. Note that the dummy nodes do not verify \cref{prop:Pi}, as their lists are of size two, but they are not adjacent to the cliques of $F$.

To extend the coloring to $H[V_{sp}]$, \cite{HKNT22} uses an algorithm called \alg{SlackColor} consisting of $O(\log^* n)$ iterations of of MCT. See \cite[Lemma 1]{HKNT22} and \cite[Algorithm 10]{HKNT22} with $s_{\min} = U\Delta^{0.22} \geq \Delta^{0.47}$ and $\kappa = 1/2$.
More precisely, each vertex of $V_{sp}$ (not dummy vertices) picks up to $\log n \le \Delta^{1/50}$ colors u.a.r.\ from their list. By \cref{prop:Pi}, all palettes contain at least $U\Delta^{0.22}$ colors. Hence, any given color is picked with probability proportional to at most 
\[ 
\frac{\log n}{U\Delta^{0.22}}
\leq \frac{1}{U\Delta^{0.22 - 1/50}}
\leq \frac{1}{U\Delta^{1/5}}\ . \] 
Hence, we can apply \cref{lem:bnd-notbig} as before to reason that each step of MCT maintains CC with probability at least $1 - \exp( - \Omega(\sqrt{s_{\min}}) ) - \Delta\exp(-\Omega(s_{\min})) \geq 1 - 1/\poly(n)$.

The analysis of slack color for nodes of $H[C_1 \cup \ldots \cup C_q]$ is the same as above. We emphasize that we do not need the put-aside set, Step 3 and 7 in \cite[Algorithm 5]{HKNT22}, because \cref{prop:Pi} ensures that all nodes of $H$ have more than $\log^2 n$ colors in their lists. After slack generation, \cite[Algorithm 5]{HKNT22} runs a synchronized color trial before it calls \alg{SlackColor} on (most) nodes. See \cite[Algorithm 7]{HKNT22}. In this algorithm, (most of) the vertices in each $C_i$ receive a random color from a selected leader. As such, the probability that any given color is tried by a vertex in some $C_i$ is $O( 1/|C_i| )$. Since the vertices have $\deg(v) = |L(v)| - 1 \geq U\Delta^{0.22}$, where the equality holds because of the dummy vertices and the inequality because of \cref{prop:Pi}, an almost-clique $C_i$ contains at least $\Omega(U\Delta^{0.22})$ vertices. And thus the CC constraint is maintained by the SCT as every vertex tries gets some color $x$ with probability at most $O(1/|C_i|) \leq O(\frac{1}{U\Delta^{0.22}}) \leq \frac{1}{ U\Delta^{1/5} }$.
\end{proof}

\printbibliography

\appendix

\section{Graph Decomposition by Molloy-Reed}
\label{app:decomposition}
Based on Molloy and Reed \cite{MR14}, Bamas and Esperet give the following structural decomposition of the graph $F$. 
\begin{lemma}[Lemma 4.5 of \cite{BamasE19}, based on Lemma 12 in~\cite{MR14}]\label{lem:12}
    \label{lem:MR-structural-decomposition}
We can construct from $H$ in $O(1)$ rounds of \LOCAL a graph $F$ of maximum degree at most $10^9\Delta$ (such
that a $c$-coloring of $H$ can be deduced from any $c$-coloring of $F$
in $O(1)$ rounds) and find a partition of the vertices of $F$ into $S, B,
A_1,\ldots, A_t$ such that: 
\begin{itemize}
\item[(a)] Every $A_i$ is a clique with
  $c-10^8\sqrt{\Delta}\le|A_i|\le c$. 
\item[(b)] Every vertex of $A_i$ has at most $10^8\sqrt{\Delta}$ neighbors in $F - A_i$.
\item[(c)] There is a set $\mathrm{All}_i \subseteq B$ of $c - |A_i|$ vertices which are adjacent to all of $A_i$. Every
other vertex of $F - A_i$ is adjacent to at most $\tfrac34\Delta + 10^8\sqrt{\Delta}$ vertices of $A_i$. 
  \item[(d)] Every vertex of $S$ either has fewer than $\Delta-3\sqrt{\Delta}$ neighbors in $S$ or has at least $900\Delta^{3/2}$ non-adjacent pairs of neighbors within $S$.
\item[(e)] Every vertex of $B$ has fewer than $c-\sqrt{\Delta}+9$
  neighbors in $F - \bigcup_j A_j$.
\item[(f)] If a vertex $v\in B$ has at least $c-\Delta^{3/4}$ neighbors
  in $F-\bigcup_j  A_j$, then there is some $i$ such that: $v$ has at most $c- \sqrt{\Delta}+9$ neighbors in $F -A_i$ and every vertex of $A_i$ has at most $30\Delta^{1/4}$ neighbors in $F - A_i$.
\item[(g)] For every $A_i$, every two vertices outside of $A_i \cup \mathrm{All}_i$ which have at least $2\Delta^{9/10}$ neighbors in $A_i$ are joined by an edge of $F$. 
\end{itemize}
\end{lemma}

Given \Cref{lem:MR-structural-decomposition}, the proof of \Cref{lem:structuralDecomposition} is relatively straightforward. 

\lemStructDecomposition*
\begin{proof}
\Cref{lem:12} decomposes the graph $F$ into $S$,$B$, a collection of cliques $A=\bigcup_j A_j$, and mentions a set $\mathrm{All}_i\subseteq B$ for each clique $A_i$. The more fine-grained decomposition of this lemma follows the algorithm of \cite{MR14}. Let $B_H\subseteq B$ be the vertices with at most $c-\Delta^{3/4}$ neighbors in $F\setminus \bigcup_j A_j$. The set $B_L$ equals $B\setminus B_H$. $A_H$ contains all cliques $A_i\in A$ such that each vertex of $A_i$ has at least $30\Delta^{1/4}$ neighbors outside of $A_i\cup \mathrm{All}_i$. $\AL$ equals $A\setminus A_H$.  We now prove each of the properties separately. 

\begin{enumerate}
    \item $S$: The decomposition of the graph $F$ in \Cref{lem:12} leaves a crucial property open, namely that the vertices in $S$, the sparse vertices, have a degree of at most $\Delta$ in their induced subgraph. Property \Cref{lem:12}(d) only specifies that a vertex in $S$ has a degree of at most $\Delta-3\sqrt{\Delta}$ or has at least $900\Delta^{3/2}$ non-edges in its neighborhood. But for the upcoming coloring step of coloring the vertices in $S$ it is crucial that their degree is also bounded above by $\Delta$ in any case. However, in their proof of \Cref{lem:12} in \cite{Reed98} the authors clearly state that this property also applies. 

    We next reason why we can claim $9.9\cdot10^5\cdot \sqrt{\Delta}$ non-edges in the neighborhood of vertices in $S$, instead of only $900\sqrt{\Delta}$ as claimed in \Cref{lem:MR-structural-decomposition}. The structural decomposition is defined in Lemma 12 of \cite{MR14}. 
It begins with the decomposition from \cite{MolloyReedBook} where nodes of $S$ are $d$-sparse for $d = 10^6\sqrt{\Delta}$, meaning that their neighborhood has at most $\binom{\Delta}{2} - d\Delta$ edges. As they claim in the proof of Lemma 12(d), the only part where they add edges in $F[S]$ that did not exist in $G$ is in Modification 2 (see Section 5.3, page 161). Importantly, this step adds edges only between big-neighbors. In the proof of Lemma 12 (section 5.3, page 163), Molloy and Reed argue that $v\in S$ with more than $\Delta - 3\sqrt{\Delta}$ neighbors in $S$ is not in any $Big_i$ set.

    Now, the argument is two-fold. If $u\in N(v) \cap S$ belongs to some $Big_i$, Molloy and Reed reason that $u$ is incident to at least $\Delta^{9/10}/2$ non-edges in $F[N(v) \cap S]$. Hence, if $v$ has at least $\alpha\Delta^{6/10}$ big neighbors, it also has at least $(\alpha/2)\Delta^{3/2}$ anti-edges in $F[N(v) \cap S]$, where $\alpha = 2 \cdot 10^6$. On the other hand, if it has fewer than $\alpha\Delta^{6/10}$ big neighbors, its induced neighborhood gained at most $(\alpha\Delta^{6/10})^2 \leq \alpha^2 \Delta^{12/10}=o(\Delta^{3/2})$ edges. As such, it has sparsity at least $d - o(\Delta^{1/2}) \geq 9.9 \cdot 10^5$.
    
    \item By definition, vertices in $B_H$ have at most $c-\Delta^{3/4}$ neighbors in $F\setminus \bigcup_j A_j$. Hence, for each $v\in B_H$ at most $x(v)=c-\Delta^{3/4}-deg_{B_H}(v)$ neighbors are in $S$ and already colored. Thus, the list of $v$ when processing $B_H$ is at least of size $|L(v)|\geq c-x(v)=deg_{B_H}+\Delta^{3/4}$.
    \item Follows from \Cref{lem:12} (b). 
\item For each $v\in B_L$, Lemma 12(f) guarantees that there is some $A_i\in A$ such that $v$ has at most $c-\sqrt{\Delta}+9$ neighbors outside of $A_i$. Lemma 12(f) even states that $A_i\in \AL$. In other words $v$ has at most $x(v)=c-\sqrt{\Delta}+9-deg_{B_L}(v)$ nodes outside of $A_i\in \AL$ and $B_L$. As all nodes in $\AL$ and $B_L$ are uncolored, we obtain $|L(v)| \geq c-x(v)\geq \deg_{B_L}(v) +\sqrt{\Delta}-9$.
\item The upper bound on the external degree follows from the definition of $\AL$.
\end{enumerate}
Additionally, we prove the second list of properties:
\begin{enumerate}[label=\alph*)]
\item This is identical to \Cref{lem:12} (a).
\item This follows from \Cref{lem:12} (c).
\item The set $\Big_i^+$ forms a clique due to \Cref{lem:12} (g). The upper bound of the degree into $A_i$ from a vertex in $\Big_i^+$ follows from \Cref{lem:12} (c). 
\end{enumerate}
\end{proof}

\section{Concentration Inequalities}

\begin{theorem}[Janson's inequality] 
    \label{thm:janson}
    Let $S\subseteq V$ be a random subset formed by sampling each $v \in V$ independently with probability $p$. Let $\mathcal{A}$ be any collection of subsets of $V$. For each $A \in \mathcal{A}$, let $I_A := \I{ A \subseteq S }$ be an indicator variable for the event that $A$ is contained in $S$. Let $f := \sum_{A \in \mathcal{A}} I_A$ and $\mu := \E[f]$. Define
    \begin{equation}
        \label{eq:jansonK}
        K := \frac{1}{2}\sum_{A, B \in \mathcal{A}, A \cap B \neq \emptyset} \E[I_A I_B]
    \end{equation}
    Then, for any $0 \le t \le \E[f]$, 
    $$
        \Prob{ f \le \E[f] - t } \le \exp \left(-\frac{t^2}{2 \mu + K}\right)
    $$
\end{theorem}

We use the classic Chernoff Bound on sums of independent random variables. See \cite{Doerr2020} or \cite{DP09} for details.
\begin{proposition}[{Chernoff bounds}]\label{lem:chernoff}
Let $X_1, \ldots, X_n$ be a family of independent random variables with values in $[0,1]$, and let $X = \sum_{i \in [n]} X_i$.
Suppose $\mu_{\mathsf{L}} \leq \Exp[X] \leq \mu_{\mathsf{H}}$, then
\begin{align}
  \text{for all } \delta \geq 0, \text{ we have that }
  &&
  \Prob*{ X > (1+\delta) \mu_{\mathsf{H}} }
& \le \exp\paren*{-\frac{\delta^2}{2 + \delta}\mu_{\mathsf{H}}}\ .
  \label{chernoff-mult-upperbound}
\\
\text{for all } \delta \in(0,1), \text{ we have that }
  &&
  \Prob*{ X < (1-\delta) \mu_{\mathsf{L}} }
& \le \exp\paren*{-\frac{\delta^2}{2}\mu_{\mathsf{L}}}\ .
  \label{chernoff-mult-lowerbound}
\end{align}
\end{proposition}

Let $Y_1, \ldots, Y_n$ be boolean random variables (with values in $\set{0,1}$). We say they form a \emph{read-$k$ family} if they can be expressed as a function of independent random variables $X_1, \ldots, X_m$ such that each $X_j$ influences at most $k$ variables $Y_i$. More formally, there are sets $P_i \subseteq [m]$ for each $i \in [n]$ such that: (1) for each $i\in[n]$, the variable $Y_i$ is a function of $\set{X_j: j \in P_i}$ and (2) $|\set{i: j\in P_i}| \le k$ for each $j \in [m]$.

\begin{proposition}[read-$k$ bound, \cite{GLSS15}]
\label{lem:k-read}
Let $Y_1, \ldots, Y_n$ be a read-$k$ family of boolean variables, and let $Y$ be their sum.
Then, for any $\delta > 0$,
\[
  \Prob{ \abs*{Y - \Exp{Y}} > \delta n}
  \le 2\exp\paren*{-\frac{2\delta^2 n}{k}} \ .
\]
\end{proposition}

\begin{lemma}[Azuma-Hoeffding]\label{lem:BEPS}
    Let $Z = Z_1 + \ldots + Z_n$ be the sum of n random variables and $X_0,\ldots, X_n$ be a sequence, where $Z_i$ is uniquely determined by $X_0,\ldots, X_i$. Let $\mu_i = \Exp[Z_i \given X_0,\ldots, X_{i-1}]$, $\mu = \sum_i \mu_i$, and $a_i \leq Z_i \leq a'_i$. Then
    $$
    \Prob{ {Z \geq \mu + t} }, \Prob{ Z \leq \mu - t } \leq \exp{\left(-\frac{t^2}{2\sum_i(a'_i - a_i)^2}\right)}.
    $$
\end{lemma}

We use the following Talagrand inequality. A function $f(x_1, \dots, x_n)$ is called $c$-\textit{Lipschitz} iff the value of any single $x_i$ affects $f$ by at most $c$. Additionally, $f$ is $r$-\textit{certifiable} if for every $x = (x_1, \dots, x_n)$, (1) there exists a set of indices $J(x) \subseteq [n]$ such that $|J(x)| \le r \cdot f(x)$, and (2) if $x'$ agrees with $x$ on the coordinates in $J(x)$, then $f(x') \ge f(x)$. 
\begin{lemma}[Talagrand's Inequality \cite{MR14}]
    \label{lem:talagrand}
    Let $X_1, \dots, X_n$ be independent random variables and $f(X_1, \dots, X_n)$ be a $c$-Lipschitz, $r$-certifiable function. For any $b \ge 1$:
    $$\Prob{ |f-\E[f]| > b + 60c\sqrt{r \E[f]} } \le 4 \exp{\left(-\frac{b^2}{8c^2r\E[f]}\right)}$$
\end{lemma}

McDiarmid \cite{McDiarmid1989} extended Talagrand's inequality to the setting where $X$ depends also on permutations.

\begin{lemma}[McDiarmid's Inequality {\cite[Section 3.1]{MR14}}]
\label{lem:mcdiarmid}
Let $X$ be a non-negative random variable determined by independent trials $T_1, \ldots, T_n$ and independent permutations $\Pi_1, \ldots, \Pi_m$. Suppose that for every set of possible outcomes of the trials and permutations, we have:
\begin{enumerate}
    \item changing the outcome of any of the trials affects the outcome by at most $c$,
    \item interchanging two elements in any one permutation can affect $X$ by at most $c$, and
    \item for each $s \geq 0$, if $X \geq s$ then there exists a set of at most $rs$ choices whose outcomes certify that $X \geq s$.
\end{enumerate}
Then for any $t > 0$, we have
\[
\Prob{ \abs*{ |X - \Exp{X}| } > t }
\leq 4\exp\paren*{ - \frac{t^2}{128c^2 r ( \Exp{X} + t )} }
\]
\end{lemma}

\section{Slack Generation}
\label{app:slackgen}

\lemSlackGeneration*
\begin{proof}
    Let $N_A(v) = N(v) \cap A$ for short. 
    Let $X \subseteq \binom{N_A(v)}{2}$ be the set of non-edges, where $|X| = \overline{m}$. Let $\chi(w)$ be the random color chosen by $w$, and let $\varphi(w)$ be the possible permanent color, or $\varphi(w)=\bot$ in case of a conflict. 
    
    Let $Z$ be the number of colors $\psi \in [q]$ such that there exists a non-edge $\{u,w\} \in X$ with $\chi(u)=\chi(w)=\psi$, and \emph{for all} such non-edges, $\varphi(u)=\varphi(w)=\psi$, i.e. all the nodes retain the color (see below for a formal definition with quantifiers). Say that a non-edge $\{u,w\} \in X$ is \textit{successful} if $\varphi(u)=\varphi(w) \neq \bot$, and no node in $N_A(v) \setminus \{u,w\}$ picks the same color.
    Let $Y_{\{u,w\}}$ be an indicator function for the event that $\{u,w\}$ is successful. We have $\E[Z] \ge \sum_{\{u,w\} \in X} \Exp[Y_{\{u,w\}}]$, since each non-edge with $Y_{\{u,w\}}=1$ counts towards $Z$. The probability of a non-edge being successful is at least 
    \begin{equation*}
        \Prob{ Y_{\{u,w\}} = 1 }
        \ge \frac{1}{q} \paren*{ \frac{q - 1}{q} }^{2\Delta - 2} \paren*{ \frac{q - 1}{q} }^{\Delta - 2} 
        \ge \frac{1}{q} \left(1 - \frac{1}{q}\right)^{3\Delta} 
    \end{equation*}
    where $u$ and $w$ select the same color w.p. $1/q$; with probability $(\frac{q - 1}{q})^{2\Delta - 2}$ none of the neighbors of $u$ nor $v$ choose that particular color, and $(\frac{q - 1}{q})^{\Delta - 2}$ is the probability that no other neighbors of $v$ choose that color. Using that $q \geq \Delta/3$ and that $\Delta \geq 30$, we can lower bound this probability as 
    \begin{equation*}
        \Prob{ Y_{\set{u,w}} = 1}
        \geq \frac{ 1 }{q} \exp\paren*{ -\frac{3\Delta}{q - 1} } 
        \geq \frac{e^{-10}}{q}
        \tag{using that $1 - x \geq \exp\paren*{ -\frac{1}{x - 1}}$ for all $x > 1$}
    \end{equation*}
    which gives
    \[
    \E[Z] \ge \sum_{\{u,w\} \in X} \Exp[ Y_{\{u,w\}} ] \ge e^{-10} \cdot \overline{m}/ q \ .
    \]

    Next, we show that $Z$ is concentrated around its mean. 
    Let $T$ be the number of colors that are randomly chosen in \trycolor by both nodes of at least one non-edge in $X$. Let $D$ be the number of colors that are chosen in \trycolor by both nodes of at least one non-edge in $X$, but are not retained by at least one of them. Formally, 
    \begin{alignat*}{2}
        T&:= \# \text{ colors } \psi \text{ s.t. } &&\exists \{u,w\} \in X: \chi(u)=\chi(w)=\psi \\
        D &:= \# \text{ colors } \psi \text{ s.t. } &&\big(\exists \{u,w\} \in X: \chi(u)=\chi(w)=\psi \big) \;\wedge \\
        & &&\big(\exists \{u,w\} \in X: (\chi(u)=\chi(w)=\psi) \wedge (\varphi(u) = \bot  \vee  \varphi(w) = \bot)\big) \\
        Z&:= \# \text{ colors } \psi \text{ s.t. } &&\big(\exists \{u,w\} \in X: \chi(u)=\chi(w)=\psi \big) \;\wedge \\
        & &&\big(\forall \{u,w\} \in X: (\chi(u)=\chi(w)=\psi) \implies (\varphi(u)=\varphi(w)=\psi)\big)
    \end{alignat*}
    We have $Z=T-D$, since $D$ counts the colors where the implication in the definition of $Z$ fails. 

    We upper bound $\Exp[T]$ (which implies the same bound for $D$, as $D \le T$). For a fixed color $c$, the probability that both $u,w$ pick $c$ is at most $1/q^2$. By union bound, the probability that $c$ is picked by at least one non-edge is at most $\overline{m} / q^2$. There are $q$ colors, so $\Exp[T] \le q \cdot \overline{m} / q^2 = \overline{m} / q$. 

    The functions $T$ and $D$ are $r$-certifiable with $r=2$ and $r=3$, respectively. See the appendix and \Cref{lem:talagrand} for the definition of an $r$-certifiable function. $T$ and $D$ are both $2$-Lipschitz: whether a node is activated, and which color it picks affects the outcome by at most $c=2$. 
    We apply \Cref{lem:talagrand} with $b = \Exp[Z] / 10 - 60c\sqrt{r \cdot \Exp[T]}$ (which is non-negative for $\overline{m}/q \geq c_0$ and $c_0$ a large enough constant):
    \begin{align*}
        \Prob{ |T-\Exp[T]| \ge b + 60 c \sqrt{r \cdot \Exp[T]} }
        &= \Prob{ |T-\Exp[T]| \ge \Exp[Z] / 10 } \\
        &\le 4\exp\left(-\frac{\big(\Exp[Z] / 10 - 60c\sqrt{r \cdot \Exp[T]}\big)^2}{8c^2 r \Exp[T]}\right) \\
        &\le \exp\left(-\Theta(1) \left(\frac{\Exp[Z]^2}{\Exp[T]} - \frac{\Exp[Z]}{\sqrt{\E [T]}} + O(1)\right)\right) \\
        &\le \exp\left(-\Omega(\overline{m}/q)\right)
    \end{align*}
    In the last inequality, we used that $\E[Z] \ge e^{-10} \overline{m}  / q$ and $\Exp[T] \le \overline{m} / q$. The same concentration bound applies for $D$, meaning that $\Prob{ |D-\Exp[D]| \ge \E[Z] / 10 } \le \exp\left(-\Omega(\overline{m}/q)\right)$. By union bound, neither of the events $|T - \Exp[T]| \ge \E[Z]/10$ and $|D - \Exp[D]| \ge \E[Z]/10$ occur with probability at least $1 - 2\exp\left(-\Omega(\overline{m}/q)\right)$. 
    Hence, 
    \[
    Z = T-D 
    \ge \E [T] - \E[Z]/10 - (\E [D] + \E [Z] / 10) 
    = (4/5) \cdot \E [Z] \ge \frac{ \overline{m} / q }{ 3 \cdot 10^{4} } \ ,
    \]
    with probability at least $1-\exp\left(-\Omega(\overline{m}/q)\right)$.
\end{proof}

\PropSparsityPreserving*
\begin{proof}
    Let $\overline{E}$ be the set of non-edges in $H[N(v)]$. For each non-edge $e \in \overline{E}$, define an indicator variable $I_e$ for the event that the non-edge $e$ is preserved in $H[N(v) \cap A]$. Call $Z$ the number of anti-edges in $H[N(v) \cap A]$, i.e., we have $\E[Z] = \sum_{e \in \overline{E}} \E[I_e] = \overline{m} p^2$. We can bound the quantity \cref{eq:jansonK} as  
    \begin{align*}
        K = \frac{1}{2}\sum_{e,e' \in \overline{E}, e \cap e' \neq \emptyset} \E[I_e I_{e'}] \le \frac{1}{2} \overline{m} (2 \Delta - 2) p^3 \le \overline{m} \Delta p^3 \ ,
    \end{align*}
    where the first inequality uses that endpoints of $e$ have maximum degree $\Delta$.
    Using \Cref{thm:janson} with $t=\E[Z]/2$, we can compute $\Prob{ Z \le p^2\overline{m} / 2 } = \Prob{ Z \le \E[f] - t }$, where 
    \begin{align*}
    \Prob{ Z \le \E[f] - t }
        \le \exp \left(-\frac{t^2}{2 \E[f] + K} \right) 
        &\le \exp \left(-\frac{p^4 \overline{m}^2 / 4}{2p^2 \overline{m} + p^3\overline{m} \Delta } \right)  \\
        &= \exp \left(-\frac{p^2 \overline{m}}{8 + 4p \Delta } \right) 
        \le \exp \left(-\frac{p \overline{m}}{5\Delta } \right) 
    \end{align*}
    using that $p\Delta\ge 8$ by assumption.
  \end{proof}
\end{document}